\documentclass[a4paper,twocolumn,accepted=2025-10-20]{quantumarticle}
\pdfoutput=1

\usepackage[numbers,sort&compress]{natbib}

\PassOptionsToPackage{
	pdfencoding=auto,
	pdfnewwindow=true,
	pdfusetitle=false,
	bookmarks=true,
	bookmarksnumbered=true,
	bookmarksopen=true,
	pdfpagemode=UseThumbs,
	bookmarksopenlevel=1,
	pdfpagelabels=false,
	breaklinks=true,  
	backref=false, 
	colorlinks=true,
}{hyperref}

\usepackage{color}
\usepackage{CJKutf8}
\usepackage[T1]{fontenc}
\usepackage{amsmath}
\usepackage{amssymb}
\usepackage{stmaryrd}
\usepackage{graphicx}
\graphicspath{ {./figures/} }
\usepackage{romannum}
\usepackage{physics}
\usepackage[caption=false]{subfig}
\usepackage{comment}
\usepackage{mathtools}
\usepackage{import}
\usepackage{victor}
\usepackage{chris}
\usepackage{tensors}

\usepackage{thmtools}
\usepackage{algorithm}
\usepackage[noend]{algpseudocode}
\usepackage{ulem}
\usepackage{orcidlink}
\usepackage{suffix}

\definecolor{LightPink}{rgb}{0.858, 0.188, 0.478}
\definecolor{RED4}{rgb}{0.55,0,0}
\usepackage[unicode=true,bookmarks=true,bookmarksnumbered=false,bookmarksopen=false,breaklinks=false,backref=false,colorlinks=true]{hyperref}
\makeatletter

\@ifundefined{textcolor}{}
{
 \definecolor{BLACK}{gray}{0}
 \definecolor{WHITE}{gray}{1}
 \definecolor{RED}{rgb}{1,0,0}
 \definecolor{GREEN}{rgb}{0,1,0}
 \definecolor{BLUE}{rgb}{0,0,1}
 \definecolor{CYAN}{cmyk}{1,0,0,0}
 \definecolor{MAGENTA}{cmyk}{0,1,0,0}
 \definecolor{YELLOW}{cmyk}{0,0,1,0}
 \definecolor{navyblue}{rgb}{0.0, 0.0, 0.5}
}

\newcommand{\cL}{\mathcal{L}}
\newcommand{\cS}{\mathcal{S}}
\newcommand{\ra}{\rangle}
\newcommand{\la}{\langle}
\DeclareMathOperator{\poly}{poly}

\newtheorem{theorem}{Theorem}

\newcommand{\renyi}{R\'enyi}
\newcommand{\IId }{\mathbbm{1}}

\pgfplotsset{compat=1.18}

\usepackage{amsfonts}\usepackage{tabularx}\usepackage{dcolumn}\usepackage{bm}\usepackage{graphicx}\usepackage{epstopdf}
\definecolor{quantumviolet}{RGB}{112, 48, 160}
\usepackage{tensor}
\usepackage{braket}

\newcommand{\ranglee}{\rangle\!\rangle}
\newcommand{\langlee}{\langle\!\langle}
\usepackage[capitalise,compress]{cleveref}

\makeatother

\makeatletter
\newsavebox{\@brx}
\newcommand{\llangle}[1][]{\savebox{\@brx}{\m@th{#1\langle}}%
  \mathopen{\copy\@brx\kern-0.5\wd\@brx\usebox{\@brx}}}
\newcommand{\rrangle}[1][]{\savebox{\@brx}{\m@th{#1\rangle}}%
  \mathclose{\copy\@brx\kern-0.5\wd\@brx\usebox{\@brx}}}
\makeatother



\begin{document}
\begin{CJK}{UTF8}{gbsn}
\pagenumbering{arabic}

\hypersetup{
  colorlinks=true,
  linkcolor=quantumviolet,
  citecolor=quantumviolet,
  urlcolor=quantumviolet
}

\title{Instability of steady-state mixed-state symmetry-protected topological order to strong-to-weak spontaneous symmetry breaking}

\author{Jeet Shah\,\orcidlink{0000-0001-5873-8129}}
  \affiliation{Joint Quantum Institute, NIST/University of Maryland, College Park, MD, 20742, USA}\affiliation{Joint Center for Quantum Information and Computer Science, NIST/University of Maryland, College Park, MD, 20742, USA}

\author{Christopher Fechisin\,\orcidlink{0000-0002-2590-8221}}
  \affiliation{Joint Quantum Institute, NIST/University of Maryland, College Park, MD, 20742, USA}\affiliation{Joint Center for Quantum Information and Computer Science, NIST/University of Maryland, College Park, MD, 20742, USA}

\author{Yu-Xin Wang (王语馨)\,\orcidlink{0000-0003-2848-1216}}
  \affiliation{Joint Center for Quantum Information and Computer Science, NIST/University of Maryland, College Park, MD, 20742, USA}

\author{Joseph~T.~Iosue\,\orcidlink{0000-0003-3383-1946}}
  \affiliation{Joint Quantum Institute, NIST/University of Maryland, College Park, MD, 20742, USA}\affiliation{Joint Center for Quantum Information and Computer Science, NIST/University of Maryland, College Park, MD, 20742, USA}

\author{James~D.~Watson\,\orcidlink{0000-0002-6077-4898}}
  \affiliation{Joint Center for Quantum Information and Computer Science, NIST/University of Maryland, College Park, MD, 20742, USA}
    \affiliation{Department of Computer Science and Institute for Advanced Computer Studies, University of Maryland, College Park, MD 20742, USA}

\author{Yan-Qi Wang\,\orcidlink{0000-0003-3900-0836}}
  \affiliation{Joint Quantum Institute, NIST/University of Maryland, College Park, MD, 20742, USA}

\author{Brayden Ware\,\orcidlink{0000-0002-3321-3198}}
  \affiliation{Joint Center for Quantum Information and Computer Science, NIST/University of Maryland, College Park, MD, 20742, USA}\affiliation{Google Quantum AI, California, USA}

\author{Alexey~V.~Gorshkov\,\orcidlink{0000-0003-0509-3421}}
  \affiliation{Joint Quantum Institute, NIST/University of Maryland, College Park, MD, 20742, USA}\affiliation{Joint Center for Quantum Information and Computer Science, NIST/University of Maryland, College Park, MD, 20742, USA}

\author{Cheng-Ju Lin\,\orcidlink{0000-0001-7898-0211}}
  \affiliation{Joint Quantum Institute, NIST/University of Maryland, College Park, MD, 20742, USA}\affiliation{Joint Center for Quantum Information and Computer Science, NIST/University of Maryland, College Park, MD, 20742, USA}

% \date{\today}
\twocolumn[
\begin{@twocolumnfalse}
\begin{abstract}
Recent experimental progress in controlling open quantum systems enables the pursuit of mixed-state nonequilibrium quantum phases.  
We investigate whether open quantum systems hosting mixed-state symmetry-protected topological states as steady states retain this property under symmetric perturbations. 
Focusing on the \textit{decohered cluster state}---a mixed-state symmetry-protected topological state protected by a combined strong and weak symmetry---we construct a parent Lindbladian that hosts it as a steady state. This Lindbladian can be mapped onto exactly solvable reaction-diffusion dynamics, even in the presence of certain perturbations, allowing us to solve the parent Lindbladian in detail and reveal previously-unknown steady states. Using both analytical and numerical methods, we find that typical symmetric perturbations cause strong-to-weak spontaneous symmetry breaking at arbitrarily small perturbations, destabilize the steady-state mixed-state symmetry-protected topological order.  However, when perturbations introduce only weak symmetry defects, the steady-state mixed-state symmetry-protected topological order remains stable. Additionally, we construct a quantum channel which replicates the essential physics of the Lindbladian and can be efficiently simulated using only Clifford gates, Pauli measurements, and feedback.
\end{abstract}
\end{@twocolumnfalse}
]

{\begingroup
		\hypersetup{urlcolor=quantumviolet}
\maketitle
		\endgroup}
\end{CJK}
\hypersetup{linkcolor=quantumviolet, urlcolor=quantumviolet, citecolor=quantumviolet, unicode=true} 

\setcounter{tocdepth}{2}
\tableofcontents

\section{Introduction}
\label{sec:intro}

Dissipation is typically viewed as detrimental to quantum information, but recent advances in experimental control have begun to harness its potential to create well-controlled nontrivial open quantum systems~\cite{ma2019dissipative,gertler2021protecting,harringtonEngineered2022,brown2022trade,cole2022resource,malinowski2022generation,vanmourik2024experimental}. 
This raises an intriguing question: can an open quantum system host exotic quantum phases that are inherently mixed and nonthermal~\cite{coser2019classification}? 
Namely, are there phases that cannot be understood within the typical equilibrium framework, where the phases are described by either a pure state or a Gibbs thermal state?
This question becomes particularly compelling if the mixed-state phase leads to a degenerate steady-state manifold, offering potential for classical or quantum memory~\cite{PhysRevLett.81.2594,Lidar2003,byrd2004overview,albert2018lindbladians,lieu2020symmetry,rakovszky_defining_2024}.
In such cases, the stability of the phase would ensure the preservation of the degeneracy.

Recently, \textit{mixed-state symmetry-protected topological (SPT) states}~\cite{ma2023average,ma2023topological, niu_strange_2023,chirame_stable_2024,guo_locally_2024, paszko_edge_2023, verissimo_dissipative_2023,zhang2023fractonic, xue_tensor_2024,zhang_quantum_2024,zhang2024strange,lee2024symmetry,ma2024symmetry,kuno2024strong,chirame2024stabilizing} have emerged as promising candidates for inherently mixed-state phases.
Mixed-state SPT phases consist of short-range entangled density matrices which cannot be connected to density matrices in other phases by symmetric two-way dissipative evolution~\cite{coser2019classification}. 
The related topics of anomaly~\cite{hsin2023anomalies,lessa_mixed-state_2024,wang_anomaly_2024} and intrinsic topological order~\cite{bao2023mixed-state,lee2023quantum,ellison_towards_2024,sohal2024noisy} in mixed states have also been investigated in several recent works. 
Tensor network descriptions of mixed-state SPTs have also been used recently~\cite{Garre_Rubio_2023,guo_locally_2024, xue_tensor_2024}. 

When considering open quantum systems or density matrices, the notion of symmetry is augmented compared to the pure state setting, which leads to two types of symmetries~\cite{buca2012note,albert2014symmetries,albert2018lindbladians,lieu2020symmetry}.
Specifically, if the density matrix needs to be conjugated by a unitary of a symmetry group from both the ket and the bra sides to get back the same density matrix, then the symmetry is called a \textit{weak symmetry}.
On the other hand, if we can also get back the density matrix by applying different unitaries on the ket and the bra sides, then the symmetry is called a \textit{strong symmetry}. 
As pointed out in Refs.~\cite{coser2019classification,degroot2022symmetry}, there are no nontrivial mixed-state SPT phases if only weak symmetry is imposed. However, a mixed state can exhibit SPT order if the imposed symmetry is strong, although this order is captured by some pure-state SPT phase. The aforementioned mixed-state SPT order, requiring a combination of strong and weak symmetry, provides an opportunity for an inherently mixed-state phase.

The distinction between strong and weak symmetries also introduces a novel type of spontaneous symmetry breaking unique to open quantum systems, where the strong symmetry is spontaneously broken down to weak symmetry, dubbed strong-to-weak spontaneous symmetry breaking (SW-SSB)~\cite{ma2024symmetry,lessa_strong--weak_2024,sala_spontaneous_2024,kuno2024strong}.  
Detecting SW-SSB requires measuring quantities that are nonlinear in the density matrix, a task achievable in quantum devices but not in conventional condensed-matter systems. It has been suggested that a spontaneously broken strong symmetry can lead to a degenerate steady-state manifold~\cite{lieu2020symmetry}, which implies that it should be interesting to explore the steady-state manifold of SW-SSB.

Most studies on mixed-state phases, including mixed-state SPT order and SW-SSB, have focused on the ``quasilocal finite-depth quantum channel'' framework~\cite{coser2019classification}, akin to that used for defining gapped phases of ground states of closed systems~\cite{Hastings2005,chen2013symmetry}. However, given the interest in potential steady-state degeneracy in mixed-state phases, our work focuses on the steady-state phase of an open quantum system. Finding the steady states of an open quantum system can be mapped to finding extremal eigenstates of a non-Hermitian superoperator. Since ground states are extremal eigenstates of a Hamiltonian, the steady-state phase viewpoint provides a different analogy. While in the gapped ground-state scenario, the ``quasilocal finite-depth quantum circuit'' and ``gapped ground-state phase'' definitions are two sides of the same coin, the relationship between the ``quasilocal finite-depth quantum channel'' and ``steady-state phase'' is unclear. 
Our focus on the ``steady-state phase'' perspective therefore complements previous works.

\begin{figure}
    \centering
    \includegraphics[width=\linewidth]{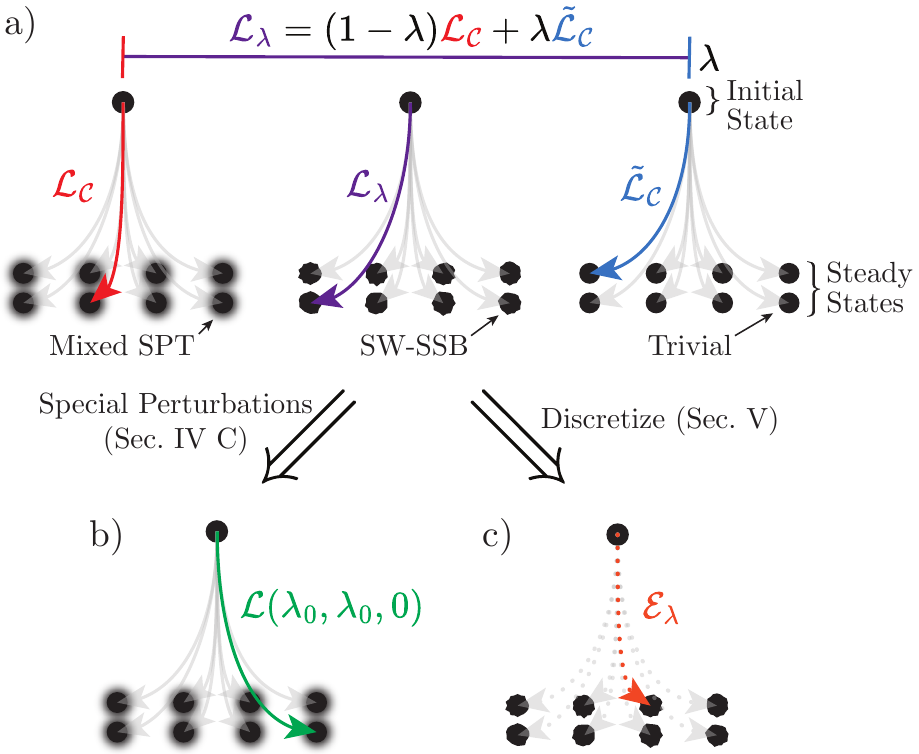}
    \caption{Summary of the main results of this work. a) In \cref{sec_perturbed_Lindbladian}, we study the steady-state phase diagram of a Lindbladian $\mathcal{L}_\lambda$ which interpolates between $\mathcal{L}_\Cl$, whose steady states exhibit mixed SPT order and $\tilde{\mathcal{L}}_\Cl$, whose steady states are trivially ordered. Rather than a finite phase of either order, we find that the entire intermediate region exhibits SW-SSB. b) We show in \cref{sec:weak_defects} that there exists a special class of perturbations to which the steady-state mixed SPT order is stable. c) In \cref{sec_clifford_simulation}, we study a parameterized quantum channel $\mathcal{E}_\lambda$ which approximates the dynamics of $\mathcal{L}_\lambda$ while being amenable to efficient classical simulation.}
    \label{fig:summary}
\end{figure}

Motivated by these considerations, we investigate the stability of mixed-state SPT order of steady-state phases of open quantum systems, using the \textit{decohered cluster state}---a simple mixed-state SPT order protected by a combined strong and weak symmetry~\cite{ma2023average}---as a test bed.
Our main results are summarized in~\cref{fig:summary}.
We first employ the matrix product density operator representation to analyze its mixed-state SPT order  characteristics, such as bulk-edge correspondence, nontrivial string order, and edge states.
We then construct a parent Lindbladian that hosts the decohered cluster state as a steady state with the required symmetry. 

Similar to the pure cluster state~\cite{briegel2001persistent,son2012topological}, the decohered cluster state can be disentangled by a depth-two circuit of controlled-Z gates, yielding the corresponding trivial mixed-state SPT order. This duality allows us to map the open system dynamics to an exactly solvable reaction-diffusion model~\cite{lushnikov1986binary}, enabling a thorough analysis of the Lindbladian, including the degeneracy of the steady-state manifold, previously unknown steady states, and the mixing time of the system.

We also study an interpolation between the parent Lindbladian of the decohered cluster state and that of the trivial mixed-state SPT order, analogous to the pure-state case~\cite{PhysRevLett.93.056402,PhysRevLett.103.020506}. Surprisingly, we find that the steady states at any point in this interpolation are neither nontrivial mixed-state SPT nor trivial mixed-state SPT orders. Instead, they exhibit SW-SSB, which destroys the mixed-state SPT order.
The mapping to the reaction-diffusion model suggests that the proliferation of point-like defects affecting the strong symmetry destroys the mixed-state SPT order, leading to SW-SSB, which is typical for symmetric perturbations. 
However, for perturbations that only affect weak symmetry, the mixed-state SPT order and its associate steady-state degeneracy remain stable.
Finally, we propose a quantum channel-based dynamics that replicates the essential physics of our findings. 
We are able to classically simulate this channel efficiently since it consists of Clifford gates, Pauli measurements, and feedback~\cite{gottesman1998heisenbergrepresentationquantumcomputers}. 

The structure of the paper is as follows. In~\cref{sec_decohered}, we introduce the decohered cluster state and describe its classification and characteristics, aided by tensor network representations. In~\cref{sec_parent_Lindbladian}, we construct a parent Lindbladian which hosts the state we introduced in~\cref{sec_decohered} as a steady state, and solve for the full spectrum of the parent Lindbladian, including all of the steady states, by mapping to an exactly solvable Lindbladian. In~\cref{sec_perturbed_Lindbladian}, we consider different types of perturbations to the parent Lindbladian and analyze the resulting steady states numerically via density matrix renormalization group (DMRG) and analytics, where possible. We find that the mixed-state SPT order of the steady state is unstable to SW-SSB for a wide class of symmetric perturbations to the Lindbladian. In~\cref{sec_clifford_simulation}, we consider a quantum channel which replicates the essential physics of the parent and perturbed Lindbladians and admits efficient Clifford simulation. We summarize the paper and give an outlook in~\cref{sec_discussion}.

\section{Decohered cluster state}\label{sec_decohered}
We begin by introducing the ``decohered cluster state''~\cite{ma2023average}, a certain type of cluster state in the presence of decoherence. We review its symmetries and classification within the framework of mixed-state SPT phases as discussed in  Refs.~\cite{coser2019classification,ma2023average,ma2023topological}. Using tensor network diagrammatics or matrix product density operators, we show that the decohered cluster state exhibits non-trivial string order and bulk-edge correspondence. 

\subsection{Cluster state and decoherence}

The well-known pure cluster state on a closed chain of $2N$ qubits~\cite{briegel2001persistent} can be expressed as
\begin{equation}\label{eq:pure_cluster}
    \ket{\Psi_\Cl}=\frac{1}{2^{N}}\sum_{\{z_{i}=\pm 1\} }\ket{\Psi_{\{z_{i}\}}},
\end{equation}
where 
\begin{equation}\label{eqn:domain wall configuration}
    \ket{\Psi_{\{z_{i}\}}}=\ket{Z\!=\!z_1}\ket{X\!=\!z_1z_2}\ket{Z\!=\!z_2}\cdots\ket{X\!=\!z_Nz_1}
\end{equation}
are \textit{decorated domain wall states}~\cite{chen2014symmetry}. Here, we define $\ket{Z = \pm 1}$ ($\ket{X = \pm}$) to be the eigenstates of $Z$ ($X$) with eigenvalues $\pm 1$. Note that we denote the eigenstates of $X$ with eigenvalues $\pm 1$ by $\ket{\pm}$ in this paper. The decorated domain wall states $\ket{\Psi_{\{z_{i}\}}}$ are constructed by first specifying a configuration $\{z_{i}\}$ which determines the $Z$ charge on the odd sites, and then ``decorating'' each domain wall with a state charged under $X$. More specifically, when $z_{i}\neq z_{i+1}$, the interceding even site $2i$ is fixed to $\ket{-}$, and to $\ket{+}$ otherwise. 

The decohered cluster state is defined by taking the equal-weight incoherent mixture of the decorated domain wall states shown in~\cref{eq:decohered_cluster_state}, rather than their coherent superposition. Defined on a closed chain of $2N$ qubits, the decohered cluster state reads \cite{ma2023average}
\begin{equation}\label{eq:decohered_cluster_state}
    \rhopp=\frac{1}{2^{N}}\sum_{\{z_{i}=\pm 1\} }\ketbra{\Psi_{\{z_{i}\}}}{\Psi_{\{z_{i}\}}}.
\end{equation}
As we will see later, it is exhibits nontrivial mixed-state SPT order. 

\subsection{Symmetry}\label{sec:symmetry}
We hereby discuss the notions of weak and strong symmetry applicable to mixed states~\cite{buca2012note,albert2014symmetries,albert2018lindbladians,lieu2020symmetry}.
In particular, a weakly $G$-symmetric matrix $M$ satisfies $U_g M U_g^\dagger= w(g) M$ for all $g\in G$, where $G$ is an Abelian symmetry group and $w(g)\in U(1)$.
On the other hand, a strongly $G$-symmetric matrix $M$ satisfies $U_g M=s_\text{ket}(g) M$ and $M U_g^\dagger=s_{\text{bra}}(g)M$ for all $g\in G$, where $s_{\text{ket}}(g)$, $s_{\text {bra}}(g) \in U(1)$. When $M$ is Hermitian, $s_\text{ket}(g)=s_\text{bra}(g)^*$. 

Consider the following operators in the Hilbert space of $2N$ qubits on a closed ring,
\begin{equation}\label{eq:Z2S and Z2W}
    S=\prod_{i \in \text{even}}X_i,\quad W=\prod_{i \in\text{odd}}X_i~.
\end{equation}
The decohered cluster state $\rho_{\mathcal{C}}$ is a mixed-state SPT order under the direct product of a strong $\ztwo$ symmetry and a weak $\ztwo$ symmetry, as classified in Refs.~\cite{ma2023average,ma2023topological}, where the strong $\ztwo$ symmetry (denoted as $\ztwo^{S}$ from now on) is generated by $S$, while the weak $\ztwo$ symmetry (denoted as $\ztwo^{W}$ from now on) is generated by $W$. 
We will use $s_{\text{ket}}$, $s_{\text{bra}}$, and $w$ (with argument $g$ suppressed) when we are referring to the charges under the strong symmetry $S$ and weak symmetry $W$.

For the decohered cluster state $\rhopp$ defined in Eq.~\eqref{eq:decohered_cluster_state}, we find that $W\rho_\mathcal{C} W^\dagger = \rho_\mathcal{C}$ and $S\rho_{\mathcal{C}} = \rho_{\mathcal{C}} S^\dagger = \rho_\mathcal{C}$, so that $\rho_\mathcal{C}$ lies in the $(s_{\text{ket}},s_{\text{bra}},w)=(+1, +1, +1)$ symmetry sector.
There is an intuitive way to understand  why $W$ acts as a weak symmetry while $S$ acts as a strong symmetry.
Note that $S$ counts the parity of the number of domain walls in each decorated domain wall state. 
With periodic boundary conditions, the system's topology enforces an even number of domain walls in every state, $S \ket{\Psi_{\{z_i\}}} = \ket{\Psi_{\{z_i\}}}$. 
This will be the case for both the bra states and the ket states in $\rhopp$.  We therefore see that $\rhopp$ must be strongly symmetric with $s_\text{ket} = s_\text{bra} = +1$. 
On the other hand, $W$ acts on the decorated domain wall states as $W \ket{\Psi_{\{z_i\}}}= \ket{\Psi_{\{\overline{z_i}\}}}$, where $\overline{z_i}\vcentcolon =-z_i$. This transformation is only a symmetry if it acts simultaneously from both the bra and the ket sides, indicating that $W$ can only be a weak symmetry.

\subsection{Classification as a mixed-state phase}
One physically motivated approach to define mixed-state phases is to define an equivalence class relation among the states, where the equivalence is defined by a short-time physical processes.
This mimics the definition of ground states being in the same phase if they can be adiabatically perturbed to each other rapidly.
This definition was put forward by Ref.~\cite{coser2019classification}, and we review their definition more rigorously in~\cref{app:string_lieb_robinson}.
Informally, it states that for two mixed states $\rho$ and $\sigma$ to be in the same phase, there has to exist two time-independent, symmetric
geometrically local Lindbladians $\mathcal{L}_1$ and $\mathcal{L}_2$ such that $\|e^{\mathcal{L}_1 t}(\rho) - \sigma \|_1$ and $\|e^{\mathcal{L}_2 t}(\sigma) - \rho \|_1$ are small, where $t$ is some short time (``fast-driven'') and where $\| X \|_1 \vcentcolon = \Tr [\sqrt{X^\dagger  X}]$ is the trace distance.
By geometrically local, we mean strictly geometrically local here, although the definition can be expanded to include quasi-local Lindbladians in certain situations~\cite{coser2019classification}.
Note that the smallness of the trace distance implies the closeness of all the observables $\hat{O}$ between $\rho$ and $\sigma$, i.e., $|\Tr [\rho O]-  \Tr [\sigma O]| \leq \|\rho - \sigma\|_1\cdot \|O\|_{\text{op}}$ by H\"older's inequality, where $\| \cdot \|_{\text{op}}$ is the operator norm. Furthermore, this ``fast-driven'' condition implies stability of local observables to local perturbations in the Lindbladian terms~\cite{cubitt2015stability, coser2019classification, onorati2023provably}.

One can also define the equivalence-class relation through quantum channels, analogous to the Lindbladian equivalence-class relation.
That is, if there exist two symmetric finite-depth local
quantum channels~\cite{ma2023average,ma2023topological} $\mathcal{E}_1$ and ${\mathcal{E}}_2$ 
such that $\|\mathcal{E}_1(\rho) - \sigma \|_1$ and $\|\mathcal{E}_2(\sigma) - \rho\|_1$ are small, then we say $\rho$ and $\sigma$ are in the same phase.
Note that in both the Lindbladian and quantum channel equivalence-class relations, the bi-directional ``fast-driven'' condition is important, as the Lindbladian evolution and the quantum channel are usually not reversible.

Building on the quantum-channel equivalence-class relation, Refs.~\cite{ma2023average,ma2023topological} show that if the strong-string order parameter (an order parameter used to diagnose SPT orders, see~\cref{sec:string_order_parameter}) in the state $\sigma$ is zero, then there does not exist a symmetric finite-depth channel $\mathcal{E}$ such that the string order parameter of $\mathcal{E}(\sigma)$ is of order one.
That is to say, if $\rho$ has a string order parameter of order one, then $\|\mathcal{E}(\sigma) - \rho\|_1$ can never be small.
We know from Refs.~\cite{ma2023average,ma2023topological} that $\rhopp$ has a nonzero string order parameter (see also~\cref{sec:string_order_parameter} for more explanation on this). 
Thus, $\rho_{\mathcal{C}}$ is in a nontrivial mixed-state SPT order protected by $\ztwo^{S} \times \ztwo^{W}$, as stated in Refs.~\cite{ma2023average,ma2023topological}.

Here, we complement the above result by showing that $\rhopp$ has nontrivial SPT order protected by $\ztwo^{S} \times \ztwo^{W}$ symmetry under the Lindbladian equivalence-class relation as well. 
Using Lieb-Robinson bounds \cite{lieb1972finite, poulin2010lieb}, we show that $\|e^{\mathcal{L}t}(\sigma) - \rho \|_1$ can never be small for a symmetric Lindbladian $\mathcal{L}$ within some short time $t$. 
The detailed statement and its proof are presented in~\cref{app:string_lieb_robinson}.

The above results can be interpreted as follows---the nontrivial SPT order of $\ztwo^{S} \times \ztwo^{W}$ is hard to build, which forbids a ``fast-driven'' symmetric Lindbladian or a symmetric finite-depth quantum channel from bringing a trivial state $\sigma$ close to $\rho_{\mathcal{C}}$, which is a nontrivial SPT state.
On the other hand, it is interesting to know if $\rho_{\mathcal{C}}$ can be brought close to a trivial state in a short time or in a short depth.
In~\cref{app:fast destruction string order}, we show that the string order parameter corresponding to the $\ztwo^{S} \times \ztwo^{W}$ symmetry defined in Eq.~(\ref{eq:Z2S and Z2W}) is easy to destroy.
In particular, we show that a simple trace channel and the corresponding trace Lindbladian can destroy the string order in depth one and in a short time, respectively. Interestingly, the destruction of the string order is accompanied by the ``strong-to-weak'' spontaneous symmetry breaking, which has received a lot of attention recently~\cite{ma2024symmetry,lessa_strong--weak_2024,sala_spontaneous_2024}.

\subsection{Tensor network representation}\label{sec:tensors}
It is instructive to express $\rhopp$ as a matrix product density operator, as it helps us expose the edge modes on an open chain, extract the projective representation, and calculate the string order parameters. We first define two families of matrices in the virtual space spanned by $\ket{v_0}=(1,0)^T$ and $\ket{v_1}=(0,1)^T$. These matrices are two-index tensors given by $A^{z,z'}=\Pi_z\delta_{z,z'}$ and $B^{x,x'} = X^{\frac{1-x}{2}}\delta_{x,x'}$, where $\Pi_z\vcentcolon = \delta_{z=+1}|v_0\ra\la v_0| + \delta_{z=-1}|v_1\ra\la v_1|$
. We can then define two four-index tensors:
\begin{align}\label{eq:odd_tensors}
\begin{split}
  \begin{tikzpicture}
    \node[label={[label distance=-0.5em]90:{$\alpha$}}] at (-0.5,0) {};
    \node[label={[label distance=-0.5em]90:{$\beta$}}] at (0.5,0) {};
    \node[label={[label distance=-0.5em]0:{$z$}}] at (0,0.5) {};
    \node[label={[label distance=-0.5em]0:{$z'$}}] at (0,-0.5) {};
    \draw (-0.5,0)--(0.5,0);
    \draw (0,-0.5) -- (0,0.5);
    \node[odd] (t) at (0,0) {};
  \end{tikzpicture}:&=
  \left[A^{z,z'}\right]_{\alpha,\beta}=
  \left[\Pi_z\delta_{z,z'}\right]_{\alpha,\beta},
\,
\end{split}
\end{align}

\begin{align}\label{eq:even_tensors}
\begin{split}
  \begin{tikzpicture}
    \node[label={[label distance=-0.5em]90:{$\alpha$}}] at (-0.5,0) {};
    \node[label={[label distance=-0.5em]90:{$\beta$}}] at (0.5,0) {};
    \node[label={[label distance=-0.5em]0:{$x$}}] at (0,0.5) {};
    \node[label={[label distance=-0.5em]0:{$x'$}}] at (0,-0.5) {};
    \draw (-0.5,0)--(0.5,0);
    \draw (0,-0.5) -- (0,0.5);
    \node[even] (t) at (0,0) {};
  \end{tikzpicture}:&= 
  \left[B^{x,x'}\right]_{\alpha,\beta}=
  \left[X^{\frac{1-x}{2}}\delta_{x,x'}\right]_{\alpha,\beta}
\,~,
\end{split}
\end{align}
where $z,z',x,x'\in\{+1,-1\}$ are the physical indices, $\alpha, \beta \in \{0,1\}$ are the virtual indices, and as above $Z\ket{Z=z} = z\ket{Z=z}$ and $X\ket{X=x} = x\ket{X=x}$. Abbreviating the configuration $\Lambda = (z_1,x_1, \cdots, z_N, x_N)$, we can then write $\rhopp = \sum_{\Lambda,\Lambda'}\rhopp(\Lambda,\Lambda')\ketbra{\Lambda}{\Lambda'}$,
where
\begin{align}
\begin{split}
    \rhopp(\Lambda,\Lambda')&=\Tr[A^{z_1,z_1'}B^{x_1,x'_1}\cdots A^{z_{N},z_{N}'}B^{x_N,x_N'}]~.
\end{split}
\end{align}
The sums over $\Lambda$ and $\Lambda'$ in the definition of $\rhopp$ each go over all $2^{2N}$ possibilities, while the restriction to decorated domain wall states is taken care of by the coefficients $\rhopp(\Lambda,\Lambda')$. 
Pictorially, we write
\begin{align}
\begin{split}
    \rhopp&=\begin{tikzpicture}
    \draw (0,0) rectangle (5.5,-0.75);
    \foreach \x in {0.5,2,4.25}{
		     \draw (\x,-0.55) -- (\x,0.55);
		    \node[odd] at (\x,0) {};
		  }
		  \foreach \x in {1.25,2.75,5}{
		     \draw (\x,-0.55) -- (\x,0.55);
		    \node[even] at (\x,0) {};
		  }
		  \node[fill=white] at (3.5,0) {$\dots$};
  \end{tikzpicture} \ ,
\end{split}
\end{align}
where we have used periodic boundary conditions. When tensors appear without indices explicitly denoted, we assume they are summed over with the corresponding bras and kets.

It is worth noting that the tensors defined in~\cref{eq:odd_tensors,eq:even_tensors} have the following ``pulling-through'' relations:
\begin{subequations}\label{eq:pulling_through_main}
\begin{align}
  \begin{tikzpicture}[baseline={(0,-.1)}]
    \draw (-.55,0)--(.55,0);
    \draw (0,-.55) -- (0,.55);
    \node[odd] (t) at (0,0) {};
    \node[X] at (0,0.325) {};
    \node[X] at (0,-0.325) {};
  \end{tikzpicture} 
\ = \begin{tikzpicture}[baseline={(0,-.1)}]
    \draw (-.55,0)--(.55,0);
    \draw (0,-0.55) -- (0,0.55);
    \node[odd] (t) at (0,0) {};
    \node[Xv] at (-0.325,0) {};
     \node[Xv] at (0.325,0) {};
  \end{tikzpicture}~, \
\end{align}
\vspace{-1 em}
\begin{equation}
  \begin{tikzpicture}[baseline={(0,-.1)}]
    \draw (-.55,0)--(.55,0);
    \draw (0,-0.55) -- (0,.55);
    \node[odd] (t) at (0,0) {};
    \node[Z] at (0,0.325) {};
  \end{tikzpicture} 
\ = \begin{tikzpicture}[baseline={(0,-.1)}]
    \draw (-.55,0)--(.55,0);
    \draw (0,-.55) -- (0,0.55);
    \node[odd] (t) at (0,0) {};
    \node[Z] at (0,-0.325) {};
  \end{tikzpicture}  \ = \begin{tikzpicture}[baseline={(0,-.1)}]
    \draw (-.55,0)--(.55,0);
    \draw (0,-0.55) -- (0,0.55);
    \node[odd] (t) at (0,0) {};
    \node[Zv] at (-0.325,0) {};
  \end{tikzpicture}  \ = \begin{tikzpicture}[baseline={(0,-.1)}]
    \draw (-.55,0)--(.55,0);
    \draw (0,-0.55) -- (0,0.55);
    \node[odd] (t) at (0,0) {};
    \node[Zv] at (0.325,0) {};
  \end{tikzpicture}~,
\end{equation}

\vspace{-1 em}
\begin{equation}
  \begin{tikzpicture}[baseline={(0,-.1)}]
    \draw (-.55,0)--(.55,0);
    \draw (0,-0.55) -- (0,.55);
    \node[even] (t) at (0,0) {};
    \node[X] at (0,0.325) {};
  \end{tikzpicture} 
\ = \begin{tikzpicture}[baseline={(0,-.1)}]
    \draw (-.55,0)--(.55,0);
    \draw (0,-.55) -- (0,0.55);
    \node[even] (t) at (0,0) {};
    \node[X] at (0,-0.325) {};
  \end{tikzpicture} \ =
\begin{tikzpicture}[baseline={(0,-.1)}]
    \draw (-.55,0)--(.55,0);
    \draw (0,-0.55) -- (0,0.55);
    \node[even] (t) at (0,0) {};
    \node[Zv] at (-0.325,0) {};
     \node[Zv] at (0.325,0) {};
  \end{tikzpicture}~, \
\end{equation}

\vspace{-1 em}
\begin{align}
  \begin{tikzpicture}[baseline={(0,-.1)}]
    \draw (-.55,0)--(.55,0);
    \draw (0,-.55) -- (0,.55);
    \node[even] (t) at (0,0) {};
    \node[Z] at (0,0.325) {};
    \node[Z] at (0,-0.325) {};
  \end{tikzpicture} 
\ = \begin{tikzpicture}[baseline={(0,-.1)}]
    \draw (-.55,0)--(.55,0);
    \draw (0,-0.55) -- (0,0.55);
    \node[even] (t) at (0,0) {};
    \node[Xv] at (-0.325,0) {};
  \end{tikzpicture} \ =
\begin{tikzpicture}[baseline={(0,-.1)}]
    \draw (-.55,0)--(.55,0);
    \draw (0,-0.55) -- (0,0.55);
    \node[even] (t) at (0,0) {};
    \node[Xv] at (0.325,0) {};
  \end{tikzpicture} \
    ~.
\end{align}
\end{subequations}
Using these relations, it will be easy to show that the symmetry $\ztwo^S \times \ztwo^W$ forms a projective representation on the virtual space. 

Let us briefly review the relationship between SPT order in 1d and projective representations. Pure-state SPT phases in 1d protected by a global symmetry $G$ are characterized and classified by edge modes: while the ground state is unique when the Hamiltonian is placed on a ring, it becomes degenerate when the Hamiltonian is placed on an open chain. Intuitively, this occurs because a non-trivial symmetry action arises at the edge which would cancel out if the edges were joined together. Algebraically, each symmetry operator $O_g$ labeled by $g\in G$ acts on the ground state space as $O_{g}^LO_{g}^R$, where $O^L_g$ has support only near the left edge and $O^R_g$ has support only near the right edge. While the global symmetries $\{O_g\}$ form a linear representation (i.e. $O_gO_h=O_{gh}$), the edge modes form only a projective representation (i.e. $O^{L/R}_gO^{L/R}_h=\omega(g,h)O^{L/R}_{gh}$, where $\omega\in U(1)$ is a phase factor.) The inequivalent assignments of phase factors $\omega(g,h)$ form a group under multiplication, which is called the second cohomology group $H^2(G,U(1))$. This is precisely the data which classifies 1d pure-state SPTs protected by $G$ symmetry in bosonic systems.

In Ref.~\cite{ma2023average}, it was shown that mixed-state SPT phases protected by a combination of a strong $K$ symmetry and weak $G$ symmetry are also classified by a cohomology group, namely $\bigoplus_{p=1}^2 H^{2-p}(G,H^p(K,U(1))$.
In the present case of the decohered cluster state, we have $G=K=\ztwo$, so that the above classification evaluates to $H^{1}(\ztwo,H^{1}(\ztwo,U(1)))=\ztwo$. This means that there is a single nontrivial mixed-state SPT phase with symmetry $\ztwo^\text{S}\times\ztwo^\text{W}$, with respect to the definition given in Ref.~\cite{ma2023average}. 
And indeed a nontrivial SPT order can be constructed from the domain-wall construction between the strong $\ztwo$ and the weak $\ztwo$ symmetries such as $\rhopp$.

To see that the decohered cluster state indeed realizes the $\ztwo^S \times \ztwo^W$ symmetry projectively at its boundary, we first construct the following family of density matrices on open boundary conditions:
\begin{equation}\label{eq:open_MPDO}
    \rho_{\alpha\beta} \vcentcolon  = \begin{tikzpicture}
    \draw (0,0) -- (5.5,0.0);
    \foreach \x in {0.5,2,4.25}{
		     \draw (\x,-0.55) -- (\x,0.55);
		    \node[odd] at (\x,0) {};
		  }
		  \foreach \x in {1.25,2.75,5}{
		     \draw (\x,-0.55) -- (\x,0.55);
		    \node[even] at (\x,0) {};
		  }
		  \node[fill=white] at (3.5,0) {$\dots$};
    \node[edge,label=above:$\alpha$] at (0,0) {};
    \node[edge,label=above:$\beta$] at (5.5,0) {};
  \end{tikzpicture} \ ,
\end{equation}
where $\begin{tikzpicture}[baseline={(0,-.1)}]
		  \draw[] (0,0) -- (0.35,0);
		  \node[edge,label=above:$\alpha$] at (0,0) {};
		\end{tikzpicture}=\bra{v_\alpha}$, $\begin{tikzpicture}[baseline={(0,-.1)}]
    \draw[] (0.65,0) -- (1,0);
	\node[edge,label=above:$\beta$] at (1,0) {};
	\end{tikzpicture}=\ket{v_\beta}$ and $\alpha,\beta \in \{ 0,1\}$.
The states in~\cref{eq:open_MPDO} span a 4-dimensional linear space with complex coefficients, where the physical density matrices correspond to the  convex space of Hermitian matrices with unit trace. This space can encode two classical bits of information.

We then act with the symmetry operators on $\rho_{\alpha\beta}$. Akin to the pure state case, we see via the pulling-through relations~\cref{eq:pulling_through_main} that the global symmetry factors to a product of operators acting on each virtual edge state:
\begin{align*}
\begin{split} 
W\rho_{\alpha\beta} W^\dagger &=
\begin{tikzpicture}[baseline={(0,-.1)}]
    \draw (-0.1,0) -- (5.6,0);
    \node[edge,label=above:$\alpha$] at (-0.1,0) {};
    \node[edge,label=above:$\beta$] at (5.6,0) {};
    \foreach \x in {0.5,2,3.5,5}{
        \draw (\x,-.55) -- (\x,.55);
		    \node[X] (t\x) at (\x,0.325) {};
		    \node[odd] at (\x,0) {};
            \node[X] (t\x) at (\x,-0.325) {};
		  }
	\foreach \x in {1.25,4.25}{
		     \draw (\x,-0.5) -- (\x,0.5);
		    \node[even] at (\x,0) {};
		  }
		  \node[fill=white] at (2.75,0) {$\dots$};
  \end{tikzpicture} \\ &=
  \begin{tikzpicture}[baseline={(0,-.1)}]
    \draw (-0.1,0) -- (5.6,0);
    \node[edge,label=above:$\alpha$] at (-0.1,0) {};
    \node[edge,label=above:$\beta$] at (5.6,0) {};
    \foreach \x in {0.5,2,3.5,5}{
        \draw (\x,-0.5) -- (\x,0.5);
		    \node[odd] at (\x,0) {};
		  }
	\foreach \x in {1.25,4.25}{
		     \draw (\x,-0.5) -- (\x,0.5);
		    \node[even] at (\x,0) {};
		  }
		  \node[fill=white] at (2.75,0) {$\dots$};
        \node[Xv] at (0.175,0) {};
        \node[Xv] at (5.5-.175,0) {};
  \end{tikzpicture} \ , 
  \end{split}
  \end{align*}
\begin{align*}
\begin{split}
S\rho_{\alpha\beta} &=
\begin{tikzpicture}[baseline={(0,-.1)}]
    \draw (-0.1,0) -- (5.6,0);
    \node[edge,label=above:$\alpha$] at (-0.1,0) {};
    \node[edge,label=above:$\beta$] at (5.6,0) {};
    \foreach \x in {0.5,2,3.5,5}{
        \draw (\x,-0.55) -- (\x,0.55);
		    \node[odd] at (\x,0) {};
		  }
	\foreach \x in {1.25,4.25}{
		     \draw (\x,-0.5) -- (\x,.55);
		    \node[even] at (\x,0) {};
      \node[X] (t\x) at (\x,0.325) {};
		  }
		  \node[fill=white] at (2.75,0) {$\dots$};
  \end{tikzpicture} \\ &=
  \begin{tikzpicture}[baseline={(0,-.1)}]
    \draw (-0.1,0) -- (5.6,0);
    \node[edge,label=above:$\alpha$] at (-0.1,0) {};
    \node[edge,label=above:$\beta$] at (5.6,0) {};
    \foreach \x in {0.5,2,3.5,5}{
        \draw (\x,-0.55) -- (\x,0.55);
		    \node[odd] at (\x,0) {};
		  }
	\foreach \x in {1.25,4.25}{
		     \draw (\x,-0.55) -- (\x,0.55);
		    \node[even] at (\x,0) {};
		  }
		  \node[fill=white] at (2.75,0) {$\dots$};
    \node[Zv] at (.175,0) {};
    \node[Zv] at (5.5-.175,0) {};
  \end{tikzpicture} \ , 
  \end{split}
  \end{align*}
\begin{align*}
\begin{split}
 \rho_{\alpha\beta} S^\dagger 
 &=
\begin{tikzpicture}[baseline={(0,-.1)}]
    \draw (-0.1,0) -- (5.6,0);
    \node[edge,label=above:$\alpha$] at (-0.1,0) {};
    \node[edge,label=above:$\beta$] at (5.6,0) {};
    \foreach \x in {0.5,2,3.5,5}{
        \draw (\x,-0.55) -- (\x,0.55);
		    \node[odd] at (\x,0) {};
		  }
	\foreach \x in {1.25,4.25}{
		     \draw (\x,-.55) -- (\x,0.55);
		    \node[even] at (\x,0) {};
      \node[X] (t\x) at (\x,-0.325) {};
		  }
		  \node[fill=white] at (2.75,0) {$\dots$};
  \end{tikzpicture} \\ 
  &=
  \begin{tikzpicture}[baseline={(0,-.1)}]
    \draw (-0.1,0) -- (5.6,0);
    \node[edge,label=above:$\alpha$] at (-0.1,0) {};
    \node[edge,label=above:$\beta$] at (5.6,0) {};
    \foreach \x in {0.5,2,3.5,5}{
        \draw (\x,-0.55) -- (\x,0.55);
		    \node[odd] at (\x,0) {};
		  }
	\foreach \x in {1.25,4.25}{
		     \draw (\x,-0.55) -- (\x,0.55);
		    \node[even] at (\x,0) {};
		  }
		  \node[fill=white] at (2.75,0) {$\dots$};
    \node[Zv] at (0.175,0) {};
    \node[Zv] at (5.5-.175,0) {};
  \end{tikzpicture} \ .
\end{split}
\end{align*}
We have found that the symmetries reduce to an effective action on the edge Hilbert space, namely
\begin{align}
\begin{split}
    W\rho_{\alpha\beta} W^\dagger &=\rho_{\alpha'\beta'},\\
    \ket{v_{\alpha'}}\ket{v_{\beta'}}&=X\otimes X\ket{v_{\alpha}}\ket{v_{\beta}},\\
    S\rho_{\alpha\beta} &= \rho_{\alpha\beta}S^\dagger =\rho_{\alpha''\beta''},\\
    \ket{v_{\alpha''}}\ket{v_{\beta''}}&=Z\otimes Z\ket{v_\alpha}\ket{v_\beta}.\\
\end{split}
\end{align}
As in the pure state case, these operators commute globally, but anti-commute if either edge is taken in isolation:
\begin{equation}
    [X\otimes X,Z\otimes Z]=0~,\quad \{X,Z\}=0~.
\end{equation}
The matrices $X$ and $Z$ obtained by pulling the symmetry to each edge generate the algebra $\{\mathbbm{1},X,Z,XZ\}$, which forms a projective representation of the group $\ztwo\times\ztwo$. This is precisely what we mean by the statement that the $\ztwo^S \times \ztwo^W$ symmetry is realized projectively at the edge. 

\subsection{String order parameters}\label{sec:string_order_parameter}
While conventional SPT orders do not spontaneously break symmetry and therefore do not admit local order parameters, they can be characterized by non-local \textit{string order parameters}~\cite{duivenvoorden2013from,nijs1989preroughening,kennedy1992hidden,pollmann2012detection,else2013hidden,bahri2014detecting}. 
In particular, it is necessary to measure a set of string order parameters labeled by the elements and irreducible representations of the group to obtain a list of zero and nonzero expectation values, which is known as the ``pattern of zeros'' ~\cite{pollmann2012detection,degroot2022symmetry}.
We here briefly review the definition of a string order parameter in the pure state case. Suppose we have an Abelian on-site symmetry $G$ realized by symmetry operators, each of which is a tensor product of single-site operators:
\begin{equation}
    U_g=\bigotimes_{\text{sites }i}U_g^{(i)}.
\end{equation}
In general, a string order parameter that detects a nontrivial projective representation of the $G$ symmetry consists of a finite string of $U_g^{(i)}$ operators in a region $\mathcal{R}$ with operators $\mathcal{O}^{(L/R)}_\alpha$ appended to the left and right ends of the string. Importantly, the operators $\mathcal{O}^{(L/R)}_\alpha$ transform under conjugate irreducible representations $\alpha$ and $\alpha^*$ of the other symmetry, which are simply complex numbers because the symmetry group is Abelian. We are then able to define
\begin{align}\label{eq:string_general_def}
\begin{gathered}
    \mathcal{S}_g=\mathcal{O}_\alpha^{(L)}\otimes\left(\prod_{i\in \mathcal{R}}U_g^{(i)}\right)\otimes\mathcal{O}_\alpha^{(R)},\\
    U_h^\dagger\mathcal{O}_\alpha^{(L)}U_h=\alpha(h)\mathcal{O}_\alpha^{(L)},\quad U_h^\dagger\mathcal{O}_\alpha^{(R)}U_h=\alpha^*(h)\mathcal{O}_\alpha^{(R)}.
\end{gathered}
\end{align}
The string operators $\mathcal{S}_g$ let us define the string order parameter for a pure state $\ket{\psi}$, which is simply $\langle \psi | \mathcal{S}_g | \psi \rangle$.

 One natural generalization of the string order parameter from pure states to mixed states is to measure the expectation value of the strong-string operator $\mathcal{S}_{nm}^S$ (defined below) in the state $\rho$:
\begin{equation} \label{eq:def_C_I_S}
    C_{\text{I}, nm}^S \vcentcolon  = \braket{\mathcal{S}_{nm}^S}_\rho=\Tr\left[\rho \mathcal{S}_{nm}^S \right]~,
\end{equation}
where 
\begin{equation}
    \label{eq:string_op}
    \mathcal{S}_{nm}^{S} \vcentcolon  = Z_{n} X_{n+1} X_{n+3} \cdots X_{m-3} X_{m-1} Z_{m}~,
\end{equation}
for odd $n$ and $m$. The Roman numeral I appearing in the subscript indicates that the correlator is linear in the density matrix (i.e. \renyi-1), and distinguishes it from other correlators which are quadratic in the density matrix, as we will define shortly. Notice that the string correlator $\mathcal{S}_{nm}^{S}$ consists of a truncated symmetry operator spanning from sites $n+1$ to $m-1$ and that the endpoint operators anticommute with the $X$ symmetry on odd sites, consistent with the construction in \cref{eq:string_general_def}. This definition has been used in prior works, and it has been shown that it is an order parameter for the $\ztwo^S \times \ztwo^W$ mixed-state SPT order~\cite{ma2023topological}. If we try to define an analogous order parameter for the weak-string operator $\mathcal{S}^W_{n, m} \vcentcolon = \mathcal{S}_{n+1, m+1}^S$, then we find that it is trivially one because the weak-string operator acts from both the ket and the bra sides and $\left( \mathcal{S}^W \right)^2 = \IId $, more specifically $\Tr\left[\mathcal{S}_{nm}^W  \rho \mathcal{S}_{nm}^W \right] = 1$. 

One way of generating a useful string correlator corresponding to the weak symmetry is to map density matrices to vectors in a doubled Hilbert space with twice the number of qubits and consider the string correlators in this doubled space~\cite{ma2024symmetry}.
The mapping of density matrices to vectors in the doubled Hilbert space is as follows:
\begin{align}
\label{eq:op_state_correspondence}
    \rho=\sum_{i,j}\rho_{ij}\ketbra{i}{j}\mapsto\sum_{i,j}\rho_{ij}\ket{i}_L\!\ket{j}_R =\vcentcolon \vert \rho \rangle~,
\end{align}
where the subscripts $L$ and $R$ denote the left and the right Hilbert spaces and $|i\ra$ and $|j\ra$ denote some computational basis.

In the doubled Hilbert space, the symmetry operators become $S \otimes \IId $, $\IId  \otimes S$ for the strong ket and the strong bra symmetries, and $W \otimes W$ for the weak symmetry.
Here the first factor of the tensor product acts on the left (ket) Hilbert space and the second factor acts on the right (bra) Hilbert space. 
We denote the expectation values of an operator $\mathcal{M}$ in the doubled Hilbert space by
\begin{equation}\label{eqn:double braket expectation}
    \langlee \mathcal{M} \ranglee_\rho \vcentcolon= \frac{\langle \rho | \mathcal{M} | \rho \rangle}{\langle \rho | \rho \rangle} = \frac{\Tr\left[ \rho^\dagger \mathcal{M}\left[ \rho \right]\right]}{\Tr\left[\rho^\dagger \rho\right]}~,
\end{equation}
where $\mathcal{M}\left[ \cdot \right]$ is a superoperator that corresponds to the operator $\mathcal{M}$ in the doubled Hilbert space. 
For example, $\la \rho| A \otimes B |\rho\ra = \Tr\left[\rho^\dagger A\rho B^{T}\right]$.

Now, by substituting the various symmetry operators in the definition of the string operator,~\cref{eq:string_general_def}, we find the various string correlators that we need.
If $U_g^{(i)} = (S\otimes \IId )^{(i)}$ on the even sites, and the end point operators are $\mathcal{O}_\alpha^{(L/R)} = Z$ on the odd sites $m$ and $n$, then we get the following correlator which we call the \renyi-2 strong string correlator:
\begin{equation}\label{eq:strong_2copy_nontrivial_correlator}
    C_{\text{II},nm}^S \vcentcolon = \langlee \mathcal{S}_{nm}^S \otimes \IId  \ranglee_\rho~, 
\end{equation}
where $\mathcal{S}_{nm}^{S}$ is defined as in~\cref{eq:string_op} and the Roman numeral II reflects the fact that the correlator is quadratic in the density matrix, or \renyi-2.
Note that $(W\otimes W)^\dagger (Z_{n/m} \otimes \IId  ) \left(W \otimes W\right) = - Z_{n/m} \otimes \IId  $.
Thus this choice of the end-point operators can in principle be used to detect nontrivial SPT order. 

If we take the symmetry to be the weak symmetry in~\cref{eq:string_general_def}, that is, $U_g^{(i)} = (W \otimes W )^{(i)}$ on the odd sites, and the end-point operators $\mathcal{O}_\alpha^{(L/R)} = Z$ on the even sites $n$ and $m$, then we get the \renyi-2 weak string correlator:
\begin{equation}\label{eq:weak_2copy_nontrivial_correlator}
    C_{\text{II},nm}^W \vcentcolon = \langlee \mathcal{S}_{nm}^W \otimes \mathcal{S}_{nm}^W \ranglee_\rho~, 
\end{equation}
where $\mathcal{S}^W_{n,m} = \mathcal{S}^S_{n+1, m+1}$ as defined earlier.
Note that $(S \otimes \IId )^\dagger (Z_{n/m} \otimes Z_{n/m}) (S \otimes \IId ) = -Z_{n/m} \otimes Z_{n/m}$. 

Next, we construct two more correlators by taking the end point operators to be $\IId $. These are the \renyi-2 trivial strong-string correlator and the \renyi-2 trivial weak-string correlator:
\begin{equation}\label{eq:2copy_trivial_correlators}
    \tilde{C}_{\text{II},nm}^S  \vcentcolon = \langlee \tilde{\mathcal{S}}_{nm}^S \otimes \IId  \ranglee_\rho \;,\; \tilde{C}_{\text{II},nm}^W  \vcentcolon = \langlee \tilde{\mathcal{S}}_{nm}^W \otimes \tilde{\mathcal{S}}_{nm}^W \ranglee_\rho~,
\end{equation}
 where $\tilde{\mathcal{S}}^S_{nm}$ has the strong symmetry operators in the bulk and $\IId $ as the end point operators with $n, m$ restricted to be odd, while $\tilde{\mathcal{S}}^W_{nm}$ has the weak symmetry operators in the bulk and $\IId $ as the end point operators with $n, m$ restricted to be even.

While for a generic density matrix, the \renyi-1 and \renyi-2 type of correlators are independent, Ref.~\cite{ma2024symmetry} shows that they are related to each other if the state is short-range entangled. The patterns of zeros used to diagnose $\ztwo^S \times \ztwo^W$ trivial and nontrivial SPT orders are summarized in Table~\ref{tab:pattern of zero}.

Now we show, using tensor network diagrams, that $C_{\text{I}, nm}^S = C_{\text{II}, nm}^S = C_{\text{II}, (n-1),(m-1)}^W = 1$ for $n,m \in \text{odd}$ in the state $\rhopp$.
Acting on $\rhopp$ with each of the string operators, we find
\begin{align*}
\begin{split} 
\mathcal{S}^W_{nm}\rhopp &\mathcal{S}^W_{nm} \\
&=
\begin{tikzpicture}[baseline={(0,-.1)}]
    \draw (0.5,0) -- (6.5,0);
    \node[fill=white] at (.5,0) {$\ldots$};
    \node[fill=white] at (6.5,0) {$\ldots$};
    \foreach \x in {2,3.5,5}{
        \draw (\x,-.55) -- (\x,.55);
		    \node[X] (t\x) at (\x,0.325) {};
		    \node[odd] at (\x,0) {};
            \node[X] (t\x) at (\x,-0.325) {};
		  }
	\foreach \x in {1.25,2.75,4.25,5.75}{
		     \draw (\x,-0.5) -- (\x,0.5);
		    \node[even] at (\x,0) {};
		  }
	\node[Z] at (1.25,-0.325) {};
    \node[Z] at (1.25,0.325) {};
    \node[Z] at (5.75,-0.325) {};
    \node[Z] at (5.75,0.325) {};
  \end{tikzpicture} \\ &=
  \begin{tikzpicture}[baseline={(0,-.1)}]
    \draw (0.5,0) -- (6.5,0);
    \node[fill=white] at (.5,0) {$\ldots$};
    \node[fill=white] at (6.5,0) {$\ldots$};
    \foreach \x in {2,3.5,5}{
        \draw (\x,-.55) -- (\x,.55);
		    \node[odd] at (\x,0) {};
		  }
	\foreach \x in {1.25,2.75,4.25,5.75}{
		     \draw (\x,-0.5) -- (\x,0.5);
		    \node[even] at (\x,0) {};
		  }
    \node[Xv] at (1.25+0.75/2,0) {};
    \node[Xv] at (5.75-0.75/2,0) {};
	\node[Zl] at (1.25,-0.325) {};
    \node[Zl] at (1.25,0.325) {};
    \node[Z] at (5.75,-0.325) {};
    \node[Z] at (5.75,0.325) {};
  \end{tikzpicture} \\ &= \rhopp ,
  \end{split}
  \end{align*}
\begin{align*}
\begin{split}\mathcal{S}^S_{nm}\rhopp &=
\begin{tikzpicture}[baseline={(0,-.1)}]
    \draw (-0.1,0) -- (5.6,0);
    \foreach \x in {0.5,2,3.5,5}{
        \draw (\x,-0.55) -- (\x,0.55);
		    \node[odd] at (\x,0) {};
		  }
	\foreach \x in {1.25,2.75,4.25}{
		     \draw (\x,-0.5) -- (\x,.55);
		    \node[even] at (\x,0) {};
      \node[X] (t\x) at (\x,0.325) {};
		  }
    \node[Z] at (0.5,0.325) {};
    \node[Z] at (5,0.325) {};
    \node[fill=white] at (-0.25,0) {$\ldots$};
    \node[fill=white] at (5.75,0) {$\ldots$};
  \end{tikzpicture} \\ &=
  \begin{tikzpicture}[baseline={(0,-.1)}]
    \draw (-0.1,0) -- (5.6,0);
    \node[fill=white] at (-0.25,0) {$\ldots$};
    \node[fill=white] at (5.75,0) {$\ldots$};
    \foreach \x in {0.5,2,3.5,5}{
        \draw (\x,-0.55) -- (\x,0.55);
		    \node[odd] at (\x,0) {};
		  }
	\foreach \x in {1.25,2.75,4.25}{
		     \draw (\x,-0.55) -- (\x,0.55);
		    \node[even] at (\x,0) {};
		  }
    \node[Zv] at (.5+0.75/2,0) {};
    \node[Zv] at (5-.75/2,0) {};
    \node[Zl] at (0.5,0.325) {};
    \node[Z] at (5,0.325) {};
  \end{tikzpicture} \\ &= \rhopp,
  \end{split}
  \end{align*}
\begin{align*}
\begin{split}
 \rhopp\mathcal{S}^S_{nm} &=
\begin{tikzpicture}[baseline={(0,-.1)}]
    \draw (-0.1,0) -- (5.6,0);
    \node[fill=white] at (-0.25,0) {$\ldots$};
    \node[fill=white] at (5.75,0) {$\ldots$};
    \foreach \x in {0.5,2,3.5,5}{
        \draw (\x,-0.55) -- (\x,0.55);
		    \node[odd] at (\x,0) {};
		  }
	\foreach \x in {1.25,2.75,4.25}{
		     \draw (\x,-.55) -- (\x,0.55);
		    \node[even] at (\x,0) {};
      \node[X] (t\x) at (\x,-0.325) {};
		  }
    \node[Z] at (0.5,-0.325) {};
    \node[Z] at (5,-0.325) {};
  \end{tikzpicture} \\ &=
  \begin{tikzpicture}[baseline={(0,-.1)}]
    \draw (-0.1,0) -- (5.6,0);
    \node[fill=white] at (-0.25,0) {$\ldots$};
    \node[fill=white] at (5.75,0) {$\ldots$};
    \foreach \x in {0.5,2,3.5,5}{
        \draw (\x,-0.55) -- (\x,0.55);
		    \node[odd] at (\x,0) {};
		  }
	\foreach \x in {1.25,2.75,4.25}{
		     \draw (\x,-0.55) -- (\x,0.55);
		    \node[even] at (\x,0) {};
		  }
    \node[Zv] at (.5+0.75/2,0) {};
    \node[Zv] at (5-.75/2,0) {};
    \node[Z] at (0.5,-0.325) {};
    \node[Z] at (5,-0.325) {};
  \end{tikzpicture} \\ &= \rhopp.
\end{split}
\end{align*}
Because $\rho_\Cl$ is invariant under the action of the strings, the string order parameters will be identically one whether we chose to contract the physical legs of $\rho_\Cl$ according to the \renyi-1 or \renyi-2 conventions. 

Using the pulling-through relations~\cref{eq:pulling_through_main} and the definition of $\rhopp$ [see~\cref{eq:decohered_cluster_state}], we can show that  $\tilde{C}_{\text{I}, nm}^S = \tilde{C}_{\text{II}, nm}^S = \tilde{C}_{\text{II}, nm}^W = 0$ in the decohered cluster state.

\section{Parent Lindbladian}\label{sec_parent_Lindbladian}
In the study of quantum phases of pure states, we traditionally begin with a Hamiltonian that depends on a set of parameters. Quantum phases are then understood to be finite volumes of Hamiltonian parameter space in which the ground state(s) of the Hamiltonian display qualitatively similar features, such as long-range order and degeneracy. These features can be diagnosed by measuring local and/or nonlocal order parameters whose values stratify Hamiltonian parameter space into distinct phases. In an effort to generalize this perspective to the mixed state case, we will construct a parent \textit{Lindbladian} for the decohered cluster state and map out the phase diagram in its vicinity by measuring local and nonlocal order parameters.

If a Hamiltonian is in a nontrivial SPT phase, then it is known that the phase is stable to symmetric perturbations. Specifically, if one adds a symmetric perturbation to the Hamiltonian, then the ground state degeneracy does not change (in the thermodynamic limit) and ground states of the perturbed Hamiltonian will have a nonzero string order parameter, reflecting a nontrivial projective representation.
Our aim in this section is to test if this analogy carries over to mixed states and Lindbladians. 
We construct a local Lindbladian whose steady state is the decohered cluster state $\rhopp$. 
We then add symmetric local Lindbladian perturbations and study the properties of the new steady states.
This raises the questions of whether there is a way of defining mixed state phases of matter which focuses on parent Lindbladians rather than density matrices, and whether this perspective reveals anything new. Toward this aim, we will construct a $\ztwo^\text{S}\times\ztwo^\text{W}$-symmetric parent Lindbladian $\mathcal{L}_{\Cl}$ which hosts $\rhopp$ as a NESS, i.e. $\mathcal{L}_{\Cl} (\rhopp) =0$.

\subsection{The model}\label{sec_model}
We define the parent Lindbladian on an open chain of $2N$ qubits labeled by $1, 2, \ldots , 2N$ as follows:
\begin{align}
\label{eq:parent_lind}
\mathcal{L}_{\Cl} (\rho) \vcentcolon = & \sum_{\substack{j=1}}^{N-1}
\mathcal{D} [L_{0,2j}] (\rho) +\mathcal{D} [L_{2,2j-1}] (\rho) \nonumber \\
& + \sum_{j=1}^N \mathcal{D} [L_{1,2j-1}] (\rho) 
,  
\end{align}
where $\mathcal{D} [A] (\rho) \vcentcolon = A \rho A^{\dag} - \frac{1}{2}\{ A^{\dag} A , \rho \} $ denotes the standard Lindblad dissipator for a jump operator $A$~\cite{lindblad1976generators,gorini1976completely}.
For periodic boundary conditions, we add the terms $\mathcal{D} [L_{0,2j}] (\rho) $ and $\mathcal{D} [L_{2,2j-1}](\rho)$ for $j=N$ to \cref{eq:parent_lind}, where the site label $i+2N \vcentcolon =  i$.
The jump operators $L_{\mu,j}$ are
\begin{subequations}\label{eq:parent_jumps}
    \begin{equation}
        L_{0,j}=\frac{1}{2}X_{j+1}\left(\IId   -Z_{j-1}X_{j}Z_{j+1}\right),
    \end{equation}
    \begin{equation}
        L_{1,j}=Z_{j},
    \end{equation}
    \begin{equation}
        L_{2,j}=Z_{j-1}X_{j}Z_{j+1}.
    \end{equation}
\end{subequations}
Since all the jump operators commute with $S$, the Lindbladian $\mathcal{L}_{\mathcal{C}}$ has the $\ztwo^S$ symmetry, while it is easy to check that $\mathcal{D}[L_{1,j}]$ renders the $W$ symmetry weak.
Thus, our parent Lindbladian has the $\ztwo^\text{S}\times\ztwo^\text{W}$ symmetry introduced in~\cref{sec:symmetry}.
Each of the terms $L_{0,j}$, $L_{1,j}$, $L_{2,j}$ plays a distinct role in driving an arbitrary density matrix in the same symmetry sector as $\rhopp$ to the steady state $\rhopp$. 

The jump operators have the following physical meanings.
$L_{0,j}$ for even $j$ is designed to stabilize the domain-wall configurations. If the configuration is violated, we have $(\IId -Z_{j-1}X_jZ_{j+1})/2 = +1$, and we fix the configuration by applying $X_{j+1}$. Note that this configuration can also be fixed by applying $X_{j-1}$, $Z_{j}Z_{j+2}$, $Z_j$, etc. However, applying $Z_j$ is forbidden as it breaks the strong symmetry $\ztwo^S$; we choose to apply $X_{j+1}$ because the resulting jump operator has the lowest weight.
Now within the domain-wall configuration space, $L_{1,j}$ is designed to decohere the ``off-diagonal'' domain-wall configurations, leaving only $\ketbra{\Psi_{\{z_i\}}}{\Psi_{\{z_i\}}}$, while $L_{2,j}$ 
makes the probability of getting all the configurations equal.

Note that $\rho_\mathcal{C}$ is a steady state for any values of the rates of the jump operators, but we will choose these rates to be equal. We have verified that changing the rates does not change the physics qualitatively. 

Recall that the decohered cluster state is a mixed-state SPT state protected by a $\ztwo^{S} \times \ztwo^{W}$ symmetry, and the parent Lindbladian $\cL_{\mathcal{C}}$ has  the same $\ztwo^{S} \times \ztwo^{W}$ symmetry. 
The strong symmetry of the Lindbladian forces there to be at least one steady state in each $s_{\text{ket}}=s_{\text{bra}}^*$ physical strong symmetry sector. The argument is as follows. Because the Lindbladian is strongly symmetric, it conserves strong symmetry charge. That is, $S\rho=s_{\text{ket}}\rho$ and $\rho S=s_{\text{bra}} \rho $ implies
\begin{equation}
    S\mathcal{L}[\rho]=\mathcal{L}[S\rho]=\mathcal{L}[s_{\text{ket}}\rho]=s_{\text{ket}}\mathcal{L}[\rho]~,
\end{equation}
and similarly when acting from the bra side. 

One can then consider the dynamics where one prepares an initial state with the symmetry charge $s_{\text{ket}}=s_{\text{bra}}^*$ and evolve it with the strongly symmetric Lindbladian. 
At an infinite time, such an initial state has to evolve into some state with the same strong symmetry charges, hence forcing the Lindbladian to have at least one steady state in the corresponding strong symmetry sector. 

On the other hand, it is not guaranteed that a steady state will exist in the strong symmetry sectors $s_{\text{ket}}\neq s_{\text{bra}}^*$ or with a nontrivial weak charge.
The above argument breaks down since one cannot prepare a physical initial state in the said symmetry sectors, as a state in the said symmetry sectors is traceless.
More specifically, in the strong case, $\Tr[\rho]=\Tr[S\rho S^\dagger]=(s_{\text{ket}}s_{\text{bra}})\Tr[\rho]$, which implies $s_{\text{ket}}s_{\text{bra}}=1$ or $\Tr[\rho]=0$.
In the weak case, $\Tr[\rho]=\Tr[W\rho W^\dagger]=w\Tr[\rho]$, so $w=1$ or $\Tr[\rho]=0$.

Since $\rho_{\mathcal{C}}$ is in the $s_{\text{ket}}= s_{\text{bra}}=+1$ sector, it is interesting to examine what is the state in the $s_{\text{ket}}= s_{\text{bra}}=-1$ sector.
As we will show, the corresponding steady state in the $s_{\text{ket}}= s_{\text{bra}}=-1$ sector is $\rhomm \propto \sum_{j \in \text{even}} Z_j \rhopp Z_j$.

\subsection{Dual perspective and steady states}\label{sec:dualperspective}
In fact, we can partially solve and understand the parent Lindbladian $\mathcal{L}_\Cl$ by mapping the system to a dual model.
This is akin to understanding the one-dimensional pure cluster state and its parent Hamiltonian by conjugating them via a depth-two circuit composed of CZ (controlled-Z) gates~\cite{briegel2001persistent}.
Recall that, for the pure cluster state defined in~\cref{eq:pure_cluster}, 
\begin{align}\label{eq:U_CZ}
    U_{\text{CZ}}\ket{\Psi_\mathcal{C}}=\ket{+}^{\otimes N},\quad U_{\text{CZ}}=\prod_{j=1}^{2N} \text{CZ}_{j,j+1},
\end{align} 
where $\text{CZ}_{j,j+1}\vcentcolon = (1+Z_j)/2 + (1-Z_j)Z_{j+1}/2$ is the controlled-Z two-qubit gate and $j+2N \! \vcentcolon = \! j$. The circuit is illustrated in~\cref{fig:U_CZ}. Note that $\text{CZ}_{j,j+1}$ is invariant under $j \longleftrightarrow j+1$ and $U_{\text{CZ}} Z_j U_{\text{CZ}}^{\dagger} = Z_j$.
Since $U_{\text{CZ}} X_j U_{\text{CZ}}^{\dagger} = Z_{j-1}X_jZ_{j+1}$, the operators $S$ and $W$ are invariant under the conjugation by $U_{\text{CZ}}$.
In the following, we apply the $U_{\text{CZ}}$ circuit to map $\rhopp$ to a tensor-product state $\tilde\rho_\Cl$ and $\mathcal{L}_\Cl$ to the Lindbladian $\tilde{\mathcal{L}}_\Cl$, allowing us to solve for the steady states and some of the eigenpairs (eigenvalues and corresponding eigenvectors) of $\mathcal{L}_\Cl$.

\begin{figure}
    \centering
    \begin{tikzpicture}[baseline={(0,.4)}]
    \node [label={[label distance=-0.2em]180:\large $U_{\text{CZ}}=$}] at (-1,1.5) {};
    \foreach \y in {0,1,1.5,2,2.5}{
		     \draw (-1,\y) -- (1,\y);}
    
    \draw (-0.5,2.5) -- (-0.5,2);
    \node[CZ] at (-0.5,2.5) {};
    \node[CZ] at (-0.5,2) {};

    \draw (0.5,1.5) -- (0.5,2);
    \node[CZ] at (0.5,1.5) {};
    \node[CZ] at (0.5,2) {};

    \draw (-0.5,1.5) -- (-0.5,1);
    \node[CZ] at (-0.5,1.5) {};
    \node[CZ] at (-0.5,1) {};

    \draw (0.5,1) -- (0.5,0.5);
    \node[CZ] at (0.5,1) {};

    \draw (-0.5,0) -- (-0.5,0.5);
    \node[CZ] at (-0.5,0) {};
    \node[CZ] at (0.5,0) {};
    \node[CZ] at (0.5,2.5) {};

    \draw [black] plot [smooth, tension=0] coordinates { (0.5,0) (0.5,-0.15) (0.65,-0.15) (0.65,2.65) (0.5,2.65)(0.5,2.5)};

    \draw[draw=white, fill=white] (-1,0.25) rectangle (1,0.75);
    \usetikzlibrary{shapes.geometric,decorations,decorations.pathmorphing}
    \node[inner sep=0.1em, outer sep=0, text height=2.7ex,text depth=0.3ex] at (0.53,0.5) {$\vdots$};
    \node[inner sep=0.1em, outer sep=0, text height=2.7ex,text depth=0.3ex] at (-0.5,0.5) {$\vdots$};
  \end{tikzpicture}\qquad\begin{tikzpicture}[baseline={(0,.4)}]
    \draw (1.5,1.75) -- (1.5,1.25);
    \node[CZ] at (1.5,1.75) {};
    \node[CZ] at (1.5,1.25) {};
    \node [label={[label distance=-0.2em]0:$=$CZ}] at (1.5,1.5) {};
  \end{tikzpicture}
    \caption{An illustration of the $U_{\text{CZ}}$ circuit appearing in~\cref{eq:U_CZ}.}
    \label{fig:U_CZ}
\end{figure}
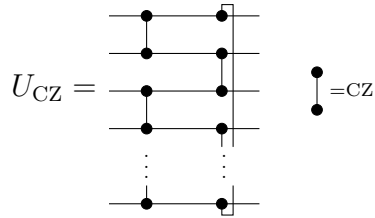

First, let us examine the state $\tilde{\rho}_{\mathcal{C}}=U_{\text{CZ}} \rhopp U_{\text{CZ}}^{\dagger}$. 
Recall that  $\rhopp$ can be written as in~\cref{eq:decohered_cluster_state}, which is an ensemble of the domain wall configurations $\ket{\Psi_{\{z_{i}\}}}$ given in Eq.~(\ref{eqn:domain wall configuration})
It is easy to see that
\begin{equation*}
    U_{\text{CZ}}\ket{\Psi_{\{z_i\}}}=\ket{Z\!=\!z_1}\ket{X\!=\!+1}\ket{Z\!=\!z_2}\cdots\ket{X\!=\!+1}~.
\end{equation*}
Summing over all the possible configurations of $z_{i}$ on the ``Z''-qubits, we have 
\begin{equation}\label{eq:CZ decohered cluster}
    \tilde{\rho}_{\mathcal{C}}=\bigotimes_{i \in \text{odd}} \frac{\IId _i}{2} \bigotimes_{j \in \text{even}} \ketbra{+}{+}_j~,
\end{equation}
which is a tensor-product state with the same $\ztwo^\text{S}\times\ztwo^\text{W}$ symmetry realized by $S$ and $W$.

We obtain the jump operators defining $\tilde{\mathcal{L}}_{\Cl}$ by conjugating the original jump operators by $U_{\text{CZ}}$, giving us 
\begin{align}\label{eq:dual_jumps}
\begin{split}
    \tilde{L}_{0,j} &=\frac{1}{2}Z_{j}X_{j+1}Z_{j+2}(\mathbbm{1} - X_{j})~, \\
    \tilde{L}_{1,j} &= Z_j~, ~~~~\tilde{L}_{2,j} = X_j~,
\end{split}
\end{align}
where $j$ is on all the even sites for $\tilde{L}_{0,j}$ and on all the odd sites for $\tilde{L}_{1,j}$ and $\tilde{L}_{2,j}$.

\subsubsection{Periodic boundary conditions}
The Lindbladian $\tilde{\mathcal{L}}_{\Cl}$ with periodic boundary conditions can be solved as follows.
First, considering the Lindbladian $\tilde{\mathcal{L}}_{12,j}\vcentcolon =\tilde{\mathcal{L}}_{1,j}+\tilde{\mathcal{L}}_{2,j}$ on the odd sites, we have $\tilde{\mathcal{L}}_{12,j}(\rho)=Z_j\rho Z_j + X_j \rho X_j - 2\rho$. The steady state can be easily solved as $\mathbbm{1}_j/2$ and is unique, with a Lindbladian (i.e.\ dissipative) gap (i.e., the eigenvalue of $\tilde{\mathcal{L}}_{12,j}$ with lowest nonzero real part) $\Delta_j = 2$.
Since a physical Lindbladian always has eigenvalues with non-positive real parts, we can try a potential steady state ansatz with the odd sites being $\bigotimes_{j \in \text{odd}} \frac{\mathbbm{1}_j}{2}$.
Within this subspace, the jump operators of $\tilde{\mathcal{L}}_{0}$ on the even sites become
\begin{align}\label{eq:effective CZ L0} 
\begin{split}
    \tilde{L}_{0,j}&=\frac{1}{2}Z_{j}Z_{j+2}(\mathbbm{1} - X_{j}) \\
    &=(|+\rangle \langle -|)_{j}(|+\rangle \langle -|+|-\rangle \langle +|)_{j+2}~,
\end{split}
\end{align}
where $j$ is on the even sites.
It can be seen that this jump operator annihilates $\bigotimes_{j \in \text{even}} |+\rangle \langle +|_j$, so we indeed have $\tilde{\rho}_{\mathcal{C}}$ as a steady state. 
In fact, if we interpret $\bigotimes_{j \in \text{even}} |+\rangle \langle +|_j$ as the vacuum and $|-\rangle\!\bra{-}_j$ as a quasiparticle or a defect on the even site $j$, then the effective jump operator is describing the processes of a particle hopping to the next empty even site on the right or annihilating with the particle on the next even site on the right if occupied. 

From this perspective, it is then also easy to guess and check that $\tilde{\rho}_{-}=(1/N)\sum_{j \in \text{even}} Z_j \tilde{\rho}_{\mathcal{C}} Z_j$ is a steady state in the $s_{\text{ket}} = s_{\text{bra}} = -1$ symmetry sector. 
Note that, since the $Z_j$ operator commutes with $U_{\text{CZ}}$, a steady state in the $s_{\text{ket}} = s_{\text{bra}} = -1$ symmetry sector of $\mathcal{L}_{\mathcal{C}}$ is $\rho_{-}=\frac{1}{N}\sum_{j \in \text{even}} Z_j \rho_{\mathcal{C}} Z_j$.
Since the Lindbladian cannot create a $|-\rangle$ particle, we expect the two states $\rho_{\mathcal{C}}$ and $\rho_{-}$ to be the only steady states of the Lindbladian when subject to periodic boundary conditions, which we confirmed numerically by exact diagonalization and density matrix renormalization group (DMRG) calculations, as discussed in~\cref{sec_perturbed_Lindbladian}.

It is natural to consider whether $\rhomm$ is also a mixed-state SPT state. We answer this question in the affirmative by showing in~\cref{ap:rho minus} that it has unit expectation value of the string order parameters used to analyze $\rhopp$. 

\subsubsection{Open boundary conditions}\label{sec:OBC_steady_states}
The Lindbladian in the case of open boundary conditions (OBC) can also be solved analogously.  
However, to keep the symmetries the same in the dual picture, we need to use the same $U_{\text{CZ}}=\prod_{j=1}^{2N} \text{CZ}_{j,j+1}$ even for open boundary conditions.
The ansatz for the steady state is constructed again by projecting the Lindbladian onto the steady-state subspace of $\tilde{\mathcal{L}}_{12,j}$. 
Note that, for $j=1$, $\tilde{\mathcal{L}}_{12,j}(\rho)=Z_j \rho Z_j - \rho$, which has the steady-state subspace spanned by $\{\mathbbm{1}_1,Z_1\}$. The steady-state subspace on the odd sites is therefore spanned by $\bigotimes_{i \in \text{odd}} \frac{\mathbbm{1}_i}{2}$ and $Z_1\bigotimes_{i \in \text{odd}} \frac{\mathbbm{1}_i}{2}$.
Note that these two states form a classical bit on the site $j=1$.

After projection onto the above steady-state subspace on the odd sites, the effective Lindbladian $\mathcal{L}_{0}$ is formed by the jump operators Eq.~(\ref{eq:effective CZ L0}) but with terms $j= 1,3, \dots , 2N\!-\!3$. This again describes a particle hopping (to-the-right) and annihilation process, but with the hopping that ends on site $j=2N$.

Again, one can see that $\tilde{\rho}_{\mathcal{C}}$ in Eq.~(\ref{eq:CZ decohered cluster}) is still a steady state, since $\tilde{L}_{0,j}$ annihilates $\tilde{\rho}_{\mathcal{C}}$. 
We also notice that, since $|-\rangle_L$ cannot ``hop'' across the boundary, $Z_{2N}\tilde{\rho}_{\mathcal{C}}$, $\tilde{\rho}_{\mathcal{C}}Z_{2N}$ and $Z_{2N}\tilde{\rho}_{\mathcal{C}}Z_{2N}$ are also steady states. 
Note that these four states form a qubit on site $j=2N$.
The other four states can be easily obtained as $Z_1 \tilde{\rho}_{\mathcal{C}}$, $Z_1 Z_{2N}\tilde{\rho}_{\mathcal{C}}$, $Z_1 \tilde{\rho}_{\mathcal{C}}Z_{2N}$, and $Z_1 Z_{2N}\tilde{\rho}_{\mathcal{C}}Z_{2N}$.
Since $U_{\text{CZ}}$ commutes with $Z_j$, we find that the steady states of $\mathcal{L}_{\mathcal{C}}$ are spanned by the eight states listed in the OBC row of~\cref{tab:steady states}. We have numerically confirmed that these eight states are the only steady states of $\mathcal{L}_{\mathcal{C}}$.

In order to better understand the OBC steady states, we can use as a basis the states $\rho_{\alpha\beta}$ defined in~\cref{eq:open_MPDO} and $\rho_{\alpha\beta}'\vcentcolon =Z_{2N}\rho_{\alpha\beta}$.
We may now express all of the OBC steady states as linear combinations of $\rho_{\alpha\beta}$ and $\rho_{\alpha\beta}'$. These alternative expressions for the steady states are also tabulated in~\cref{tab:steady states}.

\begin{table}
\begingroup
\renewcommand{\arraystretch}{1.35}
\begin{equation*}
    \begin{array}{|c|rl|ccc|}\hline
       &\multicolumn{2}{r}{\hbox{\diagbox[width=11em, height=3em]{Steady State}{Charge}}}\vline & s_\text{ket} & s_\text{bra} & w \\\hline
       \multirow{2}*{\rotatebox{90}{\text{PBC}}} & \multicolumn{2}{c}{\rhopp}\vline & +1 & +1 & +1 \\
        & \multicolumn{2}{c}{\rhomm}\vline & -1 & -1 & +1 \\\hline
      \multirow{8}*{\rotatebox{90}{\text{OBC}}} & \rhopp&=\rho_{00}+\rho_{11} & +1 & +1 & +1 \\
      & Z_1\rhopp &= \rho_{00}-\rho_{11}  & +1 & +1 & -1 \\
      & Z_{2N}\rhopp Z_{2N}&=\rho_{01}+\rho_{10}  & -1 & -1 & +1 \\
      & Z_1Z_{2N}\rhopp Z_{2N}&=\rho_{01}-\rho_{10} & -1 & -1 & -1 \\
      & Z_{2N}\rhopp&=\rho_{00}^{\prime}+\rho_{11}^{\prime} & -1 & +1 & +1 \\
      & Z_1Z_{2N}\rhopp&=\rho_{00}^{\prime}-\rho_{11}^{\prime}  & -1 & +1 & -1 \\
      & \rhopp Z_{2N}&=\rho_{01}^{\prime}+\rho_{10}^{\prime} & +1 & -1 & +1 \\
      & Z_1\rhopp Z_{2N}&=\rho_{01}^{\prime}-\rho_{10}^{\prime} & +1 & -1 & -1 \\\hline
    \end{array}
\end{equation*}
\endgroup
\captionsetup{justification=justified,singlelinecheck=false}
\caption{The steady states of the parent Lindbladian $\mathcal{L}_\Cl$ under periodic boundary conditions (PBC) and open boundary conditions (OBC) and their corresponding symmetry sectors. As discussed in~\cref{sec:OBC_steady_states}, the last four OBC steady states are accidental, and can be lifted by adding a symmetric perturbation.} 
\label{tab:steady states}
\end{table}
 
Notice that, for open boundary conditions, there is one steady state in each charge sector of the symmetries 
$s_\text{ket}$, $s_\text{bra}$, and $w$. As we discussed in~\cref{sec_model}, only states with the same strong charge on the ket and bra and with trivial weak charge are physical, in that all other density matrices are traceless. However, the unphysical steady states still contribute to the physical steady-state degeneracy because they can be added to physical steady states, resulting in new physical states which are beyond the span of the original physical states. 

Recall from~\cref{sec:tensors} that the four states $\rho_{\alpha \beta}$ naturally arise from the projective representation of the symmetry on the virtual space. These states only span a four-dimensional linear space, however, so it is possible that our excess degeneracy is accidental. If we add an additional Lindbladian $D[X_{j=2N}]$, which has the $X_{j=2N}$ jump operator on the rightmost site, we  find that only the states $\rho_{\alpha\beta}$ in~\cref{tab:steady states} remain steady states of the modified Lindbladian.
We expect a degeneracy of four is the minimal steady-state degeneracy of a Lindbladian on open boundary conditions hosting this nontrivial mixed-state SPT order, due to the nontrivial projective representation.
The states we find also satisfy the requirement imposed by strong symmetry, namely that there is a steady state in the $(+1,+1,+1)$ and $(-1,-1,+1)$ sectors of the $(s_\text{ket},s_\text{bra},w)$ symmetry.

\subsection{Lindbladian gap and mixing time}

\label{sec:lindblad_mixing}
The particle hopping-annihilation interpretation of the effective Lindbladian also allows us to solve the spectrum and analyze the mixing time. 
Since the Lindbladian $\tilde{\mathcal{L}}_{12,j}$ is gapped, we again consider the subspace where the odd sites are in the $\mathbbm{1}/2$ state, and only consider the effective Lindbladian $\tilde{\mathcal{L}}_0$ formed by the jump operators Eq.~(\ref{eq:effective CZ L0}) projected onto this subspace. 
Furthermore, we consider the ``diagonal (classical) subspace'' on the even sites $\rho = \sum_{\boldsymbol{\sigma}}  \rho(\boldsymbol{\sigma}) |\boldsymbol{\sigma}\rangle \langle \boldsymbol{\sigma}|$, where $\boldsymbol{\sigma} = (\sigma_2, \sigma_4, \ldots, \sigma_L) $ and $\sigma_{j}=+,-$ on the even sites.
For convenience, we map $|\boldsymbol{\sigma}\ra \la \boldsymbol{\sigma}|$ to $|\boldsymbol{\sigma})$ and treat $\rho \rightarrow |\rho ) \vcentcolon = \sum_{\boldsymbol{\sigma}} \rho(\boldsymbol{\sigma}) |\boldsymbol{\sigma})$ as a vector.

Within this subspace, the master equation of the density matrix reduces to a master equation for the probability distribution $\partial_t |\rho_t) = M |\rho_t)$, where $M=\sum_{j} \sigma_j^{-}\sigma_{j+2}^{+}+\sigma_j^{-}\sigma_{j+2}^{-}+(1-n_j)$, $\sigma_j^{-}\vcentcolon =|+)(-|_j$, $\sigma_j^{+}\vcentcolon =|-)(+|_j$, and $n_j\vcentcolon =|-)(-|_j$.
The effective Lindbladian in the diagonal subspace therefore becomes the famous ``reaction-diffusion'' process~\cite{lushnikov1986binary,PhysRevB.103.174305}, which can be mapped into a quadratic fermion problem using the Jordan-Wigner transformation. 

The spectrum of $M$ for periodic boundary conditions can be further solved via Fourier transformation (see Appendix~\ref{app_sec:fermion} for more details), giving us the single-particle ``excitation'' dispersion $\mathcal{E}_{k}=(1-\cos k)-i\sin(k)$, where $k \sim 1/N$ for large system sizes. 
We therefore see that the Lindbladian gaps scale as $\sim N^{-2}$ in both even and odd parity sectors.
Furthermore, the mixing time, or the time scale for a system to reach the steady state, will scale as $\sim N^2$, which is also expected from the ``reaction-diffusion'' dynamics.

On the other hand, for open boundary conditions, the corresponding quadratic fermion problem is not diagonalizable but can be put into a Jordan canonical form, which has an eigenvalue one, and the size of the Jordan block scales as $N$. 
Such a Jordan block indeed suggests that some initial state (in fact, a state with the particle on the left end) will take time $t \sim \Omega(N)$ to approach the steady state. 
Interestingly, even if we modify our parent Lindbladian so that the particle can hop both to the left and the right directions in the effective dynamics (such that the corresponding quadratic fermion problem becomes diagonalizable), the gapped nature of the Lindbladian and the $\Omega(N)$ relaxation time still hold.
Though in this case, the $\Omega(N)$ relaxation time is due to the exponential localization of the steady state towards one end of the system, with the particular side of localization set by the relative amplitudes between the leftward and rightward hopping rates. 
Our Lindbladian, under open boundary conditions, therefore provides an example wherein the mixing time does not necessarily scale as the inverse of the Lindbladian gap.

\subsection{Conserved operators}
We have discussed at length the steady states of $\mathcal{L}_\Cl$, which are its right eigenvectors with eigenvalue zero when taken as a superoperator. However, $\mathcal{L}_\Cl$ is not a Hermitian superoperator, so its left and right eigenvectors need not coincide. Note that the left eigenvectors may also be thought of as the right eigenvectors of $\mathcal{L}_\Cl^\dagger$, wherein the superoperator $\mathcal{D}$ appearing in~\cref{eq:parent_lind} is replaced with 
\begin{align}\label{eq:D_dagger}
    \mathcal{D}^\dagger [A] (\mathcal{O}) = A^{\dag} \mathcal{O} A - \frac{1}{2}\{ A^{\dag} A , \mathcal{O} \}~.
\end{align} The left eigenvectors with eigenvalue zero are known as \textit{conserved operators} because they are invariant under time evolution according to the following equation of motion:
\begin{equation}
\frac{d\mathcal{O}}{dt}=\mathcal{L}^\dagger[\mathcal{O}]=0.
\end{equation}
Using~\cref{eq:D_dagger}, it is clear that the identity matrix $\IId $ and the strong symmetry operator $S$ are conserved, as they each commute with all of the jump operators defined in~\cref{eq:parent_jumps}. If the chain is placed on open boundary conditions, then $Z_1$ and $Z_{2N}$ additionally become conserved operators because the jump operators with which they failed to commute extend beyond the boundaries and are discarded. In total then, the conserved operators are
\begin{align}
  \begin{split}
    \text{PBC}:&\ \IId,\, S,\\
    \text{OBC}:&\ \IId,\, Z_1,\, Z_{2N},\, Z_1Z_{2N},\, S,\ SZ_1,\, SZ_{2N},\,\\
               &\ SZ_1Z_{2N}.
  \end{split}
\end{align}
Notice that there are as many conserved operators in each case as there are steady states, which is required of zero eigenvalues of a non-Hermitian (super)operator.

\section{Perturbed Lindbladian and strong-to-weak SSB}\label{sec_perturbed_Lindbladian}
It has been established that the quantum phase of the ground state(s) of a gapped Hamiltonian is stable against small symmetric local perturbations. Analogously, we ask whether adding a small local symmetric perturbation to the Lindbladian, whose steady state is $\rho_\mathcal{C}$, would result in a new steady state in the same phase (according to~\cref{Def:Phase_of_Matter}) as $\rho_\mathcal{C}$ or not.

We might hope to address this question by treating the Lindbladian as a non-Hermitian operator in the space of density matrices and performing non-Hermitian perturbation theory. However, the conventional stability argument does not extend to this context. One typically argues that states in the SPT ground state manifold can only be connected by symmetric perturbations that extend from one end of the system to the other. Assuming local perturbations, this only occurs at an order of perturbation theory that is extensive in the size of the system. The degeneracy is therefore stable in the thermodynamic limit. This breaks down due to the non-Hermiticity of the Lindbladian. Because the Lindbladian is non-Hermitian, its perturbation series involves matrix elements between left and right eigenvectors. While right eigenvectors cannot be coupled to one another at finite order in perturbation theory, they can be coupled to left eigenvectors. This tells us that perturbation theory can no longer guarantee the stability of the SPT phase to small local symmetric perturbations.

We therefore turn to numerics\footnote{Code to generate most of the results of this paper can be found at \href{https://doi.org/10.5281/zenodo.17451966}{https://doi.org/10.5281/zenodo.17451966}} via density matrix renormalization group (DMRG) to investigate the stability of the steady-state SPT phase. We use the trick of ``operator-state mapping'' from Sec.\ \ref{sec:string_order_parameter}, namely, vectorizing the density matrix and the Lindbladian, resulting in a non-Hermitian matrix which is used as an input into DMRG in the ITensor library~\cite{ITensor}. We discuss the details behind the DMRG implementation in~\cref{app:dmrg}. 
In this construction, we map density matrices to pure states in the doubled Hilbert space (with the ket space denoted by $L$ and the bra space denoted by $R$) by the mapping in~\cref{eq:op_state_correspondence}.
The Lindbladian superoperator with a Hamiltonian $H$ and jump operators $L_i$ becomes
\begin{equation}\label{eq:lind_super_mapping}
\begin{aligned}
\mathbbm{L}=&i \IId  \otimes H^T  -i H\otimes \IId  \\
   & +  \sum_{i} L_i \otimes L_i^* -\frac{1}{2} L_{i}^\dagger L_{i}\otimes\IId -\frac{1}{2} \IId  \otimes \left(L_{i}^\dagger L_{i}\right)^T.
\end{aligned}
\end{equation}

Recall that in the (pure) cluster state scenario, we can use $U_{\text{CZ}}$ to map the state and its parent Hamiltonian from the nontrivial SPT state $|\Psi_{\mathcal{C}}\ra$ to the trivial SPT state $|+\ra^{\otimes N}$.
It is then natural to study the phase diagram interpolating between the nontrivial-SPT Hamiltonian and the trivial-SPT Hamiltonian~\cite{PhysRevLett.93.056402,PhysRevLett.103.020506}.

Motivated by the above consideration, here we also study the steady-state phase of the perturbed Lindbladian for $\lambda \in [0, 1]$:
\begin{align}
\label{eq:interpolating_lind}
\mathcal{L}_\lambda = (1-\lambda) \mathcal{L}_{\mathcal{C}} + \lambda \tilde{\mathcal{L}}_{\mathcal{C}}
,
\end{align}
interpolating between the nontrivial SPT Lindbladian $\mathcal{L}_{\mathcal{C}}$ [see~\cref{eq:parent_lind}] and the trivial-SPT Lindbladian $\tilde{\mathcal{L}}_{\mathcal{C}}$ [see~\cref{eq:dual_jumps}].
We study the steady state of the interpolating Lindbladian $\mathcal{L}_\lambda$ using DMRG and calculate various steady state properties, namely the steady-state degeneracy, string order parameters and two other correlators that diagnose strong-to-weak spontaneous symmetry breaking. If there were only one phase transition between the trivial and non-trivial SPT phases, we would expect this transition to take place at the self-dual point $\lambda=1/2$. However, as we will show, the instability to SW-SSB leads to a different situation.

\subsection{Mixed-state SPT order} \label{sec:string_order_numerics}

To examine the nontrivial mixed-state SPT order, we calculate the \renyi-1 and \renyi-2 strong-string order parameters,  $C_{\text{I},nm}^S$ [\cref{eq:def_C_I_S}] and $C_{\text{II},nm}^S$ [\cref{eq:strong_2copy_nontrivial_correlator}], as well as the \renyi-2 weak-string order parameter, $C_{\text{II},nm}^W$ [\cref{eq:weak_2copy_nontrivial_correlator}], for the steady state $\rho_{\lambda}$ of $\cL_{\lambda}$.
Recall that, under conjugation by the CZ circuit, $X_i \mapsto Z_{i-1} X_i Z_{i+1}$ and the string operators map as $\mathcal{S}^S_{nm} \mapsto \tilde{\mathcal{S}}^S_{nm}$ and $\mathcal{S}^W_{nm} \mapsto \tilde{\mathcal{S}}^W_{nm}$. 
We refer to $(C_{\text{I},nm}^S, C_{\text{II},nm}^S, C_{\text{II},nm}^W)$ as the nontrivial string order parameters and $(\tilde{C}_{\text{I},nm}^S, \tilde{C}_{\text{II},nm}^S, \tilde{C}_{\text{II},nm}^W)$ as the trivial string order parameters.

Upon conjugation by $U_{\text{CZ}}$,  $\mathcal{L}_{\lambda} \mapsto \mathcal{L}_{1-\lambda}$. 
Thus the steady states of $\cL_\lambda$ and $\cL_{1-\lambda}$ are related by conjugation with $U_{\text{CZ}}$, or $U_{\text{CZ}} \rho_{\lambda} U_{\text{CZ}}^\dagger = \rho_{1-\lambda}$.
Hence, calculating the nontrivial string order parameters in the steady states of $\cL_\lambda$ for all $\lambda \in [0,1]$ allows us to readily obtain the trivial string order parameters.

For $\lambda=0$, the steady state in the $(s_{\text{ket}}, s_{\text{bra}}, w) = (+1, +1, +1)$ sector is $\rhopp$, which is a nontrivial SPT state having $(C_{\text{I}, nm}^S, C_{\text{II}, nm}^S, C_{\text{II}, nm}^W) = (1,1,1)$  and $(\tilde{C}_{\text{I}, nm}^S, \tilde{C}_{\text{II}, nm}^S, \tilde{C}_{\text{II}, nm}^W) = (0,0,0)$ for any $|n-m|\geq 2$.
On the other hand, for $\lambda = 1$, the steady state in the same symmetry sector is $\tilde{\rho}_{\mathcal{C}}$ [see~\cref{eq:CZ decohered cluster}], and its nontrivial string order parameters are $(0,0,0)$ and its trivial string order parameters are $(1,1,1)$.
Recall that a single string order parameter does not determine if the state is in the nontrivial SPT phase or the trivial phase, rather it is the pattern of the string order parameters that determines the phase~\cite{pollmann2012detection,degroot2022symmetry}.

\begin{figure*}
\captionsetup[subfigure]{labelformat=empty,captionskip=-20pt}
 \centering
 \subfloat[\label{fig:cs1}]{
  \includegraphics[width=0.325\textwidth]{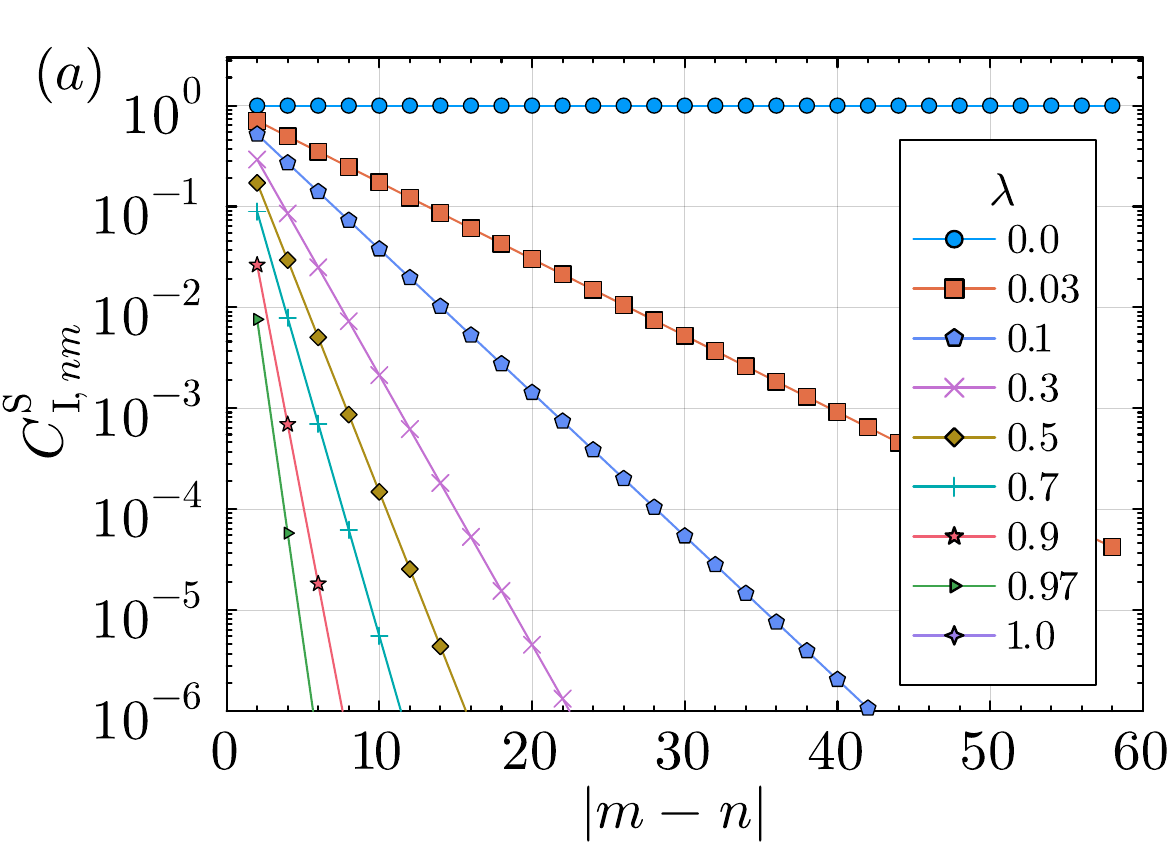}%
}
     \hfill
     \subfloat[\label{fig:cs2}]{
     \includegraphics[width=0.325\textwidth]{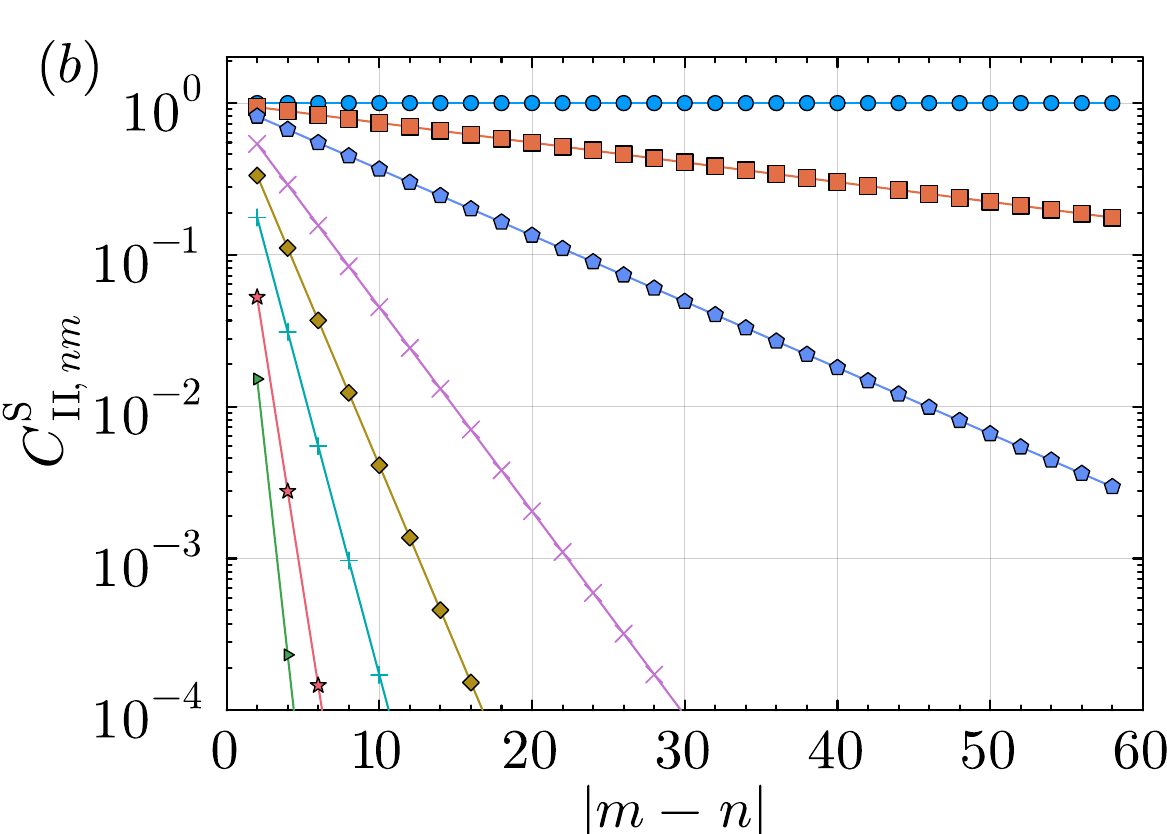}%
     }
    \hfill
     \subfloat[\label{fig:cw2}]{%
     \includegraphics[width=0.325\textwidth]{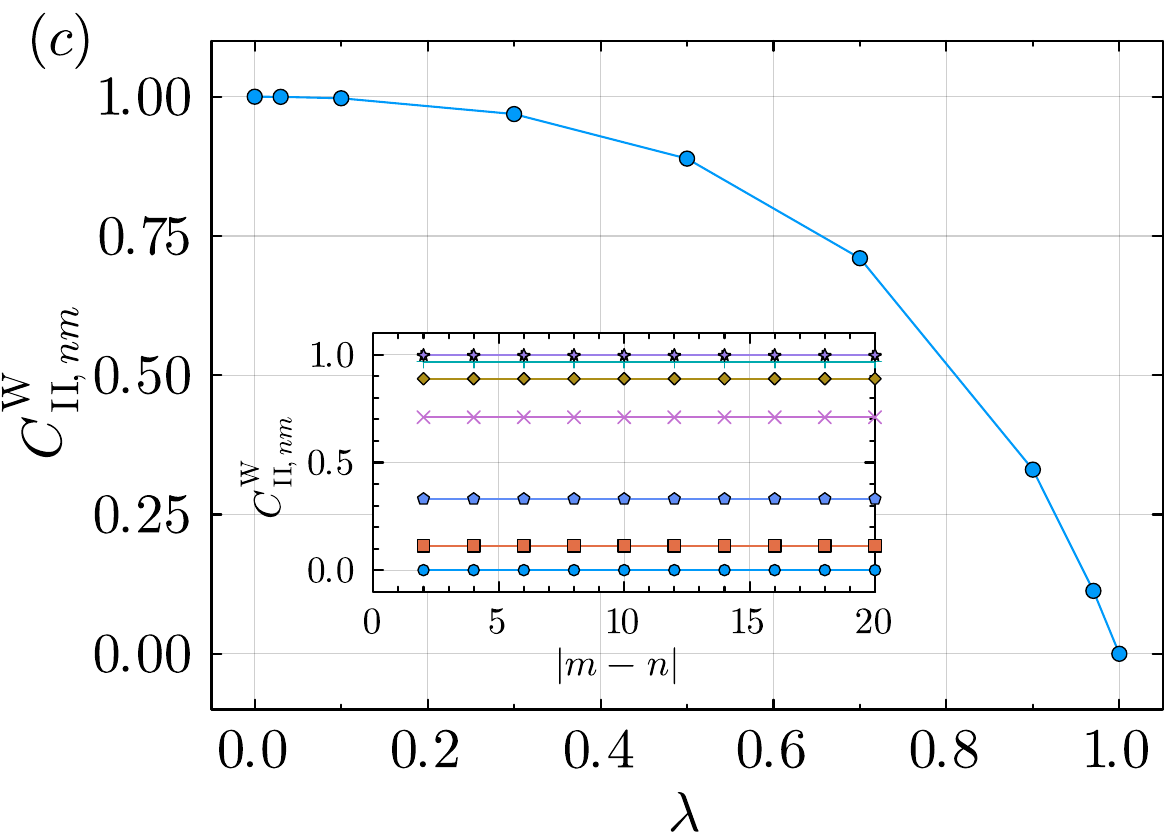}%
     }
\caption{ $(a)$ \renyi-1 strong-string correlator $C_{\text{I}, nm}^S$ as a function of the length of the string. $C_{\text{I}, nm}^S$ decays exponentially with $|m-n|$ for any $ \lambda>0$. The curve for $\lambda = 1$ is not visible within the limits of the plots because $C_{\text{I}, nm}^S=0$. $(b)$ Log of the \renyi-2 strong-string correlator $C_{\text{II}, nm}^S$ as a function of the length of the string. $C_{\text{II}, nm}^S$ decays exponentially with $|m-n|$. The curve for $\lambda = 1$ is not visible within the limits of the plots because $C_{\text{II}, nm}^S=0$. $(c)$ \renyi-2 weak-string correlator $C_{\text{II}, nm}^W$ as a function of the length of the string. $C_{\text{II}, nm}^W$ remains nonzero for all $0 \leq \lambda < 1$. }
    \label{fig:string_dual_pert}
\end{figure*}

If the nontrivial mixed-state SPT order is stable at $\lambda=0$, then we expect that after adding a small perturbation $\tilde{\cL}$, the nontrivial string order parameters will continue to remain nonzero while the trivial string order parameters will continue to remain zero in the limit of long string lengths. 
In~\cref{fig:cs1,fig:cs2}, we show the nontrivial string order parameters evaluated in the steady state in the $ s_{\text{ket}} = s_{\text{bra}} = w = +1$ sector for a system of $2N=400$ qubits for various values of $\lambda$.
In ~\cref{fig:cs1,fig:cs2}, we plot the two strong nontrivial string order parameters $C_{\text{I},nm}^S$ and $C_{\text{II},nm}^S$ for $n=181$ and $m \in \{183, 185, \ldots, 239\}$, which are chosen to minimize the boundary effects.
We find that, for all $\lambda$ other than $\lambda = 0$ and $\lambda = 1$, the strong-string order parameters $C_{\text{I}, nm}^S$ and $C_{\text{II}, nm}^S$ both decay exponentially in the string length $|n-m|$ [see~\cref{fig:cs1,fig:cs2}]. 
For $0 < \lambda <1$, the trivial strong-string order parameters $\tilde{C}_{\text{I},nm}^S$ and $\tilde{C}_{\text{II},nm}^S$ also decay to zero exponentially, which can be inferred from the $U_{\text{CZ}}$ conjugation.

On the other hand, we find that the nontrivial and trivial weak-string order parameters $C_{\text{II}, nm}^W$ and $\tilde{C}_{\text{II}, nm}^W$ are independent of the string length and depend only on $\lambda$ as shown in the inset of~\cref{fig:cw2}. 
The result is calculated with $n=182$ and varying $m \in \{184, 186, \ldots, 202\}$, and we also plot $C_{\text{II}, nm}^W$ as a function of $\lambda$ in~\cref{fig:cw2}.
We see that $C_{\text{II}, nm}^W$ remains nonzero for all $0 \leq \lambda < 1$. 
To summarize, the pattern of string order parameters is $(C_{\text{I}, nm}^S, C_{\text{II}, nm}^S, C_{\text{II}, nm}^W)=(0,0, \text{nonzero})$ and $(\tilde{C}_{\text{I}, nm}^S, \tilde{C}_{\text{II}, nm}^S, \tilde{C}_{\text{II}, nm}^W)=(0,0, \text{nonzero})$ for $0<\lambda<1$.
This matches with the patterns of zeros neither for the nontrivial SPT order nor for the trivial SPT order (see Table~\ref{tab:pattern of zero}). 
This is the first indication that the steady state $\rho_{\lambda}$ at any $0 < \lambda < 1$ cannot be classified according to the mixed-state SPT state classification, suggesting that the state $\rho_{\lambda}$ is neither a nontrivial nor a trivial mixed-state SPT state.

Examining the steady-state degeneracy using DMRG, we find that there is an eight-fold degeneracy for all $\lambda$, where each symmetry sector $s_{\text{ket}}, s_{\text{bra}}, w \in \{+1,-1\}$ hosts one steady state.
In the Hamiltonian case, the ground state degeneracy is robust to small symmetric perturbations. 
However, in the case of steady states of Lindbladians, we find that adding a symmetric perturbation with a single jump operator $X_{2N}$, which acts only on the rightmost qubit, reduces the degeneracy to four. Adding another jump operator $X_1$ further reduces the steady state degeneracy from four to two. 
From~\cref{sec_model}, we know that a steady state degeneracy of two is guaranteed by the strong symmetry $\ztwo^S$, and indeed, the remaining two steady states reside in the symmetry sectors $(s_{\text{ket}}, s_{\text{bra}}, w)=(+1,+1,+1)$ and $(-1,-1,+1)$, respectively.

\subsection{Strong-to-weak SSB}
\label{sec:ssb}
We have seen that the strong-string order parameters all decay exponentially to zero for any $0 < \lambda <1$, while the weak-string order parameter is nonzero for $0\leq \lambda <1$. 
This suggests that the steady state $\rho_{\lambda}$ (or $|\rho_{\lambda}\rangle$) is not short-range entangled.
Furthermore, the trivial strong-string order parameter $\tilde{C}_{\text{I},nm}^{S}$ can also be interpreted as a ``disorder'' order parameter for the strong symmetry $\ztwo^{S}$, whose exponential decay with $|n-m|$ indicates the strong symmetry is spontaneously broken. 

In the following, we present evidence for the \textit{strong to weak} spontaneous symmetry breaking (SSB) of $\ztwo^{S}$ ~\cite{ma2024symmetry,lessa_strong--weak_2024,sala_spontaneous_2024} when $0 < \lambda < 1$. 
We use the following connected correlators to probe strong-to-weak SSB (SW-SSB). For $n$ and $m$ even, the connected correlators are
\begin{equation}
\label{eq:ZZ_correlator}
    A_{\text{I}, nm} \vcentcolon= \langle  Z_n Z_m \rangle  - \langle  Z_n  \rangle  \langle  Z_m \rangle~,
\end{equation}
\begin{equation}
\label{eq:ZZII_correlator}
    A_{\text{II}, nm} \vcentcolon= \langlee  Z_n Z_m \otimes \IId   \ranglee  - \langlee  Z_n  \otimes \IId   \ranglee  \langlee  Z_m \otimes \IId   \ranglee~,
\end{equation}
\begin{equation}
\label{eq:ZZZZ_correlator}
    B_{\text{II}, nm} \vcentcolon= \langlee  Z_n Z_m \otimes Z_n Z_m  \ranglee  - \langlee  Z_n  \otimes Z_n  \ranglee  \langlee  Z_m \otimes Z_m  \ranglee~,
\end{equation}
where $\langle \mathcal{M} \rangle = \Tr \left[ \rho \mathcal{M} \right]$ for a density matrix normalized as $\Tr \left[ \rho \right] = 1$, $\la\la \cdots \ra\ra$ is defined in Eq.~(\ref{eqn:double braket expectation}), and the first (second) factor of the tensor product acts on the ket (bra) Hilbert space. For convenience, we provide a summary of these correlators and the string order parameters introduced in~\cref{sec:string_order_parameter} in~\cref{tab:correlator_summary}. Their values in the various phases studied in this work are given in~\cref{tab:pattern of zero}.

To see how these correlators probe strong-to-weak SSB, recall that the strong symmetry $\ztwo^S = \ztwo^{\text{bra}} \times \ztwo^{\text{ket}} = \ztwo^{\text{L}} \times \ztwo^{\text{R}}$ is in fact two $\ztwo$ symmetries with the symmetry generator $S$ operating on the bra or the ket side independently. In the vectorized $|\rho \rangle$ perspective (namely, doubled Hilbert space), it is generated by $S_L\otimes \IId _R$ and $\IId _L\otimes S_R$.  
The possible symmetry breaking patterns are $\ztwo^S \rightarrow \mathbb{Z}_1$ (broken down to the trivial group) or $\ztwo^S \rightarrow \ztwo^W$, where $\ztwo^W=\{I, S(\cdot)S^{\dagger}\}=\{I, S_L \otimes S_R\}$ (namely, broken down to a weak symmetry). 

We therefore see that the above correlators have the following physical meaning.
Since the superoperator $Z_n (\cdot)\IId $ is charged under (namely, anti-commutes with) $S(\cdot)\IId $ and $S(\cdot)S^{\dagger}$, we see that $A_{\text{I}, nm}$ is used to detect $\ztwo^S \rightarrow \mathbb{Z}_1$ and the corresponding long-range correlation in $\rho$.
Similarly, $Z_n \otimes \IId _R$ anti-commutes with $S_L \otimes \IId _R$ and $S_L\otimes S_R$, so $A_{\text{II}, nm}$ is used to detect $\ztwo^S \rightarrow \mathbb{Z}_1$ and the corresponding long-range correlation in $|\rho \rangle$, namely, a long-range correlation quadratic in $\rho$. 
On the other hand, $Z_n \otimes Z_n$ anti-commutes with $S_L \otimes \IId _R$ (and $\IId _L \otimes S_R$) but commutes with $S_L\otimes S_R$, so $B_{\text{II}, nm}$ is used to detect $\ztwo^S \rightarrow \ztwo^W$ (but cannot detect if $\ztwo^W$ is further broken) and the corresponding long-range correlation in $|\rho\ra$. 
These patterns of zeros are summarized in Table~\ref{tab:pattern of zero}.
We note that, in Ref.~\cite{lessa_strong--weak_2024}, a fidelity correlator has been proposed as another more robust diagnostic for the strong-to-weak SSB.

\begin{figure}[htbp]
     \centering
     \includegraphics[width=0.99\columnwidth]{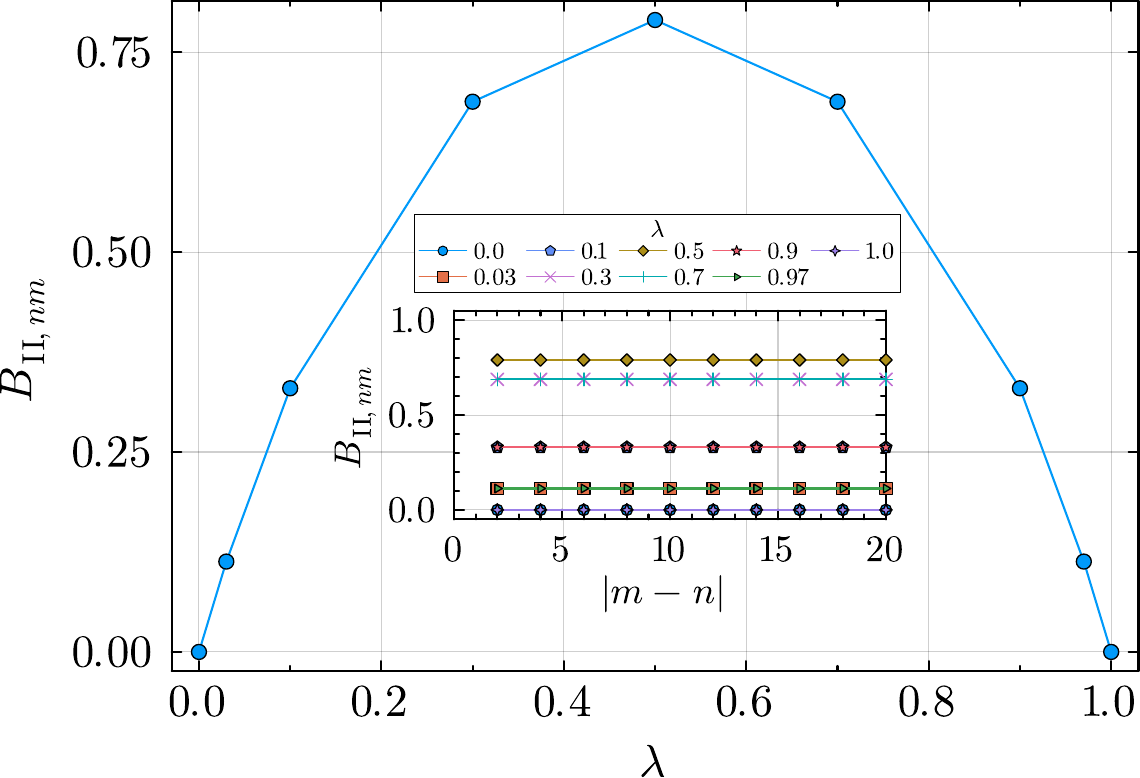}
     \caption{Connected correlator $B_{\text{II}, nm}$ for the Lindbladian that interpolates between the parent Lindbladian $\mathcal{L}_{\mathcal{C}}$ and its CZ dual, $\tilde{\mathcal{L}}_{\mathcal{C}}$ shown in Eq.~\eqref{eq:interpolating_lind}. We find that this correlator is nonzero for $0 < \lambda < 1$. This, taken together with $A_{\text{I}, nm} = A_{\text{II}, nm} = 0$, indicates that the strong symmetry on the even sites is broken to the weak symmetry. In the inset, we show $B_{\text{II}, nm}$ as a function of $|m-n|$ for various values of $\lambda$. We find that $B_{\text{II}, nm}$ is independent of $|m-n|$, but depends on $\lambda$, as shown in the main figure.
     }
     \label{fig:b2_lambda}
\end{figure}

We calculated the correlators $A_{\text{I}, nm}$,  $A_{\text{II}, nm}$, and $B_{\text{II}, nm}$ of the steady state $\rho_\lambda$ in the $s_{\text{ket}}=s_{\text{bra}}=w=+1$ symmetry sector for a system size $2N = 400$ for various perturbation strengths $\lambda$, fixing $n=190$ and varying $m \in \{192, 194, \ldots, 210 \}$. 
We find that $A_{\text{I}, nm} = A_{\text{II}, nm} = 0$ (not shown in the figures), while $B_{\text{II}, nm}$ is a nonzero constant, independent of $\abs{m-n}$ for all $0 < \lambda < 1$. 
In~\cref{fig:b2_lambda}, we plot $B_{\text{II}, nm}$ as a function of $\lambda$. $B_{\text{II},nm}$ is symmetric about $\lambda=1/2$ as expected since the operator $Z_nZ_m$ is invariant under $U_{\text{CZ}}$ conjugation while $\rho_\lambda$ maps to $\rho_{1-\lambda}$. 
$A_{\text{I}, nm} = 0$ and $A_{\text{II}, nm} = 0$ imply that the weak symmetry on the even sites $\mathbb{Z}_2^W$ (which is a subgroup of the strong symmetry on even sites $\mathbb{Z}_2^S$) is not broken. 
On the other hand, $B_{\text{II}, nm} \neq 0$ for $0<\lambda <1$ implies that the strong symmetry on even sites is broken.
Taken these two facts together, we conclude that the $\tilde{\mathcal{L}}$ perturbation leads to a strong-to-weak SSB $\ztwo^S \rightarrow \ztwo^W$ on the even sites.

We also study a more general parametrization of the Lindbladians given by
\begin{equation}\label{eqn: L012}
    \mathcal{L} (\lambda_0, \lambda_1, \lambda_2) \vcentcolon = \sum_{k \in \{0,1,2\}} (1-\lambda_k) \mathcal{L}_k + \lambda_k \tilde{\mathcal{L}}_k~,
\end{equation}
where $\mathcal{L}_k$ is composed of the jump operators $L_{k,j}$ [see~\cref{eq:parent_lind}], and $\tilde{\mathcal{L}}_k $ is composed of the jump operators $\tilde{L}_{k,j}$ [see~\cref{eq:dual_jumps}].
Note that $\mathcal{L}_{\lambda} = \mathcal{L}(\lambda, \lambda, \lambda)$.
We calculate the steady state of $\mathcal{L}(\lambda_0, \lambda_1 = \lambda_0, \lambda_2)$ for $\lambda_0, \lambda_2 \in [0,1]$ in the $s_{\text{ket}}= s_{\text{bra}}= w = 1$ symmetry sector, obtaining the six string order parameters and the three connected correlators that probe SW-SSB in the steady states.
This gives us the phase diagram shown in~\cref{fig:phase_diagram}.
We find that the steady state of $\mathcal{L}(\lambda_0, \lambda_0, 0)$ is in the nontrivial SPT phase for all $\lambda_0 \in [0,1)$.   
We discuss this in more detail in the next section (\cref{sec:weak_defects}).

In addition to the $\tilde{\mathcal{L}}_{\mathcal{C}}$ perturbation, we have also considered a Hamiltonian perturbation with the Hamiltonian proportional to $\sum_{i} Z_{2i} Z_{2i+2}$ and Lindbladian perturbations with jump operators $Z_{2i}Z_{2i+2}$, $X_{2i+1}$, and $\ket{-}\bra{+}_{2i} \otimes \ket{-} \bra{+}_{2i+2}$. 
In all of these cases, we find that both $C_{\text{I}, nm}^S$ and $C_{\text{II}, nm}^S$ decay exponentially to zero, $C_{\text{II}, nm}^W$ is a nonzero constant  independent of $|m-n|$, $A_{\text{I}, nm} = 0$, $A_{\text{II}, nm} = 0$, and $B_{\text{II}, nm}$ approaches a nonzero constant as $\abs{n-m}$ is increased. We show these plots in~\cref{app:other_perturbations} (\cref{fig:other_pert,fig:b2_other_pert}).
This suggests that SW-SSB of the strong symmetry $\ztwo^S$ on the even sites is a generic instability of the steady-state mixed-state SPT under our consideration.

\begin{figure}[htbp]
     \centering
     \includegraphics[width=0.8\columnwidth]{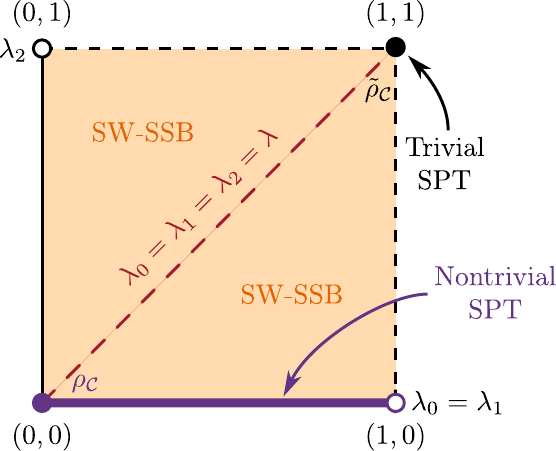}
     \caption{Phase diagram of the steady state of the Lindbladian $\mathcal{L}(\lambda_0,\lambda_1,\lambda_2)$ given in~\cref{eqn: L012}. The nontrivial SPT order is stable against weak-string defects, introduced by the perturbation $\tilde{L}_{0}$. However, introducing any strong-string defects such as $\tilde{L}_2$ leads to strong-to-weak SSB. Note the brown dashed line is the interpolated Lindbladian $\mathcal{L}_{\lambda}$ in \cref{eq:interpolating_lind}. } 
     \label{fig:phase_diagram}
\end{figure}

We remark that, because of this strong-to-weak SSB, the steady state in the $+1$ eigenvalue sector of $S$ is not an injective matrix-product state. Thus the procedure to extract the projective representation described in Ref.~\cite{pollmann2012detection} is not applicable.

\begin{table*}[t]
    \centering
    \renewcommand{\arraystretch}{1.35}
    \begin{tabular}{|c|c|c|}
    \hline
        Symbol & Defined In & Diagnoses \\\hline
        {$A_{\text{I}/\text{II},nm}$} & {\cref{eq:ZZ_correlator,eq:ZZII_correlator}} & {$\ztwo^S\times \ztwo^W \to \mathbb{Z}_1\times \ztwo^W$ SSB} \\
        {$B_{\text{II},nm}$} & {\cref{eq:ZZZZ_correlator}} & {$\ztwo^S\times \ztwo^W \to \mathbb{Z}^W_2\times \ztwo^W$ SSB}\\
        {$C^{S}_{\text{I}/\text{II},nm}$} & {\cref{eq:def_C_I_S,eq:strong_2copy_nontrivial_correlator}} & {Mixed-state SPT} \\
        {$C^{W}_{\text{II},nm}$} & {\cref{eq:weak_2copy_nontrivial_correlator}} & {Mixed-state SPT}\\\hline
    \end{tabular}
    \caption{Summary of the various correlators and the orders they detect. The subscript I/II indicates whether the correlator uses quantities that are linear or quadratic in the density matrix, where applicable.}\label{tab:correlator_summary}
\end{table*}

\begin{table*}
    \centering 
    \renewcommand{\arraystretch}{1.35}
    \begin{tabular}{|c|ccc|cccccc|}
        \hline Phase & $A_{\text{I}}$ & $A_{\text{II}}$ & $B_{\text{II}}$ & $C^S_{\text{I}}$ & $C^S_{\text{II}}$ & $C^W_{\text{II}}$ & $\tilde C^S_{\text{I}}$ & $\tilde C^S_{\text{II}}$ & $\tilde C^W_{\text{II}}$ \\\hline
        $\ztwo^S\times\ztwo^W$ trivial SPT &0&0&0&0&0&0&$+$&$+$&$+$\\
        $\ztwo^S\times\ztwo^W$ SPT &0&0&0&$+$&$+$&$+$&0&0&0\\
        $\ztwo^W\times\ztwo^W$ SSB &0&0&$+$&0&0&$+$&0&0&$+$\\
        $\mathbb{Z}_1\times\ztwo^W$ SSB &$+$&$+$&$+$&0&0&$+$&0&0&$+$\\
        \hline
    \end{tabular}
    \caption{Summary of the $\ztwo^S\times\ztwo^W$-symmetric phases discussed in this paper, along with the values of the correlators in each of the phases. The ``$+$'' sign indicates that the correlator saturates to a nonzero value in a given phase, while the ``$0$'' indicates that the correlator is zero. The correlators to the left of the dividing line are connected correlators which diagnose SSB, while the correlators to the right of the line are string order parameters which together form the ``pattern of zeros'' diagnosing SPT order. 
    }
    \label{tab:pattern of zero}
\end{table*}

\subsection{Stability against weak-string-order defects}\label{sec:weak_defects}

We have seen that adding $\tilde{\cL}_{\mathcal{C}}$ results in the destruction of the nontrivial SPT order. 
The instability of the SPT order as a steady state might be expected for the following reason. 
Recall that, in the dual picture, we can understand the establishment of the strong-string order by the reaction-diffusion process. That is, it would take $\Omega(N^2)$ time to annihilate the strong-string order defects (namely, the local configurations of $Z_{j-1}X_jZ_{j+1}=-1$ for even $j$). 
Therefore, we expect that, if there is a term in the Lindbladian effectively introducing the strong-string order defects at a constant rate, then the strong-string order in the steady-state will be destroyed. 
Indeed, we see that the jump operator $\tilde{L}_{2j}$ [defined in~\cref{eq:dual_jumps}] effectively introduces the strong-string order defects. 
We note that a recent work \cite{chirame_stable_2024} proposed to alleviate the above instability by heralding the strong-string defects.
However, the unexpected and surprising result of our work is that the destruction of the nontrivial SPT order leads to strong-to-weak symmetry breaking, instead of going to a trivial mixed-state SPT phase. 

On the other hand, we ask if the steady-state mixed-state SPT order can be stable against weak-string-order defects. To answer this question, we examine the Lindbladian $\mathcal{L}(\lambda_0,\lambda_1,\lambda_2)$ in Eq.~(\ref{eqn: L012}) with the parameters on the purple line of~\cref{fig:phase_diagram}.
This can be parameterized as $\mathcal{L}(\lambda_0, \lambda_0, 0)$ for $\lambda_0 \in [0,1]$, corresponding to the perturbations consisting of jump operators $\tilde{L}_{0,j}$ on even $j$ and $\tilde{L}_{1,j}$ on odd $j$, with an equal perturbation strength $\lambda_0 =\lambda_1$.
For $\lambda_0 = 0$, the steady state is $\rho_{\mathcal{C}}$ given in~\cref{eq:decohered_cluster_state}. 
The jump operators $\tilde{L}_{0, j}$ with $j \in \text{even}$ and $\tilde{L}_{1,j}$ with $j \in \text{odd}$ introduce domain wall defects on the odd sites, so we expect that the strong-string order parameters $C_{\text{I}, nm}^S$ and $C_{\text{II}, nm}^S$ will remain unity in the perturbed steady states, and the weak-string order parameter $C_{\text{II}, nm}^W$ will change. 
As shown in~\cref{fig:cw2_weak_symm_defects}, $C_{\text{II}, nm}^W$ remains a nonzero constant for all values of $\lambda_0 \in [0,1]$, where the calculation is done for a system of $2N = 400$ qubits with the steady state of $\mathcal{L}(\lambda_0, \lambda_0, 0)$ in the $s_{\text{ket}} = s_{\text{bra}} = w = +1$ symmetry sector.
We fix $n=182$ and vary $m \in \{ 184, 186, \ldots, 202\}$, finding that $C_{\text{II}, nm}^W$ is independent of $|m-n|$ for all $\lambda_0 \in [0,1]$.
\begin{figure}[htbp]
     \centering
     \includegraphics[width=0.9\columnwidth]{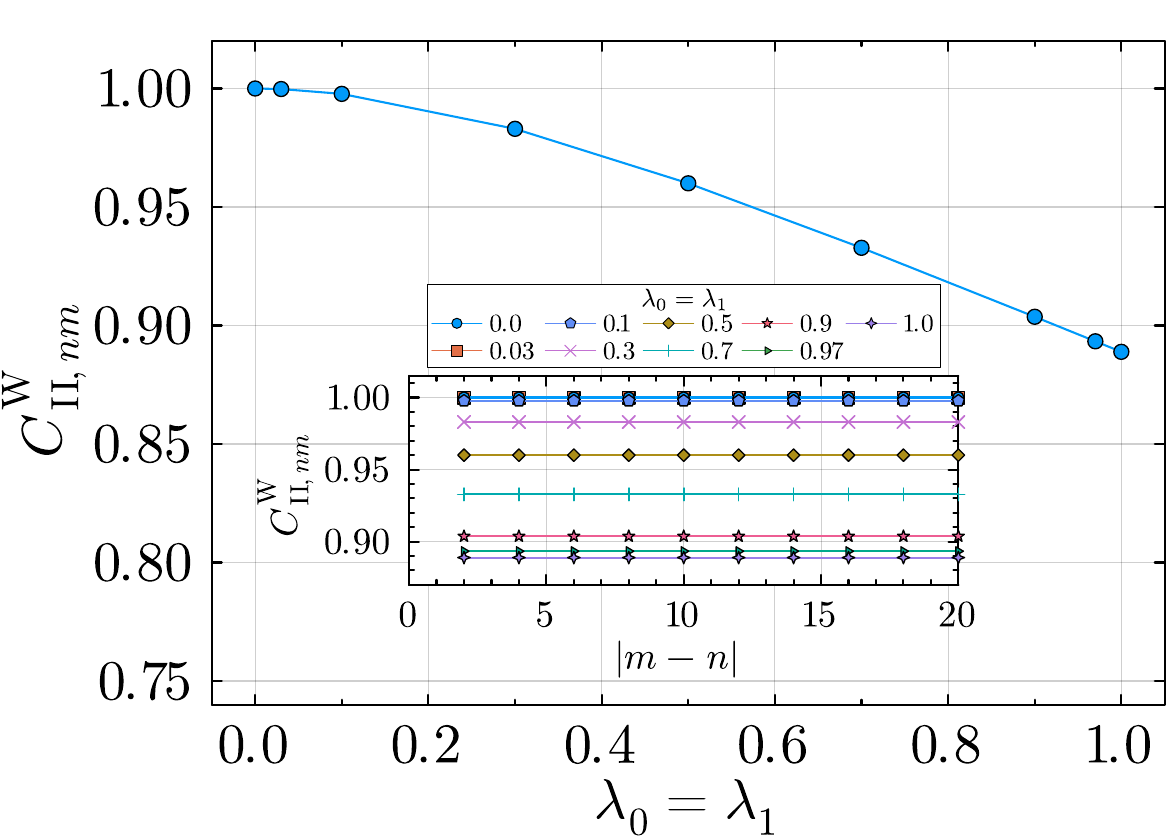}
     \caption{The weak-string order parameter $C_{\text{II}, nm}^W$ for the steady state of Lindbladian $\mathcal{L}(\lambda_0,\lambda_0,0)$ [Eq.~(\ref{eqn: L012})]. While not shown in the figure, the strong-string order parameters $C_{\text{I}, nm}^S=C_{\text{II}, nm}^S=1$. Therefore the steady state exhibits nontrivial mixed-state SPT.
     }
     \label{fig:cw2_weak_symm_defects}
\end{figure}

For $\lambda_0 \in [0,1)$, we also calculate the nontrivial strong-string order parameters $C_{\text{I}, nm}^S$ and $C_{\text{II}, nm}^S$, finding that they are one for all $\lambda_0 \in [0,1) $ and for all string lengths as expected. 
The trivial strong-string correlators $\tilde{C}_{\text{I}, nm}^S$ and $\tilde{C}_{\text{II}, nm}^S$ decay exponentially to zero while $\tilde{C}_{\text{II}, nm}^W=0$ for all $|m-n|$ and $\lambda_0 \in [0,1)$.
This pattern of string order parameters matches with the nontrivial mixed-state SPT as summarized in Table~\ref{tab:pattern of zero}.
We also find that $A_{\text{I}, nm} = A_{\text{II}, nm} = B_{\text{II}, nm}= 0$, indicating that there is no SSB.
Thus we find that the mixed-state SPT is stable to perturbations which introduce only weak-string defects.

We comment that $\mathcal{L}(\lambda_0,\lambda_0,0)$ with $\lambda_0 = 1$ is a pathological point, since the steady-state degeneracy is exponential is $N$.
We left such pathological cases as undefined, labeled with open circles in the phase diagram~\cref{fig:phase_diagram}.

\subsection{Exactly solvable perturbations}
\label{sec:exactly_solvable_perturbation}

As discussed in~\cref{sec:lindblad_mixing}, the dynamics generated by the unperturbed Lindbladian $\tilde{\mathcal{L}}_{\Cl}$ in the CZ dual picture [c.f.~\cref{eq:dual_jumps}] can be decomposed into population dynamics (namely, the diagonal part of the density matrix in the $X_{2j}$ eigenbasis) on the even sites, the effect of which can be mapped to a non-Hermitian free fermion Hamiltonian, as well as dephasing on all the odd sites.
This motivates the possibility of considering the perturbations which are exactly solvable via such a free fermion mapping, potentially furthering our understanding of the instability of the steady-state mixed-state SPT toward the SW-SSB.

As we now show, the free fermion mapping also works for a large class of dissipative perturbations to $\tilde{\mathcal{L}}_{\Cl}$, which can be solved exactly and analytically under periodic boundary conditions. Here, we discuss the manifestation of the SW-SSB, making use of the free fermion analytical solution of the perturbed Lindbladian in those cases. More specifically, we show that, for all perturbing dissipators that can be mapped to a free fermion model via the prescription described in~\cref{sec:lindblad_mixing}, the steady state either stays invariant under the perturbation or leads to a SW-SSB. 

Before discussing the mapping of perturbative dissipators to free fermion models, we first note that there exist a large group of dissipators that preserve the strong symmetry and do not affect the steady state of $\tilde{\mathcal{L}}_{\Cl}$. More concretely, from the steady state in Eq.~\eqref{eq:CZ decohered cluster}, we see that any perturbative dissipator whose jump operator $\tilde{L}_{\text{pert},2j}$ only involves a product of $X_{2j}$ on the even sites would not change the steady state. We can write such perturbed Lindbladian explicitly as $\tilde{\mathcal{L}} '_{\Cl}= \tilde{\mathcal{L}}_{\Cl}+  \gamma _{X,\{n _{\ell,j}\}} \mathcal{D} [\tilde{L}_{X,\{n _{\ell,j}\}}]$, with the jump operators $\tilde{L}_{X,\{n _{\ell,j}\}}$ given by
\begin{align}
\tilde{L}_{X,\{n _{\ell,j}\}} = 
\prod _{\ell} X_{2\ell} ^{n _{\ell,j}}
,
\end{align}
where $n _{\ell,j}$ are parameters characterizing the perturbation that take values $n _{\ell,j} =0$ or $1$. 
The jump operator $\tilde{L}_{2,j} = X_j$ in Eq.~\eqref{eq:dual_jumps} provides an example, although there are many more. Furthermore, it is also straightforward to show that, if the jump operator of the perturbation can be written as $\tilde{L}'_{j} = \prod _{\ell} X_{2\ell} ^{n _{\ell,j}} + M_{j} $ for any operator $M_j$, then the actions of $\tilde{L}_{\text{pert},j} $ and $M_{j} $ within the population sector (with respect to the $X_{2j}$ eigenbasis) are fully equivalent. We thus focus on jump operators that explicitly involve other Pauli operators.

We now consider possible dissipative perturbations that can be mapped to a free fermion non-Hermitian Hamiltonian under the correspondence $\rho \rightarrow |\rho ) \vcentcolon = \sum_{\boldsymbol{\sigma}} \rho(\boldsymbol{\sigma}) |\boldsymbol{\sigma})$ in~\cref{sec:lindblad_mixing}. In this case, it is straightforward to see that the jump operator can involve at most two spin operators acting on nearest neighbors. One can thus use this constraint along with the symmetry requirement to show that any jump operator that could be mapped to a free fermion Hamiltonian [after performing the Jordan-Wigner transformation in the space spanned by $|\rho )$] should take the following form: 
\begin{align}
\tilde{L}_{\text{pert},2j} = 
(Z _{2j} + \alpha _{j} Y _{2j} )
(Z _{2j+2} + \beta _{j} Y _{2j+2} )
,
\end{align}
which ensures that $\tilde{L}_{\text{pert},2j} \rho \tilde{L}_{\text{pert},2j} ^{\dag}$ corresponds to a quadratic fermion Hamiltonian in the $|\rho )$ space. Moreover, we  require the remaining term of the dissipator $\frac{1}{2} \{ \tilde{L}_{\text{pert},2j} ^{\dag}\tilde{L}_{\text{pert},2j}, \rho \} $ to also map to a free fermion Hamiltonian, which means that $\tilde{L}_{\text{pert},2j} ^{\dag}\tilde{L}_{\text{pert},2j} = (1+ |\alpha _{j} |^{2} + \text{Im}\;\alpha _{j}  X _{2j})(1+ |\beta _{j} |^{2} + \text{Im}\;\beta _{j}  X _{2j+2})$ cannot involve a nontrivial contribution from $X _{2j} X _{2j+2}$. We thus prove that any perturbation dissipator that can be mapped to free fermion dynamics should have a jump operator of the following form:
\begin{align}
\label{eq:fermion_pert}
\tilde{L}_{\text{pert},2j} = 
&Z _{2j} Z _{2j+2} (1+\alpha _{j} X_{2j}) 
\nonumber \\
&\text{  or  }
Z _{2j} Z _{2j+2} (1+\beta _{j} X_{2j+2})
, 
\end{align}
where $\alpha _{j}$ or $\beta _{j}$ are arbitrary complex coefficients.

As shown in~\cref{app_sec:fermion}, the effect of perturbation generated by Eq.~\eqref{eq:fermion_pert} on the steady state can always be rewritten as a rescaling of the unperturbed Lindbladian as well as a simpler perturbation $\tilde{L}_{ZZ,2j}=Z_{2j-2} Z_{2j}$. Note that, for perturbations with the specific jump operators $\tilde{L}_{ZZ,2j}=Z_{2j-2} Z_{2j}$, the dual CZ circuit does not change the form of the jump operator, so that we have  
\begin{align}
& \tilde{\mathcal{L}} '_{\Cl}= \tilde{\mathcal{L}}_{\Cl}+  \gamma _{ZZ,2j} \mathcal{D} [\tilde{L}_{ZZ,2j}]
\nonumber \\
\Leftrightarrow & 
\mathcal{L} '_{\Cl}= \mathcal{L} _{\Cl}+  \gamma _{ZZ,2j} \mathcal{D} [Z_{2j-2} Z_{2j}]
,
\end{align}
where the unperturbed Lindbladian $\mathcal{L} _{\Cl}$ without the dual CZ circuit is given by Eq.~\eqref{eq:parent_lind}. 
As such, it suffices to only consider the perturbation $L_{ZZ,2j}=Z_{2j-2} Z_{2j}$ and its effect on the steady-state mixed-state SPT order. Intriguingly, as far as the strong-symmetry string order correlators are concerned, the effect of this perturbation (in the original frame without the CZ circuit) is equivalent to adding a perturbation with the jump operator replaced by $L_{2,2j-1}=Z_{2j-2}X_{2j-1}Z_{2j}$ [see Eq.~\eqref{eq:parent_jumps}].

When the aforementioned perturbation preserves translational invariance, i.e.~when $\gamma_{ZZ,2j} =\gamma _{ZZ} $, we can solve the free fermion non-Hermitian Hamiltonian analytically. The perturbed Lindbladian can be written as 
\begin{align}
\tilde{\mathcal{L}} '_{\Cl}= \tilde{\mathcal{L}}_{\Cl}+  \gamma _{ZZ} \mathcal{D} [\tilde{L}_{ZZ,2j}]
.
\end{align}
In this case, we can explicitly compute the \renyi-2 strong-string order parameter $C_{\text{II},nm}^S $, as well as the connected correlator $B_{\text{II}, nm} $ in the perturbed state as 
\begin{align}
C_{\text{II},nm}^S & = 
\langlee \mathcal{S}_{nm}^S \otimes \IId  \ranglee_\rho = 
\left ( \frac{1}{1+2\gamma _{ZZ} } \right ) ^{|m-n|}
, \\
& B_{\text{II}, nm}  
= 1- \frac{1  }{(1+2\gamma _{ZZ} ) ^{2}}
. 
\end{align}
We thus see that any free-fermion-type perturbation of the form given in Eq.~\eqref{eq:fermion_pert} will lead to an exponentially decaying strong-string order parameter, as well as a finite connected correlator, signaling a strong-to-weak SSB in the perturbed Lindbladian steady state. 
We conjecture that strong-to-weak SSB is a generic feature of Lindbladian steady states under weak local perturbations (that have a nontrivial impact on the steady state). 

We would like to point out that the exactly solvable mapping only applies to certain types of  perturbations. For a more general type of  perturbations, we refer to Appendix~\ref{app_sec:perturbation}, where we discuss  generic perturbation theory which can capture the physics at small perturbation strengths.

\section{Clifford circuit realization}\label{sec_clifford_simulation}

So far, we have discussed the steady states of a dynamical open system modeled by a Lindbladian. 
While there are numerous methods to simulate  Lindbladian dynamics on a quantum computer or a quantum simulator~\cite{kliesch2011dissipative,an2023linear,purkayastha2021periodically}, here we propose to replicate the essential physics via a local quantum channel instead of trying to simulate the Lindbladian. 
Furthermore, we construct a quantum channel that can be realized using Clifford gates, Pauli measurements, and feedback, allowing us to simulate the quantum dynamics efficiently on a classical computer. 
It is straightforward to add operations such that the dynamics is no longer efficiently simulable on a classical computer. 
We note that all Clifford simulations performed in this work were done using the Python package \emph{stim}~\cite{gidney2021stim:-a-fast-st}.

Recall that the parent Lindbladian $\mathcal{L}_{\mathcal{C}}$ consists of the jump operators that stabilize the domain-wall configurations ($L_{0,2j}$), decohere them ($L_{1,2j-1}$), and make their probabilities uniform ($L_{2,2j-1}$).
These actions can be realized by the quantum channel that we now describe.

To stabilize the domain-wall configuration, we use the following measurements and feedback. 
If $j$ is even, we measure $Z_{j-1}X_jZ_{j+1}$; if the measurement outcome is $-1$, we apply the unitary $X_{j+1}$. This can be written as a quantum channel with the following Kraus representation:
\begin{equation}
    \mathcal{E}_{2j}(\rho)=\sum_{\mu=\pm}K_{\mu,2j}\rho K^{\dagger}_{\mu,2j}~,
\end{equation}
where $K_{+,2j}=\frac{1}{2}(\IId  +  Z_{2j-1} X_{2j} Z_{2j+1})$, $K_{-,2j}=\frac{1}{2}X_{2j+1}(\IId  -  Z_{2j-1} X_{2j} Z_{2j+1})$, and $j \in \{1,2,\ldots, N\}$.

The action of $L_{1,2j-1}$ and $L_{2,2j-1}$ can be mimicked by the following action. If $j$ is odd, we measure $Z_{j}$ and then apply the unitary $Z_{j-1}X_jZ_{j+1}$ regardless of the measurement outcome, resulting in the quantum channel with a Kraus representation
\begin{equation}
    \mathcal{E}_{2j-1}(\rho)=\sum_{\mu=\pm}K_{\mu,2j-1}\rho K^{\dagger}_{\mu,2j-1}~,
\end{equation}
where $K_{\pm,2j-1}=Z_{2j-2}X_{2j-1}Z_{2j}\frac{1}{2}(\IId  \pm  Z_{2j-1} )$ for $j \in \{1,2,\ldots, N\}$.
Note that this channel can also be realized by the following action without a measurement, where we apply the unitary $Z_{2j-2}X_{2j-1}Z_{2j}$ or $Z_{2j-2}Y_{2j-1}Z_{2j}$ with an equal probability when $j \in \{1,2,\ldots, N\}$.

With the above building blocks, we consider the following channel:
\begin{equation}
    \mathcal{E}=\frac{1}{2N}\sum_{j=1}^{N}(\mathcal{E}_{2j-1}+\mathcal{E}_{2j})~,
\end{equation}
and we consider the steady states of $\mathcal{E}^{t}$ (i.e.\ $\mathcal{E}$ applied $t$ times) when the number of steps $t \rightarrow \infty$. 
Note that $\mathcal{E}$ has the desired $\ztwo^{S} \times \ztwo^{W}$ symmetry.
The channel $\mathcal{E}$ can be realized by the following procedure. For each step, pick a site $j$ from $1, 2,  \ldots , 2N$ with an equal probability and implement the aforementioned protocol corresponding to $\mathcal{E}_{j}$ depending on $j$ being even or odd.
 $\mathcal{E}^{t}$ corresponds to repeating the above procedure $t$ times, which is also the unit of the time step we use.

\subsection{Mixing time}~\label{sec:clifford_mixing}
The first quantity we examine is the mixing time of the (unperturbed) quantum channel $\mathcal{E}$.
As one can easily verify, $\rhopp$ is a steady state of $\mathcal{E}^{t = \infty}$ in the $(s_{\text{bra}},s_{\text{ket}},w)=(+1,+1,+1)$ sector, whose strong-string order parameter satisfies $C_{\text{I}, nm}^S = 1$.
The mixing time is defined as the supremum, over all initial states, of the time required for the system to approach the/a steady state within a specified accuracy. Determining the supremum over all initial states is computationally difficult to do. 
Since the parent Lindbladian maps to a ``reaction-diffusion'' process (see Section~\ref{sec:lindblad_mixing}) whose mixing time is proportional to $N^2$, we expect the same scaling for the mixing time for the channel defined here.
Since the mixing time reflects how fast the system equilibrates, we instead use the following proxy for mixing time.
We start with an initial pure state $|\psi_0\ra = \bigotimes_{j=1}^{2N}\ket{+}_j$ in which $C_{\text{I}, nm}^S = 0$. 
The repeated application of the channel $\mathcal{E}$ grows $C_{\text{I}, nm}^S$ from $0$ to $1$.
For a given threshold $\eta$ that is close to $1$, we define the proxy mixing time to be the time (number of steps) that it takes for $C_{\text{I}, nm}^S$ to grow from 0 to $\eta$.
Note that the true mixing time does not depend on an observable, however the proxy for the mixing time here depends on the observable, which in our case is $C_{\text{I}, nm}^S$.
\begin{figure}[tb]
\captionsetup[subfigure]{labelformat=empty,captionskip=-20pt}
     \centering
     \hfill
 \subfloat[\label{fig:mixing_time}]{
  \includegraphics[width=0.47\columnwidth]{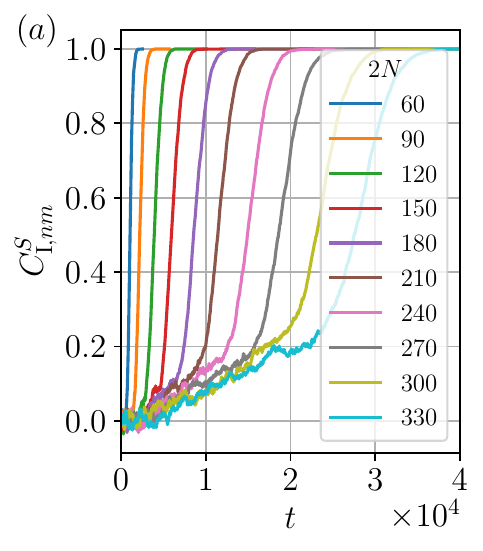}
}
     \hfill
 \subfloat[\label{fig:mixing_time_quadradic}]{%
  \includegraphics[width=0.47\columnwidth]{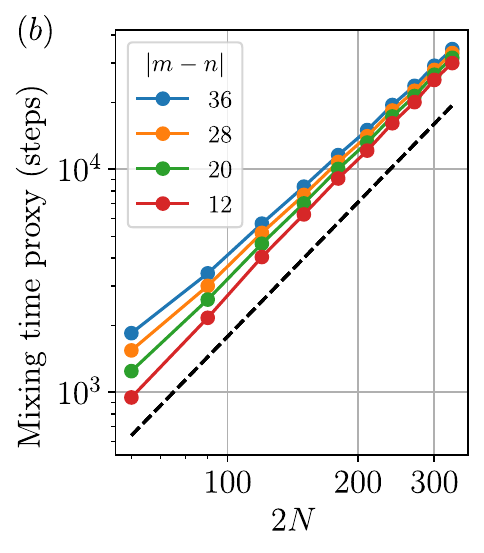}
}
     \hfill
        \caption{$(a)$ Growth of the string correlator $C_{\text{I}, nm}^S$ as a function of the time step $t$ for string length $|m-n| = 28$. The initial state is $\ket{+}^{\otimes 2N}$ while the state approaches $\rhopp$ as $t \rightarrow \infty$. Legend indicates the total system sizes $2N$.
        $(b)$ A log-log plot of the proxy for the mixing time (time taken for $C_{\text{I}, nm}^S$ to grow to $\eta = 0.95$ starting from the initial state $\ket{+}^{\otimes N}$) as a function of the number of qubits $2N$ for various string lengths. The 
        legend indicates the length of the string $|m-n|$, and the black dashed line has a slope of two described by $\text{time} \propto N^2$. We see that the curves for all the shown string lengths approach a line parallel to the black dashed line asymptotically.
        }
        \label{fig:mixing_time_combined}
\end{figure}
Picking $\eta = 0.95$, in~\cref{fig:mixing_time}, we plot $C_{\text{I}, nm}^S$ as a function of time for system sizes ranging from $2N = 60$ to $2N=330$ and for the end points of the string being at $n = 2 \lfloor \frac{N}{2} \rfloor  - 21$ and $m = 2 \lfloor \frac{N}{2} \rfloor  +7$, so $|m-n| = 28$.
In~\cref{fig:mixing_time_quadradic}, we plot the mixing time as a function of system size for various string lengths. We fix $n = 2 \lfloor \frac{N}{2} \rfloor  - 21$ and vary $m$ to obtain strings of different lengths. For each system size and string length, we average over $5000$ samples and consider open boundary conditions.
We find that the mixing time is approximately $\sim N^2$ asymptotically, which can be understood by thinking of the distance between the defects effectively undergoing a random walk in 1d. 
Namely, each local channel element has a probability to pick a defect and move it to the right along the chain.
For a pair of defects, the distance between them shrinks or grows depending on which defect is picked to move to the right.
To approach the steady state $\rhopp$, the defects need to pair up and annihilate, where the defects can be separated by a distance of order $N$ at late times.
Because the standard deviation of the distance of a random walk scales as the square root of time, we expect it to take time $\sim N^2$ to annihilate all defects on average.
The $\sim N^2$ mixing time indeed is one of the features of the parent Lindbladian we discussed in ~\cref{sec:lindblad_mixing}, despite the fact that the Lindbladian and the quantum channel with open boundary conditions are gapped.

\subsection{Steady state of perturbed quantum channels}

Recall that we studied the steady-state properties of the interpolated Lindbladian $\mathcal{L}_\lambda = (1-\lambda) \mathcal{L}_{\mathcal{C}} + \lambda \tilde{\mathcal{L}}_{\mathcal{C}}$ for $\lambda \in [0, 1]$, where $\tilde{\mathcal{L}}_{\mathcal{C}}$ is the parent Lindbladian of the trivial SPT state  $\tilde{\rho}_{\mathcal{C}}$ obtained via the $U_{\text{CZ}}$ conjugation.
Here we also consider a similar interpolation achieved by the following channel:
\begin{equation}
    \mathcal{E}_{\lambda}=(1-\lambda)\mathcal{E} + \lambda \tilde{\mathcal{E}}~,
\end{equation}
where $\tilde{\mathcal{E}}$ is the quantum channel of $\mathcal{E}$ conjugated by $U_{\text{CZ}}$.
Since $U_{\text{CZ}} X_{j} U_{\text{CZ}} = Z_{j-1}X_jZ_{j}$ and $U_{\text{CZ}}$ commutes with $Z_j$, we can infer the actions of $\tilde{\mathcal{E}}_{j}$ in $\tilde{\mathcal{E}}$ and the implementation protocol accordingly.
More specifically, the quantum channel $\mathcal{E}_{\lambda}$ corresponds applying $\mathcal{E}$ or $\tilde{\mathcal{E}}$ with a probability $(1-\lambda)$ and $\lambda$, respectively, at each step.
If we need to apply $\tilde{\mathcal{E}}$, then we just need to follow the procedure outlined for $\mathcal{E}$ but replacing all $Z_{j-1}X_jZ_{j}$ by $X_j$ and vice versa.
We again are interested in the steady state of $(\mathcal{E}_{\lambda})^t$ when $t \rightarrow \infty$, which can again be efficiently simulated on a classical computer using Clifford simulation.

\begin{figure}[tb]
     \centering
     \includegraphics[width=\columnwidth]{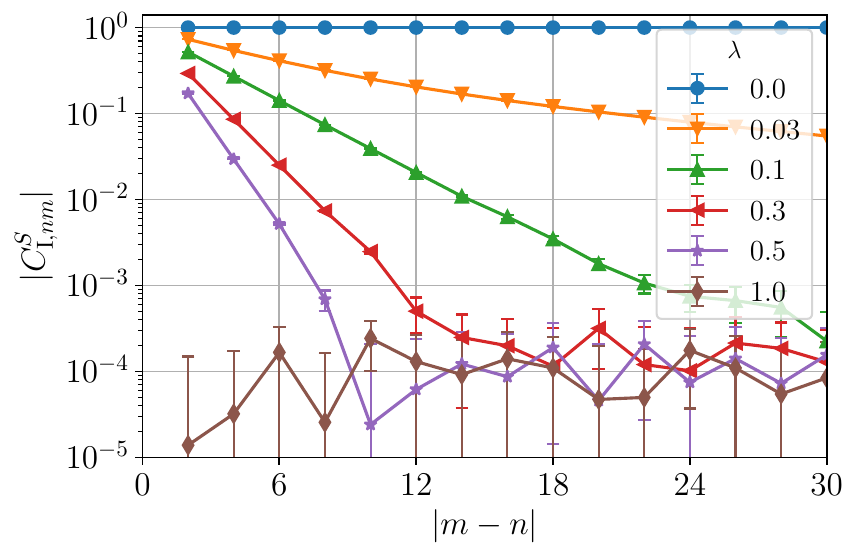}
     \caption{The string correlator $C_{\text{I,nm}}$ as a function of $|m-n|$ for a closed chain of 100 qubits. The legend shows the values of $\lambda$. For $0<\lambda <1$, its exponential decay to zero indicates that the steady state does not have the mixed-state SPT order.
     The error bars indicate the standard error in the string correlator in our numerics.
     The large error bars around $C_{\text{I}, nm}^S \lesssim 10^{-4}$ are due to the finite number of samples used in our Clifford simulation.}
     \label{fig:cs1_clifford}
\end{figure}

We examine various correlators of the steady state of the perturbed channel $\mathcal{E}_{\lambda}$ to determine its phase. 
First, we examine its $\ztwo^S \times \ztwo^W$ mixed-state SPT order via the string order parameters.
We plot $C_{\text{I}, nm}^S$ in~\cref{fig:cs1_clifford} for a system of $2N =100$ qubits with open boundary conditions for $n=11$ and various $m$, where the data is obtained by averaging over $9900$ samples for each string length and each $\lambda$. Here each sample is calculated by averaging the value of $C_{\text{I}, nm}^S$ in the stabilizer state for $2.4\times 10^{6}$ time steps after the steady state has been reached. We find that the string correlator $C_{\text{I}, nm}^S$ decays exponentially with the string length for $\lambda \neq 0$, indicating that the steady state does not possess nontrivial mixed-state SPT order. 
This is again similar to the interpolated Lindbladian we studied in~\cref{sec_perturbed_Lindbladian}.

\begin{figure}[tb]
     \centering
     \includegraphics[width=\columnwidth]{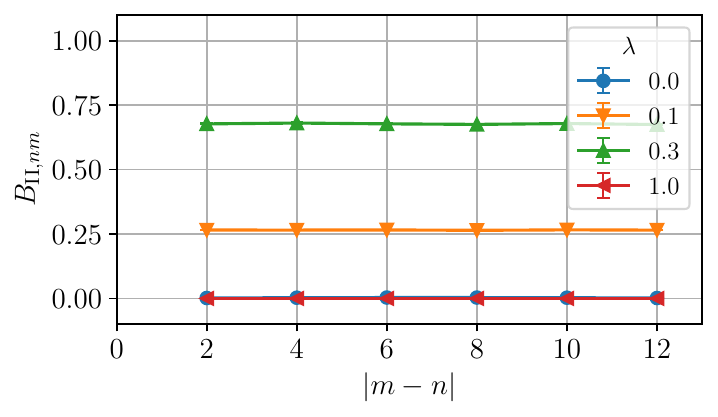}
     \caption{The R\'eyni-2 correlator $B_{\text{II}, nm}$ of the steady states of $\mathcal{E}_{\lambda}$ for various $\lambda$ and system size $2N=14$. The legend indicates the perturbation strength $\lambda$.
     The nonzero value of $B_{\text{II}, nm}$, together with $A_{\text{I}, nm}=A_{\text{II}, nm}=0$, indicates that the steady state exhibits SW-SSB.
     Note that the error bars of the data are within the size of the symbols, and the $\lambda =1.0$ data points are on top of the $\lambda =0.0$ data points.
     }
     \label{fig:b2_clifford}
\end{figure}

As discussed previously, we expect that the destruction of the $\ztwo^S \times \ztwo^W$ mixed-state SPT order comes from the strong symmetry $\ztwo^S$ being spontaneously broken.
While not shown in the figures, we confirm this by calculating the $A_{\text{I}, nm}$ correlator, finding that it is exactly zero for all $\lambda$. 

To see if the strong symmetry is broken down to the weak symmetry,
we need to calculate the $B_{\text{II},nm}$ correlator. 
To this end, we need to calculate the purity $\Tr(\rho^2)$ for the denominator and $\Tr(\rho Z_nZ_m \rho  Z_nZ_m)$ for the numerator. 
Since the Clifford simulation generates $\rho=\sum_{i}p_i|\psi_i\ra \la \psi_i|$ by sampling a stabilizer state $\vert\psi_i\ra$ with probability $p_i$, we can calculate the purity as $\Tr(\rho^2)=\sum_{i,j}p_ip_j |\la \psi_i|\psi_j\ra|^2$ by averaging $|\la \psi_i|\psi_j\ra|^2$ over realizations of two-copies of $\rho$. 
Likewise, quantities like $\Tr(\rho A\rho B)=\sum_{i,j}p_ip_j \la \psi_i|A|\psi_j\ra\la \psi_j|B|\psi_i\ra$ can be calculated by averaging $\la \psi_i|A|\psi_j\ra\la \psi_j|B|\psi_i\ra$ over realizations of two copies of $\rho$. 
The quantities $|\la \psi_i|\psi_j\ra|^2$ and $\la \psi_i|A|\psi_j\ra\la \psi_j|B|\psi_i\ra$ can be calculated efficiently if $|\psi_i\ra$ is a stabilizer state and $A$, $B$ are Pauli operators. This is shown in \cref{app:stabilizer}, and we outline the algorithm in \cref{alg:stab-overlap}.

However, notice that, for the decohered cluster state $\rho_{\cal{C}}$, the purity is $\Tr(\rho^2)=2^{-N}$. 
This indicates that, in order to obtain a fixed multiplicative accuracy for the purity, one needs the number of samples to grow exponentially with $N$.
We expect this would also generally be the case for quantities like $\Tr(\rho A \rho B)$ and for steady states at nonzero $\lambda$.
Therefore, we only compute $B_{\text{II},nm}$ for small sizes $2N$.

\cref{fig:b2_clifford} shows a plot of the $B_{\text{II}, nm}$ correlator for a closed chain of 14 qubits. For each $\lambda$ and value of $|m-n|$ in the plot, we average over between $2\times 10^5$ and $10^6$ samples. 
We find that $B_{\text{II},nm}$ is nonzero for $0<\lambda <1$. While not shown in the figures, we also calculate $A_{\text{II}, nm}$ and find that $A_{\text{II}, nm}=0$ exactly.
This indicates that the steady state $\rho_{\lambda}$ for $0 < \lambda < 1$ indeed has SW-SSB.

\subsection{Decay time}
We have seen in~\cref{sec:clifford_mixing} that the time scale for a trivial state to reach the nontrivial mixed-state SPT state $\rhopp$ scales as $\sim N^2$ in system size.
It is therefore interesting to examine the time scale for $\rhopp$ to reach the SW-SSB steady state under the perturbed quantum channel, and we characterize this ``decay time'' via the string order parameter $C_{\text{I}, nm}^S$.
As shown in~\cref{fig:decay_time}, $C_{\text{I}, nm}^S$ starts at one in $\rhopp$ and decays exponentially with time $t$.
The exponential-decay form allows us to extract the decay time scale by a linear fit of $\ln{C_{\text{I}, nm}^S}$.
In~\cref{fig:decay_time_linear}, we plot this decay time for a string of length $|m-n|=38$ and $\lambda = 0.1$ as a function of system size.
Intriguingly, we numerically observe that the decay time scales almost linearly in $N$. It would be interesting to find a simple explanation for this behavior.

\begin{figure}[tb]
\captionsetup[subfigure]{labelformat=empty,captionskip=-20pt}
     \centering
     \hfill
 \subfloat[\label{fig:decay_time}]{%
  \includegraphics[width=0.49\columnwidth]{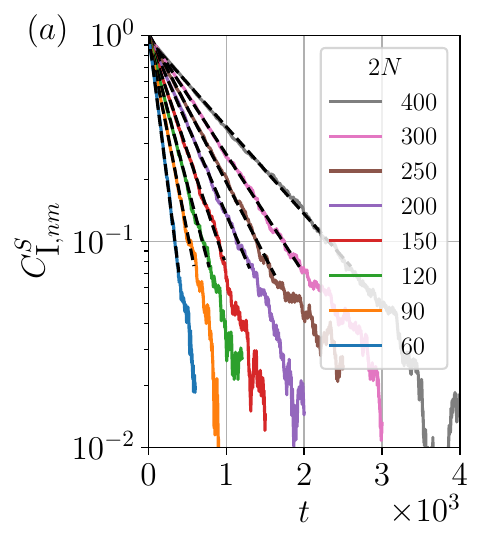}%
}
     \hfill
 \subfloat[\label{fig:decay_time_linear}]{%
  \includegraphics[width=0.49\columnwidth]{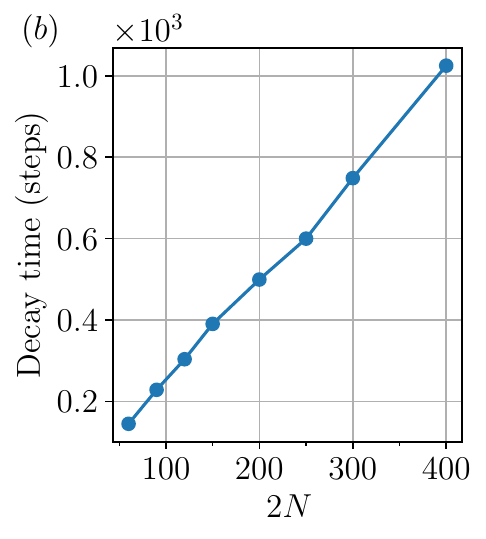}%
}
     \hfill
        \caption{ 
        $(a)$ The destruction of the mixed-state SPT order probed by $C_{\text{I}, nm}^S$ as a function of $t$ under the perturbed channel $\mathcal{E}_{\lambda}$ with the initial state $\rhopp$. 
        We take $\lambda$ to be $0.1$ and keep the string length fixed at $|m-n| = 38$.
        The legend indicates the total system size $2N$.
        The exponential decay of $C_{\text{I}, nm}^S$ enables us to extract the time scale via a linear fit to $\ln{C_{\text{I}, nm}^S}$.
        $(b)$ The extracted decay time for various system sizes $2N$. We observe that the decay time scales roughly linearly in system size $N$.
        }
       
\end{figure}

\subsection{Trajectory dynamics}

The density matrix at each time step of the channel can be understood as an average over various trajectories, corresponding to different probabilistic outcomes up to that point in the channel's history. Each of these trajectories exhibits the reaction-diffusion dynamics discussed in~\cref{sec:lindblad_mixing}. In order to understand the trajectories of the perturbed channel $\mathcal E_\lambda$, we plot the dynamics of defects on even sites. 
Specifically, we simulate $\mathcal E_\lambda$, and at each time step we compute the expectation value of $X_i$ for each even site $i$.
When the expectation value is $-1$ (denoting in black in Fig.~\ref{fig:traj-dynamics}), we say that site $i$ has a defect. 
The dynamics are such that defects are created and destroyed in pairs.
We also plot the defect density---the number of defects divided by $2N$---as a function of depth.

As we see in the upper panel of~\cref{fig:traj-dynamics}, defects are created in pairs and propagate probabilistically in a single direction, namely toward the higher numbered sites. Two defects can annihilate if they meet one another, otherwise they continue to propagate. These dynamics result in a nonzero average density of defects at long times, as can be seen in the lower panel of~\cref{fig:traj-dynamics}. The instability of the mixed-state SPT order can be understood as this average density of defects remaining finite for arbitrarily small values of $1-\lambda$, which results in the destruction of the string order.
\begin{figure}[tb]
    \centering
\includegraphics[width=\linewidth]{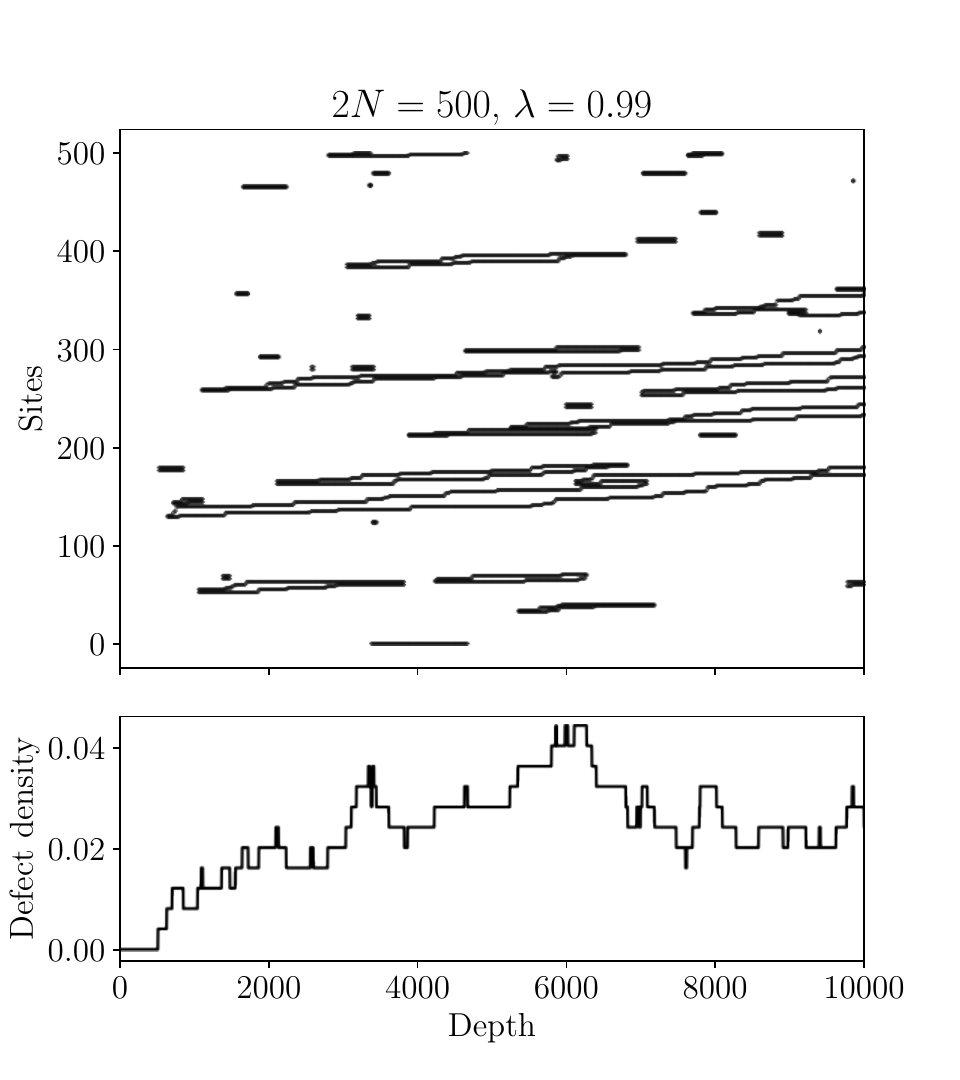}
    \caption{We plot the dynamics of defects (top) and the density of defects (bottom) for a single trajectory as a function of the depth of the channel $\mathcal E_\lambda$ with $\lambda=0.99$ and with periodic boundary conditions on $2N = 500$ qubits. Black on an even site $i$ denotes an $X_i$ expectation value of $-1$.}
    \label{fig:traj-dynamics}
\end{figure}

\section{Discussion}\label{sec_discussion}

Focusing on the decohered cluster state, we examined some of its characteristics as a mixed-state SPT order, including the bulk-edge correspondence, nontrivial string order, and the corresponding edge states via the matrix-product density operators representation.
We then constructed a parent Lindbladian that hosts the decohered cluster state as one of several steady states. Through DMRG and mappings to exactly solvable reaction-diffusion dynamics, we found that the steady-state mixed-state SPT order is unstable to arbitrary small symmetric perturbations, a stark difference from the pure state case. The instability is characterized by the onset of strong-to-weak spontaneous symmetry breaking, which has no counterpart in pure state physics. Based on the dual perspective provided by $U_{\text{CZ}}$ conjugation, we were able to interpret the destruction of SPT order as a consequence of the proliferation of strong-string-order defects. On the other hand, we found that SPT order is robust against perturbations that generate only weak-string-order defects. We further designed a quantum channel that reproduces the key features of the interpolated Lindbladian we studied. 
This channel can be efficiently simulated using only Clifford gates, Pauli measurements, and feedback. Using this method to measure observables, we found qualitative agreement with aforementioned results measured in the Lindbladian framework via DMRG. In addition, the quantum channel simulation confirms our expectation of the mixing time scaling while also giving us an intriguing decay time result.

Using the decohered cluster state as a test case, we established the instability of the steady-state mixed-state SPT phase in 1d and of its associated steady-state degeneracy. 
In particular, the instability comes about due to proliferation of point-like defect-pairs in 1d. 
These point-like defects must propagate a distance extensive in the system size to meet another defect and annihilate with it. 
With a restriction to local Lindbladians, or without the global information of the defect locations, the best one can do to annihilate such defects is to make them wander randomly.
As pointed out in a recent work~\cite{chirame_stable_2024}, one way to alleviate this is by making the defect heralded.
We expect that the instability due to such point-like defect proliferation is generic for mixed-state SPT order in 1d.

On the other hand, in 2d, the defects can be line-like rather than point-like. 
In this case, defect loops can annihilate by contracting to a point rather than fusing with one another, the time-scale for which is no longer extensive in system size. We therefore expect that higher dimensional generalizations of the decohered cluster state---and perhaps steady-state SPT order and the corresponding steady-state degeneracy more generally---may be more stable to local symmetric perturbations. 

Another direction that warrants further study is the generalization of the physics we discussed to bosonic systems. 
For example, the spontaneous strong symmetry breaking down to nothing has been generalized to bosonic open quantum systems, leading to a degenerate steady-state manifold~\cite{lieu2020symmetry}. 
While we do not expect that the SW-SSB would lead to additional degeneracy other than the ones guaranteed by the strong symmetry, it is still worth investigating if such a phenomenon can be realized in a bosonic system, or even by Gaussian states, which can be readily realized by linear optics.

Our work provides a concrete example that helps pave the way for classifying steady-state phases via their parent Lindbladians~\cite{rakovszky_defining_2024} instead of the density matrices themselves. 
One may have thought that adding local symmetric perturbations to the parent Lindbladian cannot lead to any instability. However, we have provided a wide class of examples where this is not the case. It therefore remains an open question of how to construct all potential symmetric perturbations under which the steady-state phases are guaranteed to be stable. Another natural extension of our work is to investigate steady-state phases of other types, including the cases of mixed-state phases of intrinsic topological orders or with symmetry enriched intrinsic topological orders~\cite{bao2023mixed-state,lee2023quantum,ellison_towards_2024}.
Another promising direction is to understand the mixed-state SPT order or the transition to SW-SSB at the level of individual quantum trajectories, and to examine whether a particular choice of unraveling of the Lindbladian or quantum channel provides an advantage towards this understanding.
For example, it was pointed out that the presence of strong symmetry in the dynamics of an open quantum system can lead to dissipative freezing~\cite{sanchez2019symmetries,tindall2023onthegenerality,munoz2019nonstationary,halati2022breaking} on quantum trajectories. 
Examining how this phenomenon interacts with the mixed-state SPT and SW-SSB could provide us with a novel approach to detect these phases.

\begin{acknowledgments}
We thank Meng Cheng, Tyler Ellison, Zhi Li, Ruochen Ma, Connor Mooney, Alex Turzillo, and Yizhi You for valuable discussions and feedback. 
C.-J.L.~acknowledges support from the National Science Foundation (QLCI grant OMA-2120757).
J.D.W.~acknowledges support from the United States Department of Energy, Office of Science, Office of Advanced Scientific Computing Research, Accelerated Research in Quantum Computing program, and also NSF QLCI grant OMA-2120757. 
Y.-Q.W.~is supported by a JQI postdoctoral fellowship at the University of Maryland. Y.-X.W.~acknowledges support from a QuICS Hartree Postdoctoral Fellowship.
A.V.G., C.F., J.T.I., J.S., and B.W.~were supported in part by the DoE ASCR Accelerated Research in Quantum Computing program (awards No.~DE-SC0020312 and No.~DE-SC0025341), AFOSR MURI, DARPA SAVaNT ADVENT, NSF QLCI (award No.~OMA-2120757), the DoE ASCR Quantum Testbed Pathfinder program (awards No.~DE-SC0019040 and No.~DE-SC0024220), and the NSF STAQ program. Support is also acknowledged from the U.S.~Department of Energy, Office of Science, National Quantum Information Science Research Centers, Quantum Systems Accelerator. C.F.~also acknowledges support from NSF DMR-2345644 and from the NSF QLCI grant OMA-2120757 through the Institute for Robust Quantum Simulation (RQS). The authors acknowledge the University of Maryland supercomputing resources (\href{http://hpcc.umd.edu}{http://hpcc.umd.edu}) used in this work.
\end{acknowledgments}

\section*{Data Availability}
All data generated and code used in this work are available at:
\href{https://doi.org/10.5281/zenodo.17451966}{10.5281/zenodo.17451966}

\bibliographystyle{quantum}
\bibliography{Mixed_Cluster}

\onecolumngrid
\appendix

\section{A Lindbladian definition of mixed state phases }\label{app:string_lieb_robinson}
In this appendix, we show that $\rho_\mathcal{C}$ is also a nontrivial SPT state  according to the ``fast-driving Lindbladian'' definitions of Refs.~\cite{coser2019classification, degroot2022symmetry}, complementing the finite-depth quantum channel result obtained in Refs.~\cite{ma2023average,ma2023topological}.
In Ref.~\cite{coser2019classification}, two mixed states are defined to be in the same phase if there are local Lindbladian evolutions which can rapidly bring the two states close to each other.
While Ref.~\cite{coser2019classification} essentially implies that all the mixed states are in the ``trivial'' phase in 1d if the symmetry of the Lindbladian evolution is weak, Ref.~\cite{degroot2022symmetry} demonstrates that one can have nontrivial SPT states by requiring that the Lindbladian evolution respects the strong symmetry of the system. Below, we review the definition of two states being in the same phase according to Refs.~\cite{coser2019classification, degroot2022symmetry}.
\begin{definition}[Fast Driving] \label{Def:Phase_of_Matter}
    Given two states $\rho_0$ and $\rho_1$, we say that $\rho_1$ can be driven fast to $\rho_0$ if there exists a geometrically local time-independent Lindbladian $\cL_0$ such that
    \begin{align*}
        \norm{e^{\cL_0 t}(\rho_0) - \rho_1}_1 &\leq \poly(N)e^{-\gamma t}~
    \end{align*}
    for some time $t \gtrsim polylog(N)$,
    where $\gamma=O(1)$ and the trace distance $\|X\|_1 \vcentcolon = \tr[\sqrt{X^{\dagger}X}]$.
    We denote it as $\rho_0 \xrightarrow[]{\cL_0} \rho_1$
    
\end{definition}

\begin{definition}[Lindbladian Phase of Matter]
\label{Def:Coser_Phase}
    We say that $\rho_0$ and $\rho_1$ are in the same phase if there exist geometrically  local symmetric time-independent Lindbladians $\cL_0$ and $\cL_1$ such that $\rho_0 \xrightarrow[]{\cL_0} \rho_1$ and $\rho_1 \xrightarrow[]{\cL_1} \rho_0$.
\end{definition}
\noindent Here, by a geometrically local Lindbladian, we refer to a Lindbladian that can be expressed as a sum of local terms~\cite{coser2019classification}:
\begin{equation}
    \mathcal{L} = \sum_{X \in \Lambda} \mathcal{L}_{X}~, \;\;\; \norm{\mathcal{L}_X}_{\text{op}} \sim O(1)~,  \;\;\; \text{supp}[\mathcal{L}_X] = X~, \;\;\; |X| < k
\end{equation}
for some fixed constant $k$, where  $\Lambda$ denotes the lattice, and $|X|$ refers to the diameter of the subset $X$. 
As noted in Ref.~\cite{coser2019classification}, this definition may be generalized to quasi-local Lindbladians in certain contexts.
The intuition behind this definition is that, if one state can be reached by a $\log(N)$-time evolution from another state, then the long-range correlations in the two states must be similar.
Furthermore, this rapid mixing condition implies stability of local observables to local perturbations in the Lindbladian terms when away from the boundaries of the phase~\cite{cubitt2015stability, onorati2023provably}.
Moreover, it is a natural and physically motivated generalization of the Hamiltonian definition of phases to open systems.
Indeed, Ref.~\cite{coser2019classification} demonstrates that this Lindbladian definition encompasses that standard gapped-Hamiltonian phase of matter. 
Furthermore, it can be shown that it encompasses high-temperature thermal phases using the Lindbladian of Refs.~\cite{chen2023efficient, rouze2024efficient}, and this even applies to families of states that are classically intractable to prepare \cite{bergamaschi2024quantum, rajakumar2024gibbs}.
 
 Compared to previous works using this definition \cite{coser2019classification, onorati2023provably}, we have removed the auxiliary system as it is not required for this work.
However, anything that satisfies this more restricted definition of phase must also satisfy the definition of phase in Ref. \cite{coser2019classification}.
We also note recent work in Ref. \cite{sang2024stability} where it has been shown that this Lindbladian definition implies that two states are in the same phase if one can smoothly move between fixed points without the ``correlation length'' for the mutual information diverging. 
Here by ``correlation length'', we mean the length scale over which the mutual information decays exponentially.

Here we demonstrate that the decohered cluster state $\rho
_{\mathcal{C}}$ cannot be in the trivial phase by the Coser and P\'erez-Garc\'ia as per~Defs.~\ref{Def:Phase_of_Matter} and \ref{Def:Coser_Phase} if the Lindbladian respects a $\ztwo^S \times \ztwo^W$ symmetry. 
To prove this, we first realize that any state evolving under a local Lindbladian has an associated speed limit on the speed of information propagation, known as the Lieb-Robinson velocity~\cite{poulin2010lieb}.

 More formally, consider a local observable $O_A$, supported only on a region $A$. 
 Let $A(r)$ be an area of distance $r$ around $A$, and let $\cL_{A(r)}$ be the Lindbladian restricted to only terms that are in $A(r)$.
 Assuming the Lindbladian is a sum of local terms $\cL=\sum_{j}L_j$ and the strength of each local term is bounded as $\max_j\norm{L_j}\leq J$, we have
\begin{align} \label{Eq:Lindblad_Lieb-Robinson}
    \norm{e^{\cL^\dagger t}(O_A) - e^{\cL_{A(r)}^\dagger t}(O_A)   } \leq \norm{O_A}|A|J \frac{e^{vt-\mu r}}{v}   
\end{align}
for constants $v, \mu =O(1)$ from Ref. \cite[Lemma 5.5]{cubitt2015stability} or Ref.~\cite{poulin2010lieb}.
In the case that the Lindbladian acts trivially on parts of $O_A$, we replace  $|A|$ on the right-hand side with the size of the region which $\cL$ acts nontrivially on $O_A$ i.e. $A$ is now the set of sites where $L_i(O_A)\neq 0$ $\forall i$. 
 Using this, we are able to prove the following theorem.
\begin{theorem}
    Under the Coser \& P\'erez-Garc\'ia definition of phase of matter, as per Definitions \ref{Def:Phase_of_Matter} and \ref{Def:Coser_Phase}, any state $\rho$ for which $\langle \mathcal{S}_g \rangle =\Tr [\rho \mathcal{S}_g] = \Omega(1)$ is not in the trivial phase, where 
    \begin{align}
    \label{eq:general_string}
\mathcal{S}_g=\mathcal{O}_\alpha^{(L)}\otimes\left(\prod_{i\in M}U_g^{(i)}\right)\otimes\mathcal{O}_\alpha^{(R)},
    \end{align}
    for $|L-R|=\Omega(N)$, and $M$ denotes the set of sites between $
L$ and $R$. 
\end{theorem}
\begin{proof}
    Note that, since the Lindbladian commutes with the strong symmetry $\prod_i U_g^{(i)}$, the string part of the order parameter remains invariant under the Lindbladian evolution because the bulk of the string is made of symmetry operators $U_g^{(i)}$.
    The Lindbladian is therefore only acting nontrivially on the end parts of $\mathcal{S}_g$.
    Denote $L(r)$ and $R(r)$ to be the regions within a distance $r$ of the end points $L,R$, respectively.
    Applying~\cref{Eq:Lindblad_Lieb-Robinson} we get
    \begin{align}\label{Eq:Localized_Evolution}
    \begin{split}
        &\norm{   e^{\cL^\dagger t} (\cS_g) -   e^{(\cL^\dagger_{L(r)} + \cL^\dagger_{R(r)}) t} (\cS_g) } 
        = O(e^{vt - \mu r  }).
    \end{split}
    \end{align}
  We note that, although $\cS_g$ acts over an $O(N)$ area, due to the symmetry condition, it only acts nontrivially on an $O(1)$ area, hence $|A|=O(1)$.

    We now consider the expectation value of the string order parameter on any state in the trivial phase.
    By the Coser \& P\'erez-Garc\'ia definition of phase, we can write any state in the trivial phase as $\rho = e^{t\cL}(\rho_\text{triv})$ for an appropriate Lindbladian, where  $\rho_\text{triv}$ is a reference product state in the trivial phase and $t = O(\log(N))$.
    Now consider
    \begin{align}
    \begin{split}
    \Tr\left[e^{\cL t}(\rho_\text{triv})   \cS_g\right] =&  \Tr\left[\rho_\text{triv}  e^{\cL^\dagger t} (\cS_g)\right] \\
         \leq & \Tr\left[\rho_\text{triv}   e^{(\cL^\dagger_{L(r)} + \cL^\dagger_{R(r)}) t} (\cS_g)\right] + O(e^{vt - \mu r}),
    \end{split}
    \end{align}
    where we have applied \cref{Eq:Localized_Evolution} to reach the second line.
    
    Since $t=O(\log(N))$, we can choose a localized Lindbladian which contains the Lieb-Robinson light-cone by choosing $r=c|L-R|$ for $c<1/2$.
    Since we have assumed $|L-R|=\Omega(N)$, we see that $O(e^{vt - \mu r}) = O(e^{-\mu' N})$ for some constant $\mu'$.
    Thus we can write
    \begin{align*}
        \Tr\left[e^{\cL t}(\rho_\text{triv})   \cS_g\right] \leq& \ \Tr\left[\rho_\text{triv}   e^{(\cL^\dagger_{L(r)} + \cL^\dagger_{R(r)}) t} (\cS_g)\right] + O(e^{-\mu' N}) \\
        =& \ \Tr\left[ \Tr_{L(r)}[\rho_\text{triv}]   e^{\cL^\dagger_{L(r)}}(\cS_g) \Tr_{L^c(r)}[\rho_\text{triv}] e^{ \cL^\dagger_{R(r)} t} (\cS_g)\right] + O(e^{-\mu' N})  \\ 
        =& \ \Tr\left[ \Tr_{L(r)}[\rho_\text{triv}]   e^{\cL^\dagger_{L(r)}}(\cS_g) \right] \Tr\left[\Tr_{L^c(r)}[\rho_\text{triv}] e^{ \cL^\dagger_{R(r)} t} (\cS_g)\right] + O(e^{-\mu' N}).
    \end{align*}
    To reach the second line, we have split the product state into a part supported on $L(r)$ and its complement, and we have used that $\cL^\dagger_{L(r)}$ and $\cL^\dagger_{R(r)}$ have disjoint supports, hence $e^{(\cL^\dagger_{L(r)} + \cL^\dagger_{R(r)}) t} = e^{\cL^\dagger_{L(r)} t}e^{\cL^\dagger_{R(r)} t} $.
 Since the Lindbladian preserves the symmetry operator, we see that the endpoint operators have the same charge as in the non-time-evolved case, hence we can write
 \begin{align*}
     \Tr\left[ \Tr_{L(r)}[\rho_\text{triv}]   e^{\cL^\dagger_{L(r)}}(\cS_g) \right] \Tr\left[\Tr_{L^c(r)}[\rho_\text{triv}] e^{ \cL^\dagger_{R(r)} t} (\cS_g)\right] = 0
 \end{align*}
 and thus
 \begin{align*}
     \Tr\left[e^{\cL t}(\rho_\text{triv})   \cS_g\right] =& O(e^{-\mu' N}).
 \end{align*}
 The immediate consequence is that, for any state in the trivial phase, it cannot be the case that $\langle \cS_g \rangle =\Omega(1)$  as $N$ increases.

\end{proof}

\section{Fast destruction of strong-string order}\label{app:fast destruction string order}
In this appendix, we show the fast destruction of the strong-string order $ \mathcal{S}_{n,m}=Z_{n}X_{n+1}\cdots X_{m-1}Z_m$ for odd $n$ and $m$ as defined in the main text in~\cref{eq:string_op}.
The destruction of the strong-string order can be achieved by a trace channel $\mathcal{E}=\prod_{j \in \text{odd}} \mathcal{T}_j$ where $\mathcal{T}_j(\rho) \vcentcolon = \Tr_j[\rho] \frac{\mathbbm{1}_j}{2}$. It is easy to check that $\mathcal{E}$ indeed satisfies the $\ztwo^S \times \ztwo^W$ symmetry and is of depth one. 

Assume $\rho$ is a nontrivial SPT state and denote $\sigma = \mathcal{E}(\rho)$. We would like to calculate $\la  \mathcal{S}_{n,m}\ra_{\sigma} = \Tr[\mathcal{S}_{n,m} \sigma]=\Tr[\mathcal{E}^{\dagger}(\mathcal{S}_{n,m}) \rho]$.
It is easy to check that $\mathcal{T}_j^{\dagger}(O) = \Tr_j(O) \frac{\mathbbm{1}_j}{2}$, so we have $\mathcal{E}^{\dagger}(\mathcal{S}_{n,m})=\tilde{\mathcal{S}}_{nm}/4$ where $\tilde{\mathcal{S}}_{nm}$ is the trivial string operator defined in~\cref{eq:2copy_trivial_correlators} and the paragraph succeeding it. We therefore see that $\la \mathcal{S}_{n,m} \ra_{\sigma} = \la \tilde{\mathcal{S}}_{n,m} \ra_{\rho}/4 = 0$, indicating that the strong-string order is being destroyed via a depth-one symmetric quantum channel.

This result can also be generalized to the Lindbladian case. Consider the following trace Lindbladian: $\mathcal{L}_{\mathcal{T}}\vcentcolon =\sum_{j \in \text{odd}} (\mathcal{T}_j - \mathbbm{1}_j)$.
Note that
\begin{align}
    e^{\mathcal{L}t}&=\prod_{j \in odd}[\mathcal{T}_j - e^{-t}(\mathcal{T}_j-\mathbbm{1}_j)] 
    = \prod_{j \in odd}[(1-e^{-t})\mathcal{T}_j + e^{-t}\mathbbm{1}_j)] 
    = \sum_{\Lambda} e^{-(N-|\Lambda|)t}(1-e^{-t})^{|\Lambda|} \mathcal{T}_{\Lambda}~,
\end{align}
where $\Lambda$ denotes all the possible sets of $j$, $|\Lambda|$ is the number of sites in the set $\Lambda$, and $\mathcal{T}_{\Lambda}$ is the trace channel on the set $\Lambda$.
Note that $\lim_{t \rightarrow \infty} e^{\cL t}(\rho) = \mathcal{E}(\rho) \vcentcolon = \sigma$. 
We have
\begin{align}
    \|e^{\mathcal{L}t}(\rho) - \sigma \|_1 &\leq \sum_{n=0}^{N-1} \binom{N}{n} e^{-(N-n)t}(1-e^{-t})^{n} \| \mathcal{T}_{\Lambda} (\rho)\|_1 \notag \\
    &=(e^{-t}+1-e^{-t})^N - (1-e^{-t})^N = 1-(1-e^{-t})^N \leq Ne^{-t} + O(N^2 e^{-2t})~. 
\end{align}
Therefore, for a time $t \gtrsim \alpha \log N$, $\|e^{\mathcal{L}t}(\rho) - \sigma \|_1 = O(N^{1-\alpha})$.
This implies that a nontrivial state $\rho$ can be driven fast to a trivial state $\sigma$ when $\alpha >1$.
We therefore conclude that the mixed-state SPT order is hard to build but easy to destroy.

\section{\texorpdfstring{String order of $\rhomm$}{String order of rho minus}}\label{ap:rho minus}
In this appendix, we will show that $\rhomm$ has nontrivial string order under the same string order parameters as $\rhopp$. First recall that we may express $\rhomm$ in the vectorized picture as
\begin{align}
    \rhomm=\frac{1}{N}\sum_{k\text{ even}}\begin{tikzpicture}
    \draw (0,0) rectangle (6.25,-0.35);
    \foreach \x in {0.5,2.75,5}{
		     \draw (\x,0) -- (\x,0.55);
       \draw (\x,0) arc (-60:0:.63);
		    \node[odd] at (\x,0) {};
		  }
		  \foreach \x in {1.25,3.5,5.75}{
		     \draw (\x,0) -- (\x,0.55);
       \draw (\x,0) arc (-60:0:.63);
		    \node[even] at (\x,0) {};
		  }
		  \node[fill=white] at (2,0) {$\dots$};
        \node[fill=white] at (4.25,0) {$\dots$};
        \node[Zltight] at (3.5,0.4) {};
        \node[Zblank] at (3.78,0.4) {};
        \node[] at (3.7,-.15) {$k$};
  \end{tikzpicture} \ .
\end{align}
Next, observe that
\begin{align}\label{eq:O delta}
    \Tr[\rhomm\rhomm^\dagger]=\frac{1}{N^2}\sum_{\substack{
k,k' \\
\text{even}}}\begin{tikzpicture}
    \draw (0,0) rectangle (8.5,-0.35);
    \draw (0,1.1) rectangle (8.5,1.45);
    \foreach \x in {0.5,2.75,5,7.25}{
		     \draw (\x,0) -- (\x,1.1);
       \draw (\x,0) arc (-60:60:.63);
		    \node[odd] at (\x,0) {};
            \node[odd] at (\x,1.1) {};
		  }
		  \foreach \x in {1.25,3.5,5.75,8}{
		     \draw (\x,0) -- (\x,1.1);
       \draw (\x,0) arc (-60:60:.63);
		    \node[even] at (\x,0) {};
            \node[even] at (\x,1.1) {};
		  }
		  \node[fill=white] at (2,0) {$\dots$};
        \node[fill=white] at (2,1.1) {$\dots$};
        \node[fill=white] at (4.25,0) {$\dots$};
        \node[fill=white] at (4.25,1.1) {$\dots$};
        \node[fill=white] at (6.5,0) {$\dots$};
        \node[fill=white] at (6.5,1.1) {$\dots$};
        \node[Zltight] at (3.5,0.4) {};
        \node[Zblank] at (3.78,0.4) {};
        \node[] at (3.7,-.15) {$k$};
        \node[Zltight] at (5.75,0.6) {};
        \node[Zblank] at (6.03,0.6) {};
        \node[] at (6,1.25) {$k'$};
        \draw[|-|] (3.5,1.7) -- node[above] {$\Delta$}(5.75,1.7) {};
  \end{tikzpicture} \ ,
\end{align}
where we have defined $\Delta = \frac{k'-k}{2}\text{ mod }N$ to be the distance between $k$ and $k'$ within the even sublattice. Notice that each term in the summand depends only on $\Delta$ due to translational invariance.
Defining the tensor diagram appearing in~\cref{eq:O delta} to be $O_\Delta$, we have
\begin{align}\label{eq:O delta sum}
\begin{split}
    \Tr[\rhomm\rhomm^\dagger]&=\frac{1}{N^2}\sum_{\substack{
k,k' \\
\text{even}}}O_\Delta\\
&=\frac{1}{N}\sum_{\Delta=0}^{N-1}O_\Delta,
\end{split}
\end{align}
where we have replaced the sum over $k$ and $k'$ with a sum over $\Delta$ and a multiplicity factor of $N$ which accounts for the different possible endpoints separated by $\Delta$ sites.

Now, consider the numerator of the two-copy expectation value of the strong-string $\mathcal{S}^S_{nm}$:
\begin{align}
\begin{split}
    \Tr[\rhomm\mathcal{S}^S_{nm}\rhomm^\dagger]&=\frac{1}{N^2}\sum_{\substack{
k,k' \\
\text{even}}}\begin{tikzpicture}
    \draw (0,0) rectangle (8.5,-0.35);
    \draw (0,1.1) rectangle (8.5,1.45);
    \foreach \x in {0.5,2.75,5,7.25}{
		     \draw (\x,0) -- (\x,1.1);
       \draw (\x,0) arc (-60:60:.63);
		    \node[odd] at (\x,0) {};
            \node[odd] at (\x,1.1) {};
		  }
		  \foreach \x in {1.25,3.5,5.75,8}{
		     \draw (\x,0) -- (\x,1.1);
       \draw (\x,0) arc (-60:60:.63);
		    \node[even] at (\x,0) {};
            \node[even] at (\x,1.1) {};
		  }
		  \node[fill=white] at (2,0) {$\dots$};
        \node[fill=white] at (2,1.1) {$\dots$};
        \node[fill=white] at (4.25,0) {$\dots$};
        \node[fill=white] at (4.25,1.1) {$\dots$};
        \node[fill=white] at (6.5,0) {$\dots$};
        \node[fill=white] at (6.5,1.1) {$\dots$};
        \node[Zltight] at (3.5,0.25) {};
        \node[Zblank] at (3.72,0.25) {};
        \node[] at (3.7,-.15) {$k$};
        \node[Zltight] at (5.75,0.85) {};
        \node[Zblank] at (5.97,0.85) {};
        \node[] at (6,1.25) {$k'$};
        \node[Xltight] at (5.75,0.55) {};
        \node[Xltight] at (3.5,0.55) {};
        \node[Xltight] at (1.25,0.55) {};
        \node[Zltight] at (0.5,0.55) {};
        \node[Zltight] at (7.25,0.55) {};
        \node[] at (0.7,-.15) {$n$};
        \node[] at (7.47,-.15) {$m$};
  \end{tikzpicture} \\
  &=\frac{1}{N^2}\sum_{\substack{
k,k' \\
\text{even}}}(-1)^{s(k)}\begin{tikzpicture}
    \draw (0,0) rectangle (8.5,-0.35);
    \draw (0,1.1) rectangle (8.5,1.45);
    \foreach \x in {0.5,2.75,5,7.25}{
		     \draw (\x,0) -- (\x,1.1);
       \draw (\x,0) arc (-60:60:.63);
		    \node[odd] at (\x,0) {};
            \node[odd] at (\x,1.1) {};
		  }
		  \foreach \x in {1.25,3.5,5.75,8}{
		     \draw (\x,0) -- (\x,1.1);
       \draw (\x,0) arc (-60:60:.63);
		    \node[even] at (\x,0) {};
            \node[even] at (\x,1.1) {};
		  }
		  \node[fill=white] at (2,0) {$\dots$};
        \node[fill=white] at (2,1.1) {$\dots$};
        \node[fill=white] at (4.25,0) {$\dots$};
        \node[fill=white] at (4.25,1.1) {$\dots$};
        \node[fill=white] at (6.5,0) {$\dots$};
        \node[fill=white] at (6.5,1.1) {$\dots$};
        \node[Zltight] at (3.5,0.85) {};
        \node[Zblank] at (3.72,0.85) {};
        \node[] at (3.7,-.15) {$k$};
        \node[Zltight] at (5.75,0.85) {};
        \node[Zblank] at (5.97,0.85) {};
        \node[] at (6,1.25) {$k'$};
        \node[Xltight] at (5.75,0.55) {};
        \node[Xltight] at (3.5,0.55) {};
        \node[Xltight] at (1.25,0.55) {};
        \node[Zltight] at (0.5,0.55) {};
        \node[Zltight] at (7.25,0.55) {};
        \node[] at (0.7,-.15) {$n$};
        \node[] at (7.47,-.15) {$m$};
  \end{tikzpicture} \\
  &=\frac{1}{N^2}\sum_{\substack{
k,k' \\
\text{even}}}(-1)^{s(k)}O_\Delta,
\end{split}
\end{align}
where we have defined $s(k)$ to be $1$ if $n\leq k\leq m$ and $0$ otherwise and used the fact that $\mathcal{S}^S_{nm}\rhopp=\rhopp$. This sum is nearly the same as that appearing in~\cref{eq:O delta sum}, but some terms now appear with a minus sign. For a particular choice of $\Delta$, there will be $|m-n|$ terms for which $k$ lies between $n$ and $m$ and therefore for which $s(k)=1$. We can therefore write
\begin{align}\label{eq:S delta}
\begin{split}
    \Tr[\rhomm\mathcal{S}^S_{nm}\rhomm^\dagger]&=\frac{1}{N^2}\sum_{\Delta=0}^{N-1}\left(N-|m-n|\right)O_\Delta-|m-n|O_\Delta\\
    &=\frac{N-2|m-n|}{N^2}\sum_{\Delta=0}^{N-1}O_\Delta.
\end{split}
\end{align}
Combining~\cref{eq:O delta sum} and~\cref{eq:S delta}, we find
\begin{align}
\begin{split}
    \frac{\Tr[\rhomm\mathcal{S}^S_{nm}\rhomm^\dagger]}{\Tr[\rhomm\rhomm^\dagger]}&=\left[\frac{N-2|m-n|}{N^2}\sum_{\Delta=0}^{N-1}O_\Delta\right]\left[\frac{1}{N}\sum_{\Delta=0}^{N-1}O_\Delta\right]^{-1}\\
    &=\frac{N-2|m-n|}{N}\to1\text{ as }N\to\infty.
\end{split}
\end{align}

We can perform a similar calculation of the two-copy expectation value of the weak-string order parameter. Consider
\begin{align}
\begin{split}
    \Tr[\rhomm\mathcal{S}^W_{nm}\rhomm^\dagger]&=\frac{1}{N^2}\sum_{\substack{
k,k' \\
\text{even}}}\begin{tikzpicture}
    \draw (0,0) rectangle (8.5,-0.35);
    \draw (0,1.1) rectangle (8.5,1.45);
    \foreach \x in {0.5,2.75,5,7.25}{
		     \draw (\x,0) -- (\x,1.1);
       \draw (\x,0) arc (-60:60:.63);
		    \node[odd] at (\x,0) {};
            \node[odd] at (\x,1.1) {};
		  }
		  \foreach \x in {1.25,3.5,5.75,8}{
		     \draw (\x,0) -- (\x,1.1);
       \draw (\x,0) arc (-60:60:.63);
		    \node[even] at (\x,0) {};
            \node[even] at (\x,1.1) {};
		  }
		  \node[fill=white] at (2,0) {$\dots$};
        \node[fill=white] at (2,1.1) {$\dots$};
        \node[fill=white] at (4.25,0) {$\dots$};
        \node[fill=white] at (4.25,1.1) {$\dots$};
        \node[fill=white] at (6.5,0) {$\dots$};
        \node[fill=white] at (6.5,1.1) {$\dots$};
        \node[Zltight] at (3.5,0.25) {};
        \node[Zblank] at (3.72,0.25) {};
        \node[] at (3.7,-.15) {$k$};
        \node[Zltight] at (5.75,0.85) {};
        \node[Zblank] at (5.97,0.85) {};
        \node[] at (6,1.25) {$k'$};
        \node[Xltight] at (5,0.55) {};\node[Xblank] at (5.32,0.55) {};
        \node[Xltight] at (2.75,0.55) {};
        \node[Xblank] at (3.07,0.55) {};
        \node[Zltight] at (1.25,0.55) {};
        \node[Zblank] at (1.57,0.55) {};
        \node[Xltight] at (7.25,0.55) {};
        \node[Xblank] at (7.57,0.55) {};
        \node[Zltight] at (8,0.55) {};
        \node[Zblank] at (8.32,0.55) {};
        \node[] at (1.45,-.15) {$n$};
        \node[] at (8.22,-.15) {$m$};
  \end{tikzpicture} \\
  &=\frac{1}{N^2}\sum_{\substack{
k,k' \\
\text{even}}}O_\Delta\\
&=\Tr[\rhomm\rhomm^\dagger],
\end{split}
\end{align}
so that 
\begin{align}
    \frac{\Tr[\rhomm\mathcal{S}^W_{nm}\rhomm^\dagger]}{\Tr[\rhomm\rhomm^\dagger]}&=1.
\end{align}

We can also calculate the expectation value of the 1-copy string order parameters. Again, using vectorized tensor network notation for the density matrix, we have
\begin{align}
\begin{split}
    \Tr[\rhomm\mathcal{S}^S_{nm}]&=\frac{1}{N}\sum_{k \text{ even}}\begin{tikzpicture}
    \draw (0,0) rectangle (8.5,-0.35);
    \foreach \x in {0.5,2.75,5,7.25}{
		     \draw (\x,0) -- (\x,0.55);
       \draw (\x,0) arc (-60:0:.63);
		    \node[odd] at (\x,0) {};
      \draw [black] plot [smooth, tension=0.3] coordinates { (\x,0.54) (\x,0.7) (\x+0.315/2,0.72) (\x+0.315,0.7) (\x+0.315,0.54)};
		  }
		  \foreach \x in {1.25,3.5,5.75,8}{
		     \draw (\x,0) -- (\x,0.55);
       \draw [black] plot [smooth, tension=0.3] coordinates { (\x,0.54) (\x,0.7) (\x+0.315/2,0.72) (\x+0.315,0.7) (\x+0.315,0.54)};
       \draw (\x,0) arc (-60:0:.63);
		    \node[even] at (\x,0) {};
		  }
		  \node[fill=white] at (2,0) {$\dots$};
        \node[fill=white] at (4.25,0) {$\dots$};
        \node[fill=white] at (6.5,0) {$\dots$};
        \node[Zltight] at (0.5,0.55) {};
        \node[Zltight] at (3.5,0.25) {};
        \node[Zblank] at (3.72,0.25) {};
        \node[] at (3.7,-.15) {$k$};
        \node[Xltight] at (5.75,0.55) {};
        \node[Xltight] at (3.5,0.55) {};
        \node[Xltight] at (1.25,0.55) {};
        \node[Zltight] at (7.25,0.55) {};
        \node[] at (0.7,-.15) {$n$};
        \node[] at (7.47,-.15) {$m$};
  \end{tikzpicture} \\
  &=\frac{1}{N}\sum_{k\text{ even}}(-1)^{s(k)}\begin{tikzpicture}
    \draw (0,0) rectangle (8.5,-0.35);
    \foreach \x in {0.5,2.75,5,7.25}{
		     \draw (\x,0) -- (\x,0.55);
       \draw (\x,0) arc (-60:0:.63);
		    \node[odd] at (\x,0) {};
      \draw [black] plot [smooth, tension=0.3] coordinates { (\x,0.54) (\x,0.7) (\x+0.315/2,0.72) (\x+0.315,0.7) (\x+0.315,0.54)};
		  }
		  \foreach \x in {1.25,3.5,5.75,8}{
		     \draw (\x,0) -- (\x,0.55);
       \draw [black] plot [smooth, tension=0.3] coordinates { (\x,0.54) (\x,0.7) (\x+0.315/2,0.72) (\x+0.315,0.7) (\x+0.315,0.54)};
       \draw (\x,0) arc (-60:0:.63);
		    \node[even] at (\x,0) {};
		  }
		  \node[fill=white] at (2,0) {$\dots$};
        \node[fill=white] at (4.25,0) {$\dots$};
        \node[fill=white] at (6.5,0) {$\dots$};
        \node[Zltight] at (3.5,0.25) {};
        \node[Zblank] at (3.72,0.25) {};
        \node[] at (3.7,-.15) {$k$};
  \end{tikzpicture} \\
  &=\frac{1}{N}\sum_{k\text{ even}}(-1)^{s(k)}\\
    &=\frac{N-2|m-n|}{N}\to1\text{ as }N\to\infty,
\end{split}
\end{align}
where we have once again defined $s(k)$ to be $1$ if $n\leq k\leq m$ and $0$ otherwise and used the fact that $\mathcal{S}^S_{nm}\rhopp=\rhopp$.

Finally, we have 
\begin{align}
\begin{split}
    \Tr[\rhomm\mathcal{S}^W_{nm}]&=\frac{1}{N}\sum_{k \text{ even}}\begin{tikzpicture}
    \draw (0,0) rectangle (8.5,-0.35);
    \foreach \x in {0.5,2.75,5,7.25}{
		     \draw (\x,0) -- (\x,0.55);
       \draw (\x,0) arc (-60:0:.63);
		    \node[odd] at (\x,0) {};
      \draw [black] plot [smooth, tension=0.3] coordinates { (\x,0.54) (\x,0.7) (\x+0.315/2,0.72) (\x+0.315,0.7) (\x+0.315,0.54)};
		  }
		  \foreach \x in {1.25,3.5,5.75,8}{
		     \draw (\x,0) -- (\x,0.55);
       \draw [black] plot [smooth, tension=0.3] coordinates { (\x,0.54) (\x,0.7) (\x+0.315/2,0.72) (\x+0.315,0.7) (\x+0.315,0.54)};
       \draw (\x,0) arc (-60:0:.63);
		    \node[even] at (\x,0) {};
		  }
		  \node[fill=white] at (2,0) {$\dots$};
        \node[fill=white] at (4.25,0) {$\dots$};
        \node[fill=white] at (6.5,0) {$\dots$};
        \node[Zltight] at (3.5,0.25) {};
        \node[Zblank] at (3.72,0.25) {};
        \node[] at (3.7,-.15) {$k$};
        \node[Xltight] at (5,0.55) {};
        \node[Xblank] at (5.32,0.55) {};
        \node[Xltight] at (2.75,0.55) {};
        \node[Xblank] at (3.07,0.55) {};
        \node[Zltight] at (1.25,0.55) {};
        \node[Zblank] at (1.57,0.55) {};
        \node[Xltight] at (7.25,0.55) {};
        \node[Xblank] at (7.57,0.55) {};
        \node[Zltight] at (8,0.55) {};
        \node[Zblank] at (8.32,0.55) {};
        \node[] at (1.45,-.15) {$n$};
        \node[] at (8.22,-.15) {$m$};
  \end{tikzpicture} \\
  &=\frac{1}{N}\sum_{k\text{ even}}\begin{tikzpicture}
    \draw (0,0) rectangle (8.5,-0.35);
    \foreach \x in {0.5,2.75,5,7.25}{
		     \draw (\x,0) -- (\x,0.55);
       \draw (\x,0) arc (-60:0:.63);
		    \node[odd] at (\x,0) {};
      \draw [black] plot [smooth, tension=0.3] coordinates { (\x,0.54) (\x,0.7) (\x+0.315/2,0.72) (\x+0.315,0.7) (\x+0.315,0.54)};
		  }
		  \foreach \x in {1.25,3.5,5.75,8}{
		     \draw (\x,0) -- (\x,0.55);
       \draw [black] plot [smooth, tension=0.3] coordinates { (\x,0.54) (\x,0.7) (\x+0.315/2,0.72) (\x+0.315,0.7) (\x+0.315,0.54)};
       \draw (\x,0) arc (-60:0:.63);
		    \node[even] at (\x,0) {};
		  }
		  \node[fill=white] at (2,0) {$\dots$};
        \node[fill=white] at (4.25,0) {$\dots$};
        \node[fill=white] at (6.5,0) {$\dots$};
        \node[Zltight] at (3.5,0.25) {};
        \node[Zblank] at (3.72,0.25) {};
        \node[] at (3.7,-.15) {$k$};
  \end{tikzpicture} \\
  &=1,\\
\end{split}
\end{align}
where we have again used the fact that $\mathcal{S}^W_{nm}\rhopp=\rhopp$. We see that the density matrix $\rhomm$ exhibits non-trivial string order under the one- and two-copy string order parameters used to characterize the decohered cluster state.

\section{Some details of DMRG}
\label{app:dmrg}

In this appendix, we provide some details behind how we obtain the steady states of the Lindbladians via DMRG.
First, we vectorize the density matrices by the following mapping:
% \begin{align}
\begin{multline}
    \rho = \sum_{\substack{s_1, s_2, \ldots, s_{2N} \\ s_1', s_2', \ldots, s_{2N}' } } \rho_{s_1,  \ldots, s_{2N}, s_1',  \ldots, s_{2N}' } \ket{s_1, \ldots, s_{2N} } \bra{s_1', \ldots, s_{2N}'} \\
    \mapsto  \sum_{\substack{s_1, s_2, \ldots, s_{2N} \\ s_1', s_2', \ldots, s_{2N}' } } \rho_{s_1,  \ldots, s_{2N}, s_1',  \ldots, s_{2N}' } \ket{s_1,s_1', s_2,s_2', \ldots, s_{2N}, s_{2N}' },
    \end{multline}
% \end{align}
Note that we choose the $s_1, s_1', s_2, s_2', \ldots, s_{2N}, s_{2N}'$ ordering of the qubits in the doubled Hilbert space as opposed to $s_1, s_2,  \ldots, s_{2N}, s_1', s_2', \ldots, s_{2N}'$ ordering.
This is done to preserve the locality of the vectorized Lindbladian $\mathbb{L}$.
Using this mapping, we construct the vectorized Lindbladian $\mathbb{L}$ and provide it as an input to the \texttt{dmrg} function in ITensors with $\texttt{ishermitian = false}$ argument.
Setting this argument gives the eigenstate of the vectorized Lindbladian that has the smallest real part of the eigenvalue.
Since, for the steady state, we have $\mathcal{L} \rho = 0$, and the eigenvalues of $\mathcal{L}$ are nonnegative, the state obtained after DMRG has converged is one of the steady states.
To obtain the steady state in the $s_{\text{ket}} = s_{\text{bra}} = w = +1$ sector, we project the output of DMRG into that sector by calculating $\frac{1}{8}(\IId  + U_{\text{even}}^{\text{ket}}) ((\IId  + U_{\text{even}}^{\text{bra}})(\IId  + U_{\text{odd}})\ket{\rho}$, where
\begin{equation}
\begin{aligned}
    U_{\text{even}}^{\text{ket}} & = I_1 I_{1'} X_2 I_{2'} \;\; I_3 I_{3'} X_4 I_{4'} \;\; \ldots X_{2N} I_{2N'} \\    U_{\text{even}}^{\text{bra}} & = I_1 I_{1'} I_2 X_{2'} \;\; I_3 I_{3'} I_4 X_{4'} \;\; \ldots I_{2N} X_{2N'}\\
    U_{\text{odd}} & = X_1 X_{1'} I_2 I_{2'} \;\;  X_3 X_{3'} I_4 I_{4'} \;\; \ldots X_{2N-1}, X_{(2N_1)'} I_{2N} I_{2N'} ~.
\end{aligned}
\end{equation}
To make sure that the output of DMRG, say $\ket{\rho_{\text{out}}}$, is an eigenstate with a good enough numerical precision, we also compute the variance of the vectorized Lindbladian $\mathbb{L}$ of the state $\ket{\rho_{\text{out}}}$ as
\begin{equation}
    \langle \rho_{\text{out}} | \mathbb{L}^\dagger \mathbb{L} |  \rho_{\text{out}} \rangle - | \langle  \rho_{\text{out}} | \mathbb{L}|  \rho_{\text{out}} \rangle |^2~,
\end{equation}
where $\langle \rho_{\text{out}} | \rho_{\text{out}} \rangle = 1$. If $\ket{\rho_{\text{out}}}$ is an eigenstate of $\mathbb{L}$, then the variance should be 0. 
Thus calculating the variance serves as a good check to ensure that we have an eigenstate of $\mathbb{L}$, which typically has a value of about $10^{-10}$ in our calculations.

\section{Other perturbations}
\label{app:other_perturbations}

\begin{figure*}[htbp]
\captionsetup[subfigure]{labelformat=empty,captionskip=-20pt}
 \centering
 \subfloat[\label{fig:cs1_other_pert}]{%
  \includegraphics[width=0.33\textwidth]{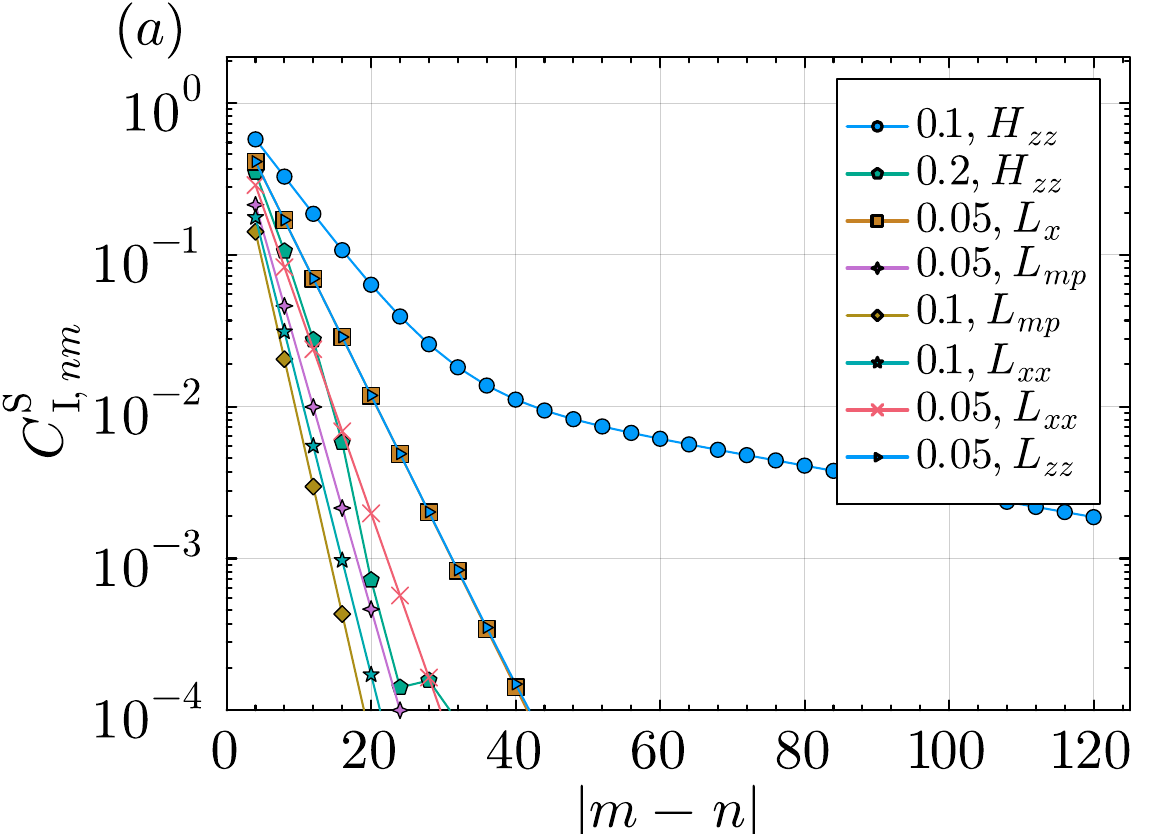}%
}
     \hfill
     \subfloat[\label{fig:cs2_other_pert}]{%
     \includegraphics[width=0.33\textwidth]{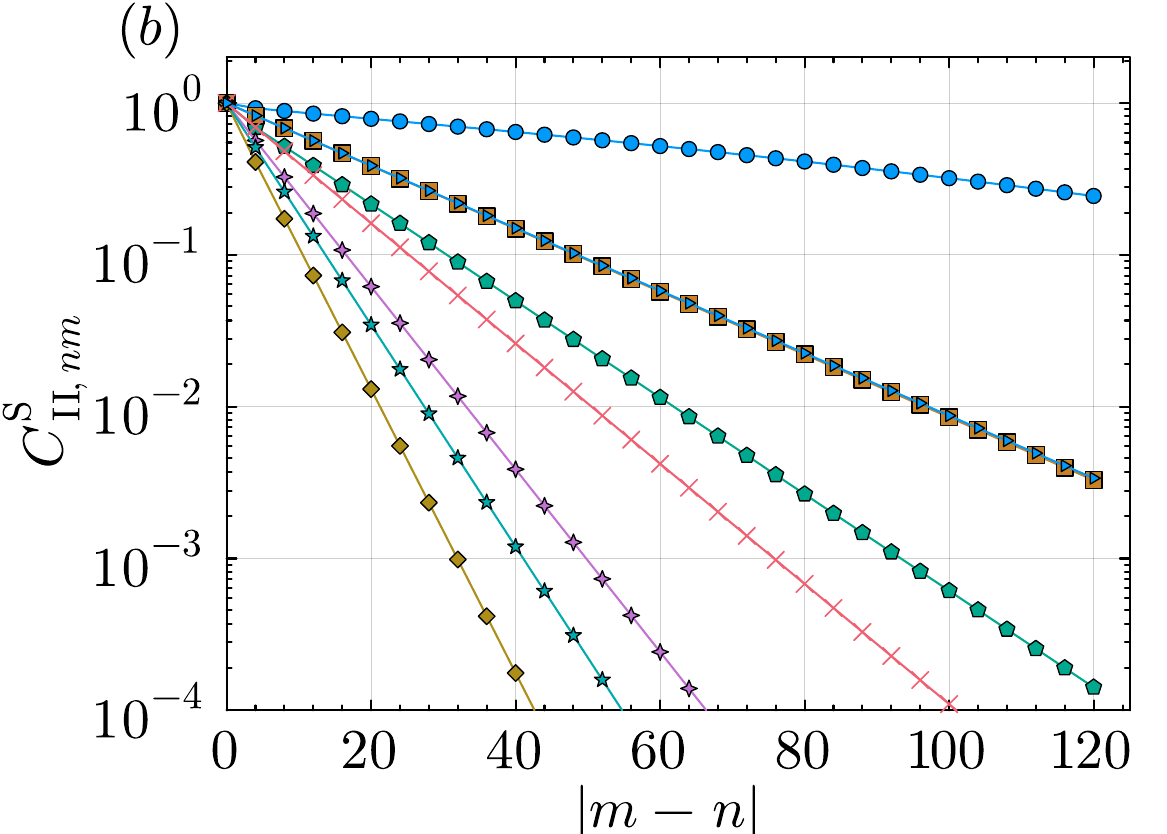}%
     }
    \hfill
     \subfloat[\label{fig:cw2_other_pert}]{%
     \includegraphics[width=0.33\textwidth]{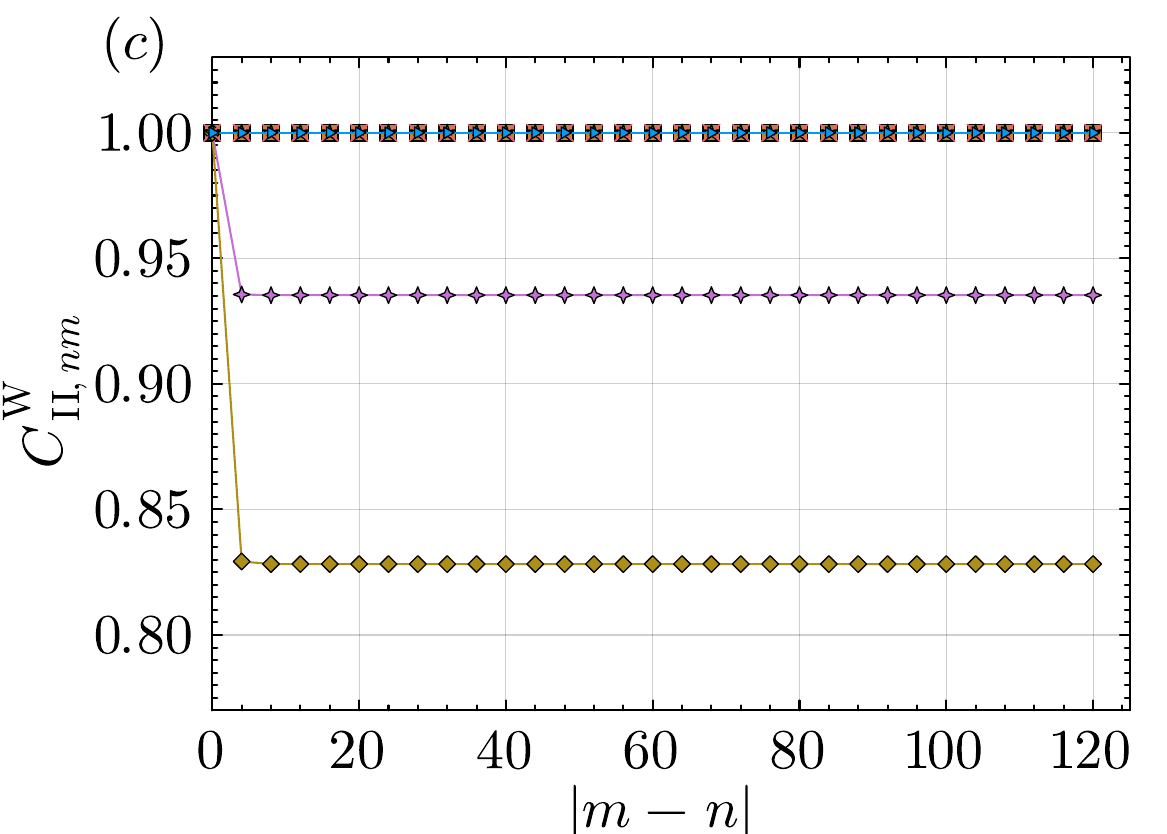}%
     }
    \caption{ $(a)$ \renyi-1 strong-string correlator $C_{\text{I}, nm}^S$ as a function of the length of the string. $C_{\text{I}, nm}^S \rightarrow 0$ with increasing $|m-n|$ for any of the perturbations considered. The legend shows the perturbations that we consider. The numbers in the legend are the perturbation strengths. These are the coefficients with which the perturbation is added to the superoperator (and not the coefficient of the jump operators). For example, ''$0.05, L_{mp}$'' corresponds to the Lindbladian $\mathcal{L} = \mathcal{L}_{\mathcal{C}} + 0.05 \sum_{j} \mathcal{D} \left[ L_{mp, j} \right] $.
    $(b)$ \renyi-2 strong-string correlator $C_{\text{II}, nm}^S$ as a function of the length of the string. $C_{\text{II}, nm}^S$ decays exponentially with $|m-n|$. The curve for $\lambda = 1$ is not visible within the limits of the plots because $C_{\text{II}, nm}^S=0$. We see that it decays exponentially to zero for all perturbations. 
    $(c)$ \renyi-2 weak-string correlator $C_{\text{II}, nm}^W$ as a function of the length of the string. $C_{\text{II}, nm}^W$ remains nonzero for all perturbations. }
    \label{fig:other_pert}
\end{figure*}
\begin{figure}[htbp]
     \centering
     \includegraphics[width=0.37\columnwidth]{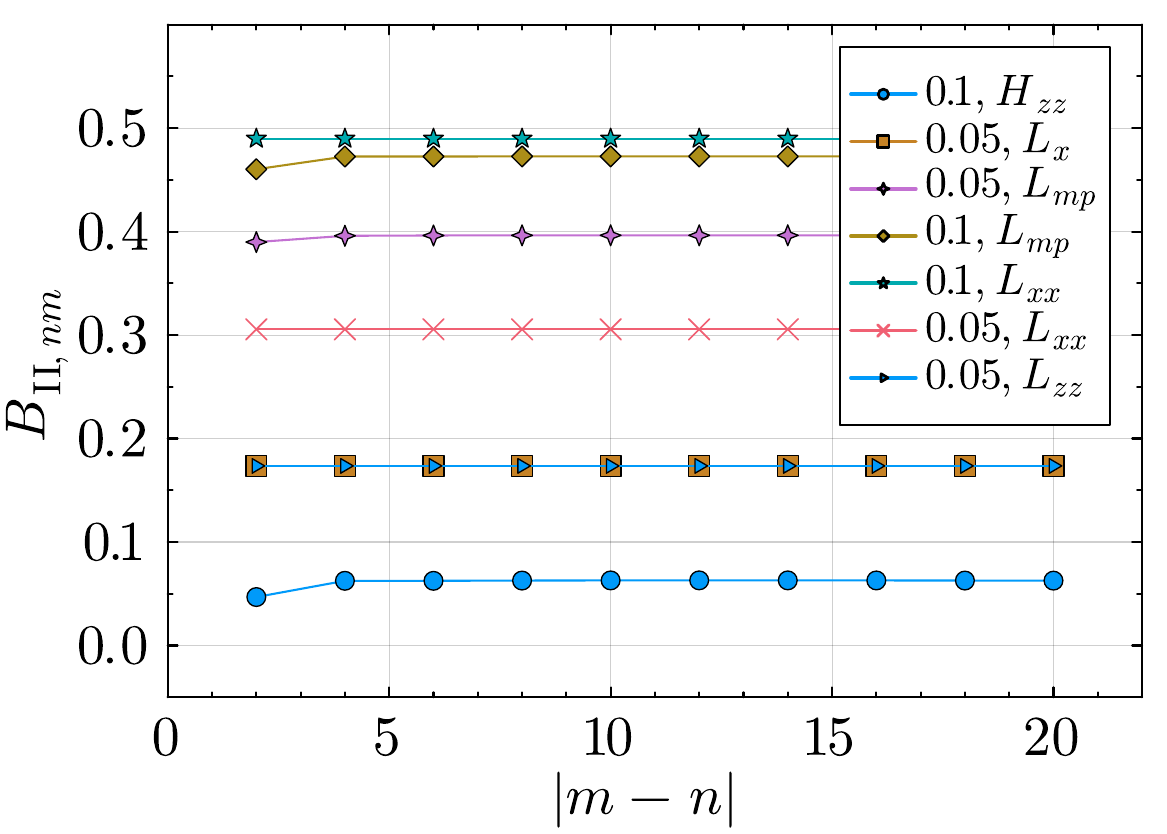}
     \caption{$B_{\text{II}, nm}$ for various other perturbations that we have tried. We find that the correlator saturates to a nonzero constant as $|m-n|$ is increased indicating that the strong symmetry is broken. The legend indicates the perturbation strength and the perturbation type. Similar to~\cref{fig:other_pert}, the numbers in the legend indicate the coefficient with which the perturbed superoperator is added to $\mathcal{L}_{\mathcal{C}}$ and not the coefficient of the jump operator. }
     \label{fig:b2_other_pert}
\end{figure}

In this appendix, we present the results for other symmetric perturbations that we have examined, which include Hamiltonian perturbations as well as Lindbladian perturbation.
Specifically, we consider the Hamiltonian perturbation $H_{zz} = \sum_{j} Z_{2j} Z_{2j+2}$ and the Lindbladian perturbation of the form 
$\sum_j \gamma _{\alpha} \mathcal{D} \left[ L_{\alpha, j} \right]$, where the jump operators $L_{\alpha, j} $ can take one of the following (with $\alpha \in \{x, mp, xx, zz\} $):
\begin{align}
L_{x, j} = X_{2j+1}, \quad 
L_{mp,2j} = \ket{-}\bra{+}_{2j} \otimes \ket{-} \bra{+}_{2j+2}, \quad 
L_{xx, j} = X_j X_{j+1} , \quad 
L_{zz, j} = Z_{2j} Z_{2j+2}~, 
\end{align}
and the sum is over all the allowed values of $j$. These perturbations are added separately with a small perturbation strength to the parent Lindbladian $\mathcal{L}_{\mathcal{C}}$. It is interesting to note that in the CZ dual picture, $L_{x, j} \to \tilde{L}_{x, j}  = L_{2,j}=Z_{j-1}X_{j}Z_{j+1}$ [see Eq.~\eqref{eq:parent_jumps} in the main text], while $L_{zz, j} $ stays invariant as $L_{zz, j} \to \tilde{L}_{zz, j} = Z_{2j} Z_{2j+2} $, so that those two perturbations generate identical strong string order correlators; this agrees well with the numerical DMRG calculations, see Fig.~\ref{fig:other_pert} and the following discussions.

We calculate the steady state of the perturbed Lindbladian in the $s_{\text{ket}} = s_{\text{bra}} = w = +1$ symmetry sector using DMRG. 
We then calculate the six string order parameters defined in~\cref{sec:string_order_numerics} and the three connected correlators defined in~\cref{sec:ssb} to probe SW-SSB.
We find that all the correlators indicate strong-to-weak SSB on even sites, with some of the correlators plotted in~\cref{fig:other_pert,fig:b2_other_pert}.
For these plots, the system size is $2N = 300$ qubits, and the left end of the correlator for $C^S_{\text{I}, nm}, C^S_{\text{II}, nm}$ is kept fixed at $n=101$ while the right end is varied in $m\in \{ 103, 105, \ldots , 221\} $. 
For $C_{\text{II}, nm}^W$, we fix $n=102$ and vary $m \in \{ 104, 106, \ldots, 222\}$; for $B_{\text{II}, nm}$, we fix $n = 140$ and vary $m\in\{142, 144, \ldots, 160\}$.
We find that the connected correlators $A_{\text{I}, nm}$ and $A_{\text{II},nm}$ are identically zero.
Also, unlike the $\mathcal{L}_\lambda$ Lindbladian considered in~\cref{sec:ssb}, the $B_{\text{II}, nm}$ correlator is independent of $|m-n|$  for a generic perturbation as shown in~\cref{fig:b2_other_pert}.
In conclusion, our numerical results suggest that the steady-state mixed-state SPT state $\rho_{\mathcal{C}}$ is generically unstable to SW-SSB (unless the perturbation only introduces weak symmetry defects like the one in~\cref{sec:weak_defects}).

\section{Finite-size scaling of string order parameters}

The results of~\cref{sec:string_order_numerics,sec:ssb} show that the steady-state mixed-state SPT order is unstable for any small perturbation $\lambda > 0$, leading to strong-to-weak SSB.
In this appendix, we perform finite-size scaling of the \renyi-1 strong-string order parameter $C_{\text{I}, nm}^S$ which further shows the scaling relation between the perturbation strength and system size $N$. 
First, we show our scaling collapse result for the following Lindbladian:
\begin{equation}
    \mathcal{L}_{\gamma_{\text{ZZ}}} = \mathcal{L}_\mathcal{C} + \gamma_{\text{ZZ}} \sum_{j=1}^{N-1} \mathcal{D} [Z_{2j}Z_{2j+2}]~.
\end{equation} 

\begin{figure}[htbp]
\captionsetup[subfigure]{labelformat=empty,captionskip=-20pt}
 \centering
 \subfloat[\label{fig:scaling_collapse_cs1_Lzz}]{%
 \hfill
  \includegraphics[width=0.43\textwidth]{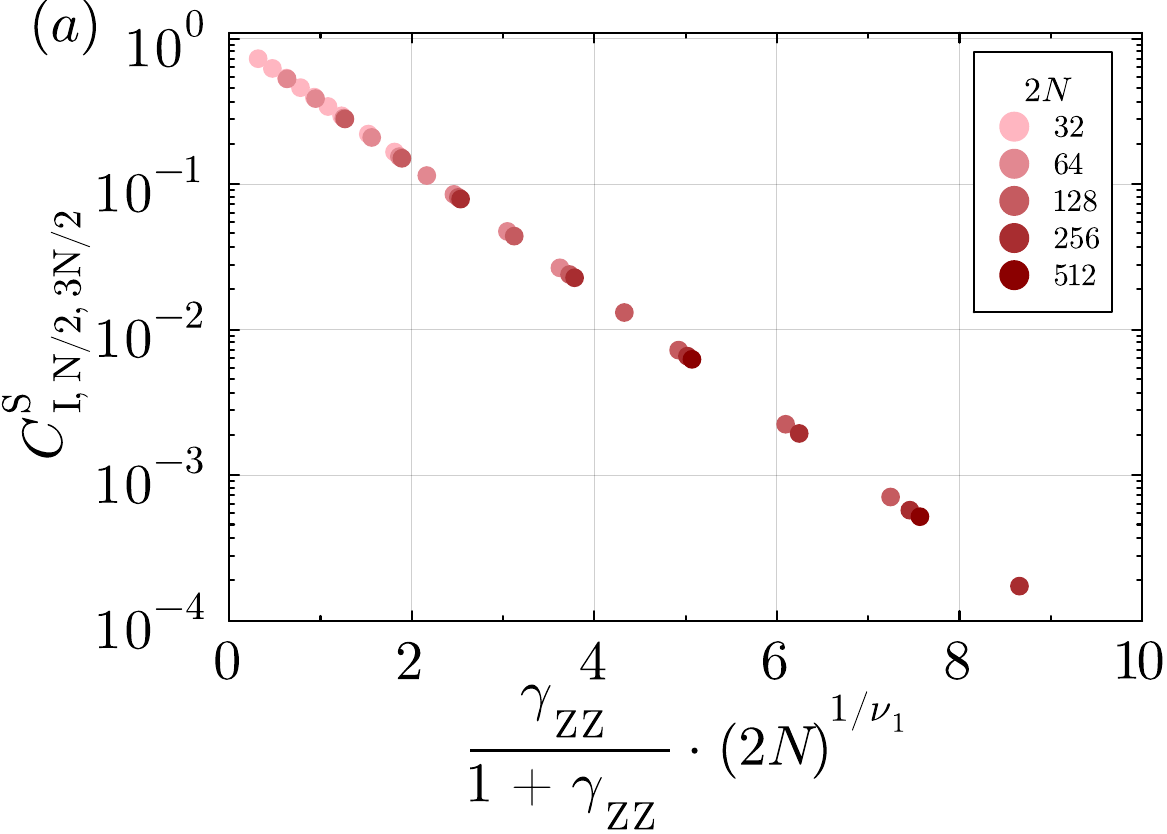}%
}
     \hfill
     \subfloat[\label{fig:scaling_collapse_cs2_Lzz}]{%
     \includegraphics[width=0.43\textwidth]{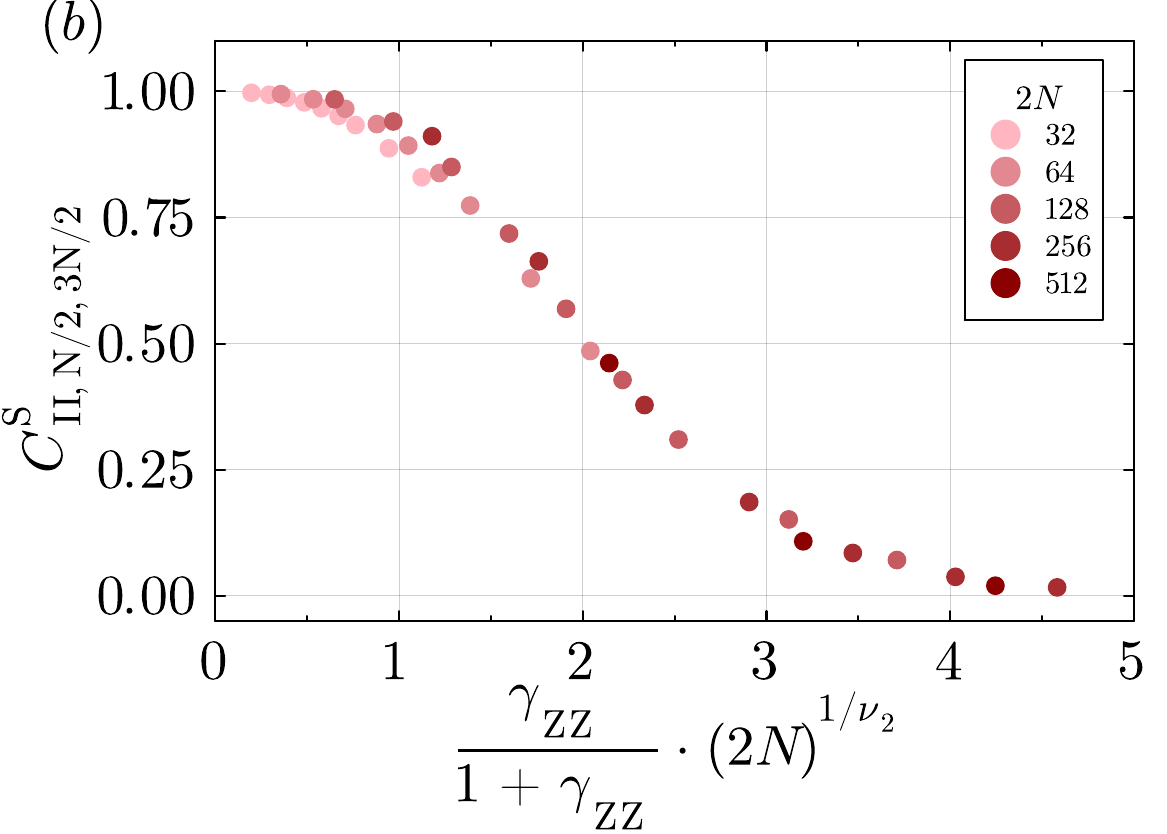}%
     }
     \hfill
    \caption{ Scaling collapse of $(a)$ \renyi-1 $C_{\text{I}, nm}^S$ and $(b)$ \renyi-2 $C_{\text{II}, nm}$ strong string correlators of the steady state of $\mathcal{L}_{\gamma_{\text{ZZ}}}$. We fix the end points $n$ and $m$ to be $N/2$ and $3N/2$, respectively. Here we find $\nu_1 =1$ in $(a)$ and $\nu_2 = 1.16$ in $(b)$ from the scaling collapse.}
    \label{fig:scaling_collapse_Lzz}
\end{figure}
For a range of system sizes that vary from $2N = 32$ to $2N = 512$ and perturbation strengths $\gamma_{\text{ZZ}} \leq 0.06$, we calculate the strong-string order parameters $C_{\text{I}, nm}^S$ and $C_{\text{II}, nm}^S$, where $n = N/2$ and $m = 3N/2$ so the string length is $|n-m|=N$.     
We use the following scaling ansatz:
\begin{equation}
    C_{\text{I}, N/2, 3N/2}^S = \mathcal{F} \left( \frac{\gamma_{\text{ZZ}}}{1 + \gamma_{\text{ZZ}}}  (2N)^{1/\nu_1} \right), \quad\quad\quad C_{\text{II}, N/2, 3N/2}^S = \mathcal{G} \left( \frac{\gamma_{\text{ZZ}}}{1 + \gamma_{\text{ZZ}}}  (2N)^{1/\nu_2} \right)~,
\end{equation}
for some function $\mathcal{F}$.
We find good data collapse for $\nu_1=1$ as shown in~\cref{fig:scaling_collapse_cs1_Lzz}, which also shows that $\mathcal{F}(x)$ decays exponentially with $x$.
Note that here we simply try $\nu_1 = 1$ and find good collapse rather than fitting the data to a function and finding the value of $\nu_1$ for which the fit is the best.
For the \renyi-2 correlator, $\C_{\text{II}, nm}^S$, we fit a sixth degree polynomial and find the value of $\nu_2$ that minimizes the residual error of the fit.
Doing so gives us the optimal value of $\nu_2 = 1.16$. and we find that the data collapses for this value of $\nu_2$ as shown in~\cref{fig:scaling_collapse_cs2_Lzz}.

We perform a similar analysis for the interpolated Lindbladian $\mathcal{L}_{\lambda}$ in~\cref{eq:interpolating_lind}, and we employ the scaling ansatz
\begin{equation}
    C_{\text{I}, N/2, 3N/2}^S = \mathcal{F} \left( \lambda (2N)^{1/\nu_1} \right),\quad\quad\quad C_{\text{II}, N/2, 3N/2}^S = \mathcal{G} \left( \lambda (2N)^{1/\nu_2} \right)~.
\end{equation}
Again we find an excellent collapse for $\nu_1 = 1$ as shown in ~\cref{fig:scaling_collapse_cs1_Lzz} without any fitting.
For $C_{\text{II, nm}}^S$, we again minimize the residual error in fitting a sixth degree polynomial to the data to determine the optimal value of $\nu_2$ at which the collapse happens.
We find the optimal value of $\nu_2$ to be  $1.19 $, and the data collapses reasonably well as shown in~\cref{fig:scaling_collapse_cs2_Lzz}.

\begin{figure}[htbp]
\captionsetup[subfigure]{labelformat=empty,captionskip=-20pt}
 \centering
 \subfloat[\label{fig:scaling_collapse_cs1_lam}]{%
 \hfill
  \includegraphics[width=0.43\textwidth]{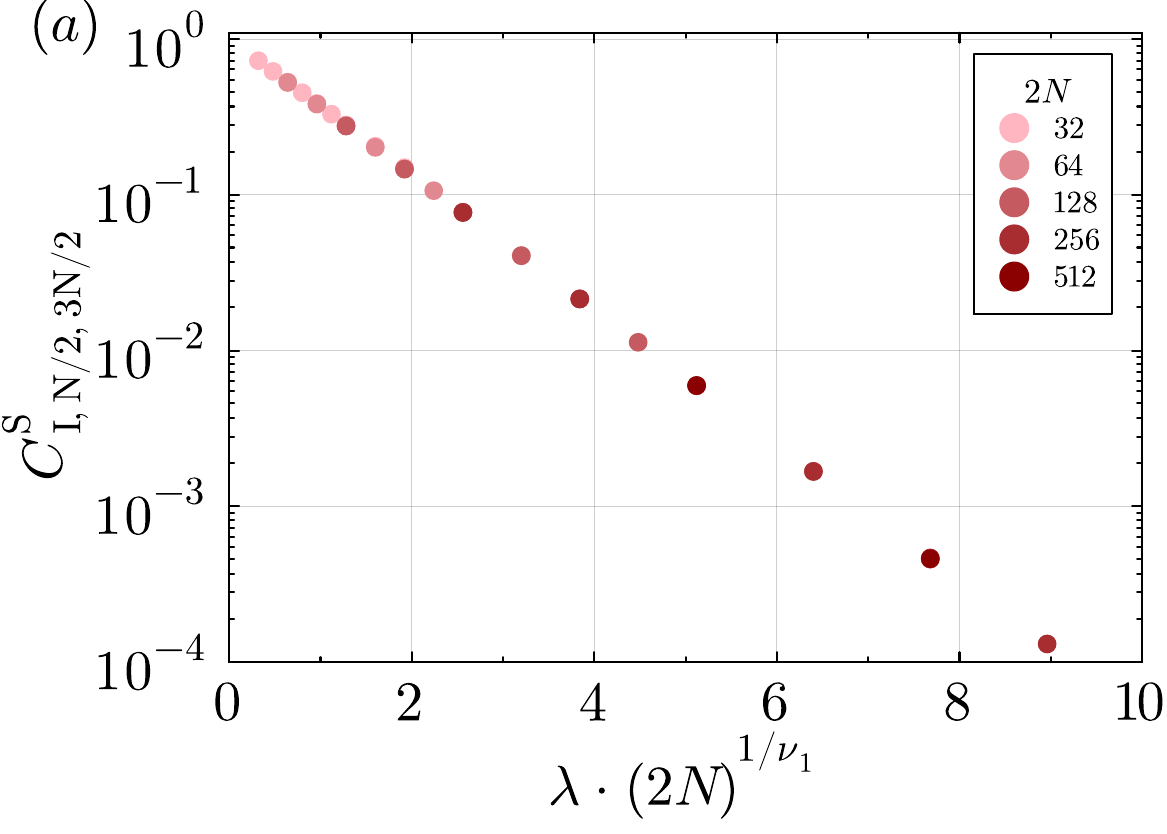}%
}
     \hfill
     \subfloat[\label{fig:scaling_collapse_cs2_lam}]{%
     \includegraphics[width=0.43\textwidth]{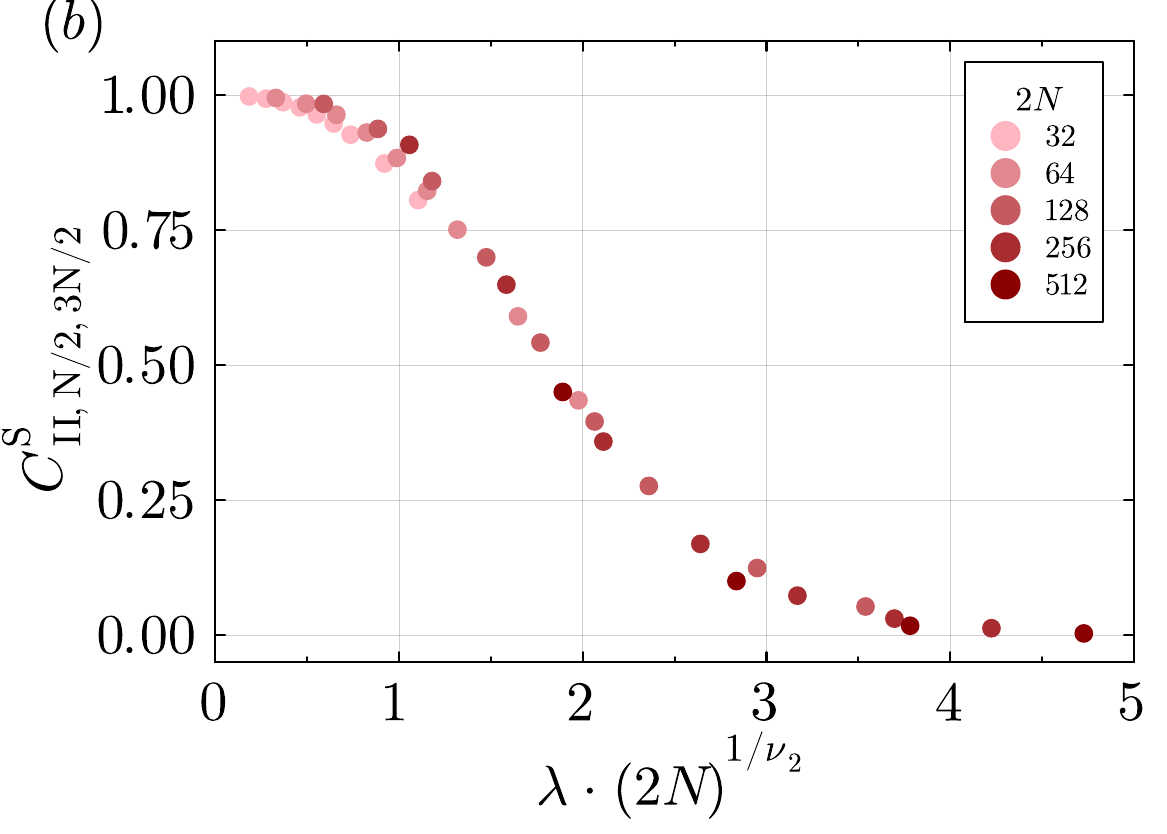}%
     }
     \hfill
    \caption{ Scaling collapse of $(a)$ \renyi-1 $C_{\text{I}, nm}^S$ and $(b)$ \renyi-2 $C_{\text{II}, nm}$ strong string correlators of the steady state of $\mathcal{L}_{\lambda}$. We fix the end points $n$ and $m$ to be $N/2$ and $3N/2$, respectively. Here we find $\nu_1 =1$ in (a) and $\nu_2 = 1.19$ in (b).}
    \label{fig:scaling_collapse_lam}
\end{figure}

\section{Free fermion solution to the Lindbladian dynamics with ZZ perturbation}

\label{app_sec:fermion}

As mentioned in the main text, dynamics generated by the parent Lindbladian as well as a certain class of perturbations can be exactly solved by the mapping to a non-Hermitian free fermion Hamiltonian. In this appendix, we discuss in more detail  the solution and its implications for mixing times and various order parameters.

We consider the following Lindbladian acting on all the even sites:
\begin{align}
\mathcal{L}_{\text{even}} &= \sum_{\substack{j=1}}^{N} \mathcal{L}_{\text{even},j} 
, \\
\mathcal{L}_{\text{even},j} 
&= (1-\lambda ) 
\gamma _{j} \eta _{\leftarrow} \mathcal{D} [
Z_{2j} Z_{2j+2}
( \mathbbm{1} - X_{2j}) /2]
\nonumber \\
& +(1-\lambda ) 
\gamma _{j} \eta _{\rightarrow} \mathcal{D} [
Z_{2j-2} Z_{2j}
( \mathbbm{1} - X_{2j}) /2] 
\nonumber \\
& + \lambda  \tilde{\gamma} _{2,j} 
\mathcal{D} [Z_{2j-2} Z_{2j}]
. 
\end{align}
Our goal is to solve for the spectrum of $\mathcal{L}_{\text{even}} $. To achieve this, we first note that, since the Lindbladian has a weak symmetry generated by $X_{2j}$ for all $j$ (which is a subgroup of the strong symmetry $Z_2^{S}$), the dynamics generated by the Lindbladian is decoupled between the $X_{2j}$ basis populations (i.e., the diagonal part) versus the coherences (i.e., the off-diagonal part) of the density matrix. Therefore, we can exactly map the action of the Lindbladian to a non-Hermitian Hamiltonian acting on the vectors spanned by the corresponding basis states to the $X_{2j}$ population or the coherence sectors. For example, focusing on the population sector that also hosts the even-parity steady state, we consider the following mapping (introducing the $X$ basis states $X_{2j} | \pm \rangle _{2j} = \pm | \pm \rangle _{2j} $, and we use $\left | \psi \right ) $ to distinguish the mapped state from the pure state wavefunction of the original system):
\begin{align}
\label{seq:LtonHM.map}
| + \rangle _{2j} \langle + |  \Rightarrow 
\left | \downarrow \right ) _{j}
, \quad 
| - \rangle _{2j} \langle - |  \Rightarrow 
\left | \uparrow \right )_{j}
, 
\end{align}
in which case the Lindbladian equation of motion can be mapped to a linear equation of motion of the state vector $|\rho_t)$:
\begin{align}
\partial_t \rho_t = \mathcal{L}_{\text{even}}  \rho_t
\Rightarrow 
 \partial_t |\rho_t) = \mathcal{M}_{\text{even}}  |\rho_t)~. 
\end{align}
Note that the equation above takes a form similar to the standard Schr{\"{o}}dinger equation, but with a non-Hermitian dynamical generator (which in this case is given by the non-Hermitian operator $i\mathcal{M}_{\text{even}} $).
The generator $\mathcal{M}_{\text{even}} $ of non-Hermitian dynamics is given by 
\begin{align}
&\mathcal{M}_{\text{even}} = \sum_{\substack{j=1}}^{N} \mathcal{M}_{\text{even},j} \nonumber \\
\label{seq:nHferm.dynmat}
=& (1-\lambda ) \sum_{\substack{j=1}}^{N} 
\gamma _{j} ( \eta _{\rightarrow} \mathcal{M}_{\rightarrow,j} +\eta _{\leftarrow} \mathcal{M}_{\leftarrow,j} )
+ \lambda  \tilde{\gamma} _{2,j} \mathcal{M}_{ZZ,j} 
,  
\end{align}
which can be written in terms of the spin raising and lowering operators $\hat \sigma  _{j}^{-} = \left | \downarrow  ) _{j}  ( \uparrow \right | 
= (\hat \sigma  _{j}^{+} ) ^{\dag} $ as
\begin{align}
\mathcal{M}_{\rightarrow,j} =& 
\hat \sigma  _{j}^{-} \hat \sigma  _{j+1}^{+} +\hat \sigma  _{j}^{-} \hat \sigma  _{j+1}^{-}  - \hat \sigma  _{j}^{+} \hat \sigma  _{j}^{-} 
\\
\mathcal{M}_{\leftarrow,j} =& \hat \sigma  _{j}^{-} \hat \sigma  _{j-1}^{+} +\hat \sigma  _{j}^{-} \hat \sigma  _{j-1}^{-}  - \hat \sigma  _{j}^{+} \hat \sigma  _{j}^{-} 
, \\
\mathcal{M}_{ZZ,j} = &  
(\hat \sigma  _{j-1}^{-} +\hat \sigma  _{j-1}^{+} )(\hat \sigma  _{j}^{-} +\hat \sigma  _{j}^{+} ) -  \mathbbm{1} 
. 
\end{align}
It is interesting to note that the non-Hermitian Hamiltonian $\mathcal{M}_{\text{even}} $ describes a one-dimensional spin chain that only involves quadratic terms acting on nearest-neighbor sites, which can be mapped to a free fermion non-Hermitian Hamiltonian via the standard Jordan-Wigner (JW) transformation. More concretely, introducing fermion operators via the JW prescription as 
\begin{align}
\hat \sigma  _{j}^{-} = e ^{i \pi 
\sum _{\ell=1} ^{j-1} \hat f _{\ell}^{\dag} \hat f _{\ell}}
\hat f _{j} 
, \, 
\hat \sigma  _{j}^{+} = \hat f _{j}^{\dag} e ^{-i \pi 
\sum _{\ell=1} ^{j-1} \hat f _{\ell}^{\dag} \hat f _{\ell}}
, 
\end{align}
we obtain 
\begin{align}
\mathcal{M}_{\rightarrow,j} =& 
-\hat f _{j}  \hat f _{j+1}^{\dag} 
-\hat f _{j}  \hat f _{j+1}  - \hat f  _{j}^{\dag} \hat f _{j}  
\\
\mathcal{M}_{\leftarrow,j} =& - \hat f _{j}  \hat f _{j-1}^{\dag} 
+ \hat f _{j}  \hat f _{j-1}  - \hat f  _{j}^{\dag} \hat f _{j} 
, \\
\mathcal{M}_{ZZ,j} = &  
(\hat f _{j-1} ^{\dag} -\hat f  _{j-1} )
(\hat f  _{j} + \hat f _{j} ^{\dag}  ) -  \mathbbm{1} 
. 
\end{align}
The spectrum of the free fermion non-Hermitian Hamiltonian $\mathcal{M}_{\text{even}} $ can be generally computed by solving for its right eigenmodes $\hat \beta _{\ell} $, which are the solutions to the following eigenvalue equation: 
\begin{align}
\label{seq:nHferm.eigv.eq}
[ \mathcal{M}_{\text{even}}  , 
\hat \beta _{\ell, \pm} ] =  \mathcal{E} _{\ell,\pm } \hat \beta _{\ell, \pm} 
. 
\end{align}
Here, we make use of the fact that the eigenvalues of 
$ \mathcal{M}_{\text{even}}$ come in pairs $\nu _{\ell,+ } = - \nu _{\ell,- } ^{*}$,  as $ \mathcal{M}_{\text{even}}$ is invariant under the  anti-unitary transformation consisting of the combination of Hermitian conjugation and reflection of the one-dimensional chain $\hat f _{j} \to \hat f _{N+1-j} $. Without loss of generality, we require $\text{Re} \mathcal{E} _{\ell,+ } \ge 0$.

If the lattice is translationally invariant, we can take $ \gamma _{2} = \tilde{\gamma} _{2,j} =1 $ and $\eta _{\leftarrow} + \eta _{\rightarrow} =1$ without loss of generality. The eigenvalue equation Eq.~\eqref{seq:nHferm.eigv.eq} can then be solved analytically by Fourier transforming the real space free fermion to momentum space. After some algebra, we obtain 
\begin{align}
\label{seq:nHferm.eigvals}
\mathcal{E} _{k, \pm }  
=  - i (\eta _{\rightarrow}  - \eta _{\leftarrow} ) (1-\lambda )\sin k 
 \pm  ( 1+ \lambda - \cos k  +\lambda \cos k)~,  
\end{align}
with the corresponding right eigenmode operators given by 
\begin{align}
\hat \beta _{k,+} & \propto  
\sin  \frac{k}{2} \hat f _{k}  
-  i  \lambda  \cos  \frac{k}{2} \hat f _{-k} ^{\dag} 
. 
\end{align}
The steady state as defined by $\mathcal{M}_{\text{even}}  |\rho_{t \rightarrow \infty}) = 0$ can thus be written as 
\begin{align}
\label{seq:nHferm.rhoNESS.def}
|\rho_{t \rightarrow \infty}) = e ^{\sum _{j,\ell} \frac{C _{j\ell} }{2}
\hat f _{j} ^{\dag} \hat f _{\ell} ^{\dag}  }
| \text{vac} ) 
,  
\end{align}
where the coefficients $C _{j\ell}$ can be expressed as 
\begin{align} 
C _{j\ell} = \frac{1}{N } \sum _{k } \frac{i  \lambda  \cos  \frac{k}{2} }{  \sin  \frac{k}{2} }  e ^{ik(j- \ell) }
. 
\end{align}
For a chain with even number of sites, we have $k = \pm \frac{\pi}{N},  \pm \frac{3\pi}{N}
, \ldots ,  \pm \frac{N-1}{N}\pi$, and the above sum over $k$ can be evaluated explicitly as 
\begin{align} 
C _{j\ell} & = \frac{2 \lambda }{ N } \sum _{m=1 } ^{\frac{N}{2}} \frac{
\cos  \frac{2m-1}{2N} \pi \sin  \frac{(2m-1)(j- \ell)}{N} \pi }{\sin \frac{2m-1}{2N} \pi  }  
\\
\label{seq:quad.fermion.cij.sol}
& = \lambda \left[ \delta _{j\ell} -1 \right ]
\text{sgn}(j - \ell) 
. 
\end{align}
Note that, in the unperturbed case ($\lambda=0$), the steady state is $|\rho_{t \rightarrow \infty}) =| \text{vac} )$. 
We also note that the steady state solution given by Eq.~\eqref{seq:quad.fermion.cij.sol} agrees with the results in Ref.~\cite{lushnikov1986binary}, which considered the specific case with $\eta _{\rightarrow} = \eta _{\leftarrow}  $. However, Ref.~\cite{lushnikov1986binary} did not consider the case with bidirectional hoppings in the unperturbed Lindbladian, and to the best of our knowledge, previous works did not compute the correlators in this non-Hermitian free fermion problem, which we show below. 

As shown above, the steady state structure in fact does not depend on the directionality of hopping (as parametrized by $\eta _{\rightarrow} / \eta _{\leftarrow} $). 
Transforming the fermion operators back into the spin ones, we have 
\begin{align}
\label{seq:rhoss.wf.spin}
|\rho_{t \rightarrow \infty}) = e ^{ \lambda  \sum _{j<\ell}  
\hat \sigma  _{j}^{+}  
\prod _{m=j+1} ^{\ell-1} (-\hat \sigma ^{z} _{m} )
\hat \sigma  _{\ell}^{+} }
\left | \downarrow  \downarrow \ldots \downarrow \right  ) 
.  
\end{align}
We can further compute all two-body correlators making use of the solution in Eq.~\eqref{seq:nHferm.rhoNESS.def}, as 
\begin{align}
\frac{ (  \rho_{t \rightarrow \infty} | ( f  _{n} ^{\dag} -  f _{n} ) 
( f  _{m} +  f _{m} ^{\dag}  )  |\rho_{t \rightarrow \infty} )}{( \rho_{t \rightarrow \infty} | \rho_{t \rightarrow \infty} ) }
= & 
\begin{cases}
0 & n >m\\
\frac{\lambda-1}{\lambda+1} & n = m\\
\frac{4 \lambda   }{(1+\lambda) ^{2}}
\left ( \frac{1-\lambda}{1+\lambda} \right ) ^{m-n-1}
& n <m
\end{cases}
. 
\end{align}

\subsection{Mixing time}

We note that the full Lindbladian spectrum in the population sector can be generated by Eq.~\eqref{seq:nHferm.eigvals}. The Lindblad equation of motion for the coherences in the $X_{2j}$ basis can be mapped to a similar non-Hermitian Hamiltonian Eq.~\eqref{seq:nHferm.dynmat} by a modified mapping, which also corresponds to a non-Hermitian free fermion Hamiltonian, and we expect the coherences to decay faster compared to the minimal decay rate in the population sector. We can thus directly compute the spectral (Lindbladian) gap $\Delta _{\mathcal{L}} $, i.e.~the smallest real part of the Lindbladian eigenvalues, via the smallest real part of the free fermion non-Hermitian Hamiltonian eigenvalues $\text{Re} \mathcal{E} _{\ell,+ } \ge 0$. From Eq.~\eqref{seq:nHferm.eigvals}, we can calculate the Lindbladian gap as 
\begin{align} 
\Delta _{\mathcal{L}} (\lambda ) = \text{Re} \mathcal{E} _{k = \pm \frac{N-1}{N}\pi,+ }
= & 1-\cos \frac{\pi}{N}  + \lambda (2+\cos \frac{\pi}{N} )
.  
\end{align}
In the large-system-size limit, the spectral gap of the unperturbed Lindbladian with $\lambda =0$ thus scales as 
\begin{align} 
N \to \infty : \quad 
\Delta _{\mathcal{L}} (\lambda =0) \sim N^{-2}
. 
\end{align}
We note that this also sets a lower bound on the mixing time scaling of the open boundary Lindbladian in the thermodynamic limit, as in the open-boundary-conditions case, relaxation dynamics in the bulk is agnostic to the boundary conditions.

\subsection{Connected correlators \& strong-string order parameter}

It is worth noting that, when the Lindbladian steady state only involves populations in the $X_{2j}$ basis, then the correlators $A_{\text{I}, nm} $, $A_{\text{II}, nm} $ as defined in the main text must be trivially zero. 

Further, we can directly use the wavefunction in the mapping Eq.~\eqref{seq:LtonHM.map} to compute the \renyi-2 correlators and the \renyi-2 strong string order parameters, as by definition, the \renyi-2 expectation value of an operator $O$ can be written in terms of $| \rho  ) $ as 
\begin{equation} 
\langlee  O \otimes  O  \ranglee_\rho 
\vcentcolon=  
\frac{\Tr  ( \rho  O \rho  O^{T} )}{\Tr(\rho ^{2} ) }~
= \frac{ (  \rho | O' |\rho  )}{(  \rho | \rho  ) }~
,
\end{equation}
and expectation averages $( \rho | O|\rho )/( \rho | \rho  )$ of the mapped wavefunction can be computed using the free fermion state by Wick's theorem. As an example, we consider 
\begin{equation} 
B_{\text{II}, nm} \vcentcolon= \langlee  Z_n Z_m \otimes Z_n Z_m  \ranglee  - \langlee  Z_n  \otimes Z_n  \ranglee  \langlee  Z_m \otimes Z_m  \ranglee~,
\end{equation}
and one can straightforwardly show that $\langlee  Z_n  \otimes Z_n  \ranglee =0$ for all even $n$ and any steady state with only populations in the $X_{2j}$ basis, so that we have $B_{\text{II}, nm} = \langlee  Z_n Z_m \otimes Z_n Z_m  \ranglee $. From Eqs.~\eqref{seq:LtonHM.map} and~\eqref{seq:rhoss.wf.spin}, we thus have 
\begin{align}
& B_{\text{II}, nm} = 
\langlee  Z_n Z_{m} \otimes Z_n Z_{m}  \ranglee 
= \frac{ (  \rho | \sigma^{x} _{n} \sigma^{x} _{m} |\rho  )}{(  \rho | \rho  ) }~
, 
\end{align}
which can be rewritten in terms of the fermion operators as
\begin{align}
B_{\text{II}, nm}  
= & \frac{ (  \rho | ( f  _{n} ^{\dag} -  f _{n} ) \prod _{\ell=n+1} ^{m-1} (1-2 f _{\ell}^{\dag}  f _{\ell}) 
( f  _{m} +  f _{m} ^{\dag}  )  |\rho  )}{(  \rho | \rho  ) }
=  \frac{4 \lambda   }{(1+\lambda) ^{2}}
. 
\end{align}

We can also compute the two-copy strong symmetry order parameter, defined in Eq.~\eqref{eq:strong_2copy_nontrivial_correlator}, as 
\begin{align}
& C_{\text{II},nm}^S = 
\langlee \mathcal{S}_{nm}^S \otimes \IId  \ranglee_\rho
=   \frac{ (  \rho | \prod _{\ell=n} ^{m} (1-2 f _{\ell}^{\dag}  f _{\ell})  |\rho  )}{(  \rho | \rho  ) }
= \left ( \frac{1-\lambda }{1+\lambda} \right ) ^{|m-n|}
.  
\end{align}

\section{Overview of perturbation theory}

\label{app_sec:perturbation}

As mentioned in~\cref{sec:exactly_solvable_perturbation} of the main text, the exactly solvable mapping only applies to certain kinds of perturbations.
In this appendix, we discuss a general perturbation theory which can capture the physics of sufficiently small perturbations that cannot be exactly solved by the procedure of~\cref{sec:exactly_solvable_perturbation}. 
In particular, we will start from the CZ dual of the parent Lindbladian $\tilde{\mathcal{L}}_{\mathcal{C}}$ and add a perturbation $\mathcal{L}_{\text{pert}}$ which is a linear combination of the dissipators formed by the jump operators shown in~\cref{eq:dual_jumps} with a small perturbation strength.
Referring to~\cref{sec:dualperspective} of the main text, we define $\rho_0^0 = \tilde \rho_{\mathcal C}$ as the density matrix of the unperturbed state with $2N$ qubits, where $\tilde \rho_{\mathcal C}$ is the density matrix defined in~\cref{eq:CZ decohered cluster}:
\begin{equation}\label{eq:rhoc}
	\begin{aligned}
		\rho_{0}^0 &=  \bigotimes_{i \in {\rm odd}} \frac{{\mathbb I}_{i}}{2} \bigotimes_{j \in {\rm even}}  \ket{+}\bra{+}_{j}  = \bigotimes_{i \in {\rm odd}} \frac{{\mathbb I}_{i}}{2} \bigotimes_{j \in {\rm even}}  (\frac{1}{2} {\mathbb I}_{j} + \frac{1}{2} X_{j})~.
	\end{aligned}
\end{equation}
For simplicity, we only consider up to  the first order in perturbation theory. 
Note that each jump operator of~\cref{eq:dual_jumps}, when acting on $\rho_0^0$ defined in~\cref{eq:rhoc}, will either leave $\rho_0^0$ unchanged or give the following matrices:
\begin{equation}\label{eq:perturbation_ij}
	\rho_{i,j}^0 =   \bigotimes_{m \in {\rm odd}} \frac{{\mathbb I}_m}{2} \bigotimes_{n \in \rm{even}} F_n~,
\end{equation}
where $i< j$, and 
\begin{equation}
	F_n = \left\{ 
	\begin{aligned}
		&(\ket{+}\bra{+})_n = \frac{1}{2} {\mathbb I}_n^{2\times 2} + \frac{1}{2} X_n, \quad n \neq 2i, 2j~,\\
		&(\ket{-}\bra{-})_n = \frac{1}{2} {\mathbb I}_n^{2\times 2} - \frac{1}{2} X_n, \quad n = 2i, 2j~.
	\end{aligned}
    \right.
\end{equation}
The above perturbation process can be mapped to an  intuitive ``particle hopping'' picture, where the particles are defect sites with $Z_{i-1} X_i Z_{i+1} = -1$.  The $\rho_0^0$ ($\rho^0_{i,j}$) has zero (two) defect sites with $Z_{i-1} X_i Z_{i+1} = -1$, and can be mapped to a classical state with zero (two) particles. More specifically, we use the notation $| \cdot )$ to denote a classical basis state [c.f.~\cref{seq:LtonHM.map}].
We map the original problem with $2N$ qubits to a chain of hard-core bosons with $N$ sites, such that the density matrix in the original problem of qubits is mapped to the classical states in the hard-core boson problem as follows:
\begin{equation}
    \begin{aligned}
        \rho_0^0 \mapsto |{\rm vac}) \equiv \bigotimes_{l = 1}^{N} |0)_l~, \quad \rho_{i,j}^0 \mapsto |i,j)  \equiv \bigg{[}\bigotimes_{l=1,l\neq i,j}^{N}|0)_l  \bigg{]}\bigotimes |1)_i \bigotimes |1)_j~.
    \end{aligned}
\end{equation}
Here $|0)_i$ ($|1)_i$) describes the classical state where no (one) boson sits on site $i$ of the hard-core boson chain. The density matrix $\rho_0^0$ defined in~\cref{eq:rhoc} is mapped to the classical state with no bosons on the chain, which we denote as $|{\rm vac})$. The density matrix $\rho_{i,j}^0$ in~\cref{eq:perturbation_ij} is mapped to a classical state $|{i,j})$ that describes the state with two bosons on sites $(i,j)$ on the chain of hard-core bosons. 

After mapping to the particle hopping picture, the perturbed dynamics can be captured by the Markovian transition matrix ${\mathcal P} = {\mathcal P}_0 + \lambda {\mathcal P}_{p}$. The unperturbed transition matrix ${\mathcal P}_0$ corresponds to ${\mathcal L}_0$ and consists of processes in which one of the two bosons in a classical state $|{i,j})$ hops rightwards.
For example, it can take $|{i,j})$ to $|{i,j + 1})$.
The perturbation to the transition matrix, $\lambda {\mathcal P}_{p}$, corresponds to ${\mathcal L}_{\rm pert}$.
The perturbation $\lambda {\mathcal P}_p$ leads to  two new effects in addition to the effects produced by ${\mathcal P}_0$. 
The first effect is that it can generate a pair of bosons at sites $i,j$ out of $|{\rm vac})$ [i.e.~it can take the classical state from $|{\rm vac})$ to the state $|i,j)$].
The second effect is that it can lead to the simultaneous hopping of two bosons sitting on sites $(i,j)$ to sites $(p,q)$, where $i \neq p$ and $j \neq q$. 
In other words, it takes one classical state $|{i,j})$ to another classical state $|{p,q})$.

We further denote the steady state of the unperturbed transition matrix, i.e, the right eigenvector of ${\mathcal P}_0$ with eigenvalue zero, as $|{R_0})$, so that ${\mathcal P}_0 |R_0) = 0$. 
Taking the first-order perturbed right eigenvector as $|{R}) = |{R_0}) + \lambda |{R_1})$, we have
\begin{equation}
	({\mathcal P}_0 + \lambda {\mathcal P}_p) (|{R_0}) + \lambda |{R_1})) = 0~.
\end{equation}
To leading order in $\lambda$, we have
\begin{equation}\label{Eq_R1}
	{\mathcal P}_0 |{R_1}) + {\mathcal P}_p |{R_0}) = 0\quad \implies \quad  |{R_1}) = - {\mathcal P}_0^{-1} {\mathcal P}_p |{R_0})~,
\end{equation}
where ${\mathcal P}_0^{-1}$ is the pseudoinverse of ${\mathcal P}_0$.

The steady state of the transition matrix ${\mathcal P}$ of the hard-core boson hopping problem can then be described by
\begin{equation}
	|{\vec c}) = (c_0, c_{1,2}, c_{1,3}, \ldots, c_{N-1,N})~,
\end{equation}
where $c_0$ is the coefficient of the vacuum state $|{\rm vac})$, $c_{1,2}$ is the coefficient of the state $|{1,2})$, $c_{1,3}$  is the coefficient of the state $|{1,3})$, etc... These coefficients are normalized to one according to
\begin{equation}\label{Eq_C_Normalization}
	\sum_{i=0}^{N(N-1)/2+1} c_i = c_0 + \sum_{i,j} c_{i,j}= 1~.
\end{equation}

After solving for the steady states of the perturbed particle hopping problem, we can determine the steady state of the original problem defined on $2N$ qubits via 
\begin{equation}
	\ket{\vec c} \mapsto \rho_{\vec c} \equiv c_0 \rho_0^0 + \sum_{i,j} c_{i,j} \rho_{i,j}^0~. 
\end{equation}
Here $\rho_{\vec c}$ captures the steady state in the presence of perturbation ${\mathcal L}_{\rm pert}$ on top of ${\rho}_0^0 = \tilde \rho_{\mathcal C}$. It is easy to check that ${\rm Tr}(\rho_{\vec c}) = 1$ under the normalization condition we considered.

One can then compute the string order parameters in the perturbed steady state up to first order in the perturbation. Consider the strong-string operator $\mathcal{S}_{nm}^S$, shown in~\cref{eq:string_op},  which acts on the original model of qubits as an example:
\begin{equation}
	\mathcal{S}_{nm}^S = Z_n \big{(}\bigotimes_{k \in {\rm even}} X_{k} \big{)}Z_m~,
\end{equation} 
with $n$ and $m$ odd, $n<k<m$, and $k$ even. Since we are working in the CZ dual picture, we consider the CZ dual of the string operator, $\tilde{\mathcal{S}}_{nm}^S$. $\tilde{\mathcal{S}}_{nm}^S$ is defined as
\begin{equation}\label{Eq_Xi}
	\begin{aligned}
		\tilde{\mathcal{S}}_{nm}^S & \equiv U^\dagger_{\rm CZ} \mathcal{S}_{nm}^S U_{\rm CZ} = U^\dagger_{\rm CZ} Z_n \bigg{(} \bigotimes_{\substack{n<k<m,\\k~{\rm even}} } X_k \bigg{)} Z_m U_{\rm CZ} \\
        &= Z_n \bigg{(} \bigotimes_{\substack{n<k<m,\\k~{\rm even}}} Z_{k-1} X_k Z_{k+1} \bigg{)} Z_m =  \bigotimes_{n<k<m,~k~{\rm even}} X_k,
	\end{aligned}
\end{equation}
where $U_{\rm CZ}$ is given in \cref{eq:U_CZ}.
It is easy to check that $\tilde{\mathcal{S}}_{nm}^S \rho_0^0 = \rho_0^0$. Similarly, we have $\tilde{\mathcal{S}}_{nm}^S \rho_{i,j}^0 = -\rho_{i,j}^0 $ for $n < 2i < m < 2j$ or $2i < n < 2j < m$, and $\tilde{\mathcal{S}}_{nm}^S \rho_{i,j}^0 = \rho_{i,j}^0$ otherwise. Given the above, we can compute the string order parameters $\tilde{C}_{\text{I},nm}^S = \langle \tilde{\mathcal{S}}_{nm}^S \rangle_\rho$ defined in Eq.~(\ref{eq:def_C_I_S}) as 
\begin{equation}\label{Eq_One_Copy}
	\begin{aligned}
		\tilde{C}_{\text{I},nm}^S = \langle \tilde{\mathcal{S}}_{nm}^S \rangle_\rho= \frac{{\rm Tr}(\tilde{\mathcal{S}}_{nm}^S \rho_{\vec c})}{{\rm Tr}\left[\rho_{\vec c}\right]} &= {\rm Tr}[\tilde{\mathcal{S}}_{nm}^S (c_0 \rho_0^0 + c_{1,2}\rho_{1,2}^0 + c_{1,3}\rho_{1,3}^0 + \cdots)]/{\rm Tr}\left[\rho_{\vec c}\right] = 1 - 2\sum_{(i,j) \in \zeta} c_{i,j}~,
	\end{aligned}
\end{equation}
where $\sum_{(i,j) \in \zeta}$ sums over all the cases that satisfy $n < 2i < m < 2j$ or $2i < n < 2j < m$. Note that we have used the normalization condition for $c_{0}$ and $c_{i,j}$ given in Eq.~(\ref{Eq_C_Normalization}). 
Thus, given a perturbation $\mathcal{L}_{\text{pert}}$, we can determine the perturbation to the transition matrix, $\mathcal{P}_p$. We can then calculate the coefficients $c_0$ and $ c_{i,j}$ efficiently numerically. Plugging these coefficients into~\cref{Eq_One_Copy} gives us the string correlator up to first order in perturbation theory.

\section{Stabilizer-state calculations}\label{app:stabilizer}

In this appendix, we briefly outline the algorithm to calculate quantities like $|\la \psi | \phi \ra |^2$ and $\la \psi|A | \phi \ra \la \phi| B |\psi\ra $ for stabilizer states $|\psi\ra$ and $|\phi\ra$ and Pauli operators $A$ and $B$.

Assume the stabilizer states $|\psi\ra$ and $|\phi\ra$ are stabilized by the stabilizers generated by $\la g_1, \cdots, g_n \ra$ and $\la h_1, \cdots, h_n \ra$, respectively. 
We then have $|\psi\ra \la \psi| = \prod_{i=1}^{n}\frac{1}{2}(I + g_i)$.
We therefore see that $|\la \psi | \phi \ra| ^2=\la \phi |  \prod_{i=1}^{n}\frac{1}{2}(I + g_i)|\phi \ra$, which is the probability of measuring $g_1,\dots g_n$ on $|\phi \ra$ with all $+1$ outcomes.
We therefore also see that, if $A$ is a Pauli operator, $A |\psi\ra \vcentcolon = |\tilde{\psi}\ra$ is also a stabilizer state. 
So $\la \psi|A | \phi \ra \la \phi|A|\psi\ra = |\la \tilde{\psi} | \phi \ra| ^2$ can be calculated accordingly.

We now describe how to compute $\bra{\psi}A \ket{\phi}\bra{\phi}B \ket{\psi}$.
Consider $\bra{\tilde\psi}\tilde A \ket{\phi}\braket{\phi \vert \tilde\psi}$,
where $\ket{\tilde\psi} = B \ket{\psi}$ is still a stabilizer state, and $\tilde A = BA$.
We must therefore compute $\bra{u} P \ket{v}\braket{v\vert u}$ for stabilizer states $u,v$ and a Pauli operator $P$.
Note that $\bra{u} P \ket{v}\braket{v\vert u}=\bra{u} \frac{1}{2}(1+P) \ket{v}\braket{v\vert u}-\bra{u} \frac{1}{2}(1-P) \ket{v}\braket{v\vert u}$. Again, assume the stabilizers of $\ket{v}$ are generated by $g_i$, then $\bra{u} \frac{1}{2}(1\pm P) \ket{v}\braket{v\vert u}$ is the probability of measuring $g_1 \cdots g_n , P$ with outcomes all $+1$ for $g_i$ but outcomes $\pm 1$ for $P$.
However, note that, for a stabilizer state $\ket{\psi}$, the probability difference between the measurement outcome $+1$ and $-1$ is exactly $\la \psi | P |\psi\ra$.

This yields the following procedure (see \cref{alg:stab-overlap}).
The idea is to force the state $\ket{u}$ into a $+1$ eigenstate of each stabilizer generator of $\ket{v}$, while keeping track of the probabilities of $\ket{u}$ being in such an eigenstate.
Then, we can simply compute the expectation value of $P$ in $\ket{u}$. To calculate $|\la u | v\ra|^2$, we can simply take $P=I$.

\begin{algorithm}[H]
\caption{Procedure to compute $\bra{u} P \ket{v}\braket{v\vert u}$}
\begin{algorithmic}[1]
\label{alg:stab-overlap}
\State $p \gets 1$
\For{$h$ a stabilizer generator of $v$}
    \State $m \gets \bra u h\ket u$
    \If{$m = - 1$}
        \State \Return 0
    \ElsIf{$m=0$}
        \State $p \gets p / 2$
        \State Postselect $u$ onto the $+1$ eigenstate of $h$
    \EndIf
\EndFor
\State \Return $p \cdot \bra u P \ket u$
\end{algorithmic}
\end{algorithm}

Note that in \cref{alg:stab-overlap}, there are three possible branches for $m$: $m=-1$ (the \texttt{if} branch), $m=0$ (the \texttt{else if} branch), and $m=1$ (the \texttt{else} branch).
The \texttt{else} branch is not included in the pseudocode for \cref{alg:stab-overlap} because it is redundant.
The \texttt{else} condition ($m = 1$) would proceed as $p \gets p$ and postselect $u$ onto the $+1$ eigenstate of $h$.
But if $m = 1$, then $u$ is already in the $+1$ eigenstate of $h$, thus making the \texttt{else} condition redundant.

% \twocolumngrid

\end{document}